\documentclass[10pt]{article}

\pdfoutput=1

\usepackage{hyperref}
\usepackage{cleveref}
\usepackage{geometry}
\usepackage{enumerate}%使用小标题（1）宏包
\usepackage{color}%字体颜色

\usepackage{array}%需要该宏包
\usepackage{amsmath}%使用公式数学命令宏包
\usepackage{amssymb}%使用数学符号宏包
\usepackage{amsfonts}
\usepackage{cite}%能在正文显示的该引用出处的时候显示[1-3], \cite{bibtex1，bibtex2，bibtex3} --> [1]-[3]
\usepackage{float}
\usepackage{supertabular}
\usepackage{threeparttable}
\usepackage[labelfont=footnotesize,textfont=footnotesize]{caption}
\usepackage{longtable}
\usepackage{multirow}
\usepackage{booktabs}
\usepackage{tabularx}

\usepackage{graphicx}
\usepackage{subfigure}
\usepackage{epstopdf}
\usepackage{lineno,hyperref}
\usepackage{natbib}
\modulolinenumbers[5]
\usepackage{ntheorem}%定义环境时需要的宏包
\allowdisplaybreaks[4] %公式可跨页的命令
\newtheorem{theorem}{Theorem}

\newtheorem{lemma}{Lemma}
\newtheorem*{problem}{Problem}
\newcommand{\tabincell}[2]{\begin{tabular}{@{}#1@{}}#2\end{tabular}}

\newtheorem{definition}{Definition}

\newtheorem{corollary}{Corollary}

\newenvironment{proof}{{\bf Proof: \ }}{\hfill$\bull$\medskip}
\newtheorem{remark}{Remark}

\newcommand{\bull}{\vrule height 1.8ex width 1.0ex depth 0ex}

\renewcommand{\theequation}{\arabic{section}.\arabic{equation}}
\begin{document}
\title{A hybrid stochastic differential reinsurance and investment game with bounded memory
\footnote{ This research is supported by the National Natural Science Foundation of China (Nos. 71771082, 71801091) and Hunan Provincial Natural Science Foundation of China (No. 2017JJ1012).}
}
\author{\normalsize Yanfei Bai$^1$,\quad Zhongbao Zhou$^{1,}$\footnote{Corresponding author. E-mail: z.b.zhou@163.com; z.b.zhou@hnu.edu.cn.},\quad Helu Xiao$^2$,\quad Rui Gao$^3$,\quad Feimin Zhong$^1$\\
\small $^1$School of Business Administration, Hunan University, Changsha 410082, China\\
\small $^2$School of Business, Hunan Normal University, Changsha 410081, China\\
\small $^3$ College of Mathematics and Econometrics, Hunan University, Changsha 410082, China
}
\date{}
\maketitle
\begin{abstract}
This paper investigates a hybrid stochastic differential reinsurance and investment game  between one reinsurer and two insurers, including a stochastic Stackelberg differential subgame and a non-zero-sum stochastic differential subgame.
The reinsurer, as the leader of the Stackelberg game, can price reinsurance premium and invest its wealth in a financial market that contains a risk-free asset and a risky asset.
The two insurers, as the followers of the Stackelberg game, can purchase proportional reinsurance from the reinsurer and invest in the same financial market.
The competitive relationship between two insurers is modeled by the non-zero-sum game, and their decision making will consider the relative performance measured by the difference in their terminal wealth.
We consider wealth processes with delay to characterize the bounded memory feature.
This paper aims to find the equilibrium strategy for the reinsurer and insurers by maximizing the expected utility of the reinsurer's terminal wealth with delay and maximizing the expected utility of the combination of insurers' terminal wealth and the relative performance with delay.
By using the idea of backward induction and the dynamic programming approach,
we derive the equilibrium strategy and value functions explicitly.
Then, we provide the corresponding verification theorem.
Finally, some numerical examples and sensitivity analysis are presented to demonstrate the effects of model parameters on the equilibrium strategy.
We find the delay factor discourages or stimulates investment depending on the length of delay.
Moreover, competitive factors between two insurers make their optimal reinsurance-investment strategy interact, and reduce reinsurance demand and reinsurance premium price.
\end{abstract}

\section{Introduction}\label{section 1}

Insurers and reinsurers, as special financial institutions, not only have to face the investment risks in the financial market, but also have to manage the risk of random claims in the insurance market.
Insurers can sign reinsurance contracts from the reinsurer and transfer part of the risk of claims to the reinsurer, because the reinsurer is more risk-seeking than insurers.
Research on optimal reinsurance and investment strategies has been an important part of mainstream study in the actuarial field.
In recent decades, many scholars have made extensive studies on reinsurance and investment optimization problem under different objectives, for example, minimizing the probability of ruin (\cite{Browne-1995-20}, \cite{Chen-2010-47}, \cite{Li-2015-259}, etc.), maximizing the expected utility of the terminal wealth (\cite{Liang-2011-32}, \cite{Li-2012-51}, \cite{HUANG2016443}, \cite{Zhao-2017-28}, etc.), maximizing the expected terminal surplus as well as minimizing the variance of the terminal surplus ( \cite{Bi-2014-212},  \cite{Zhou-Zhang-2019-28}, \cite{ZHOU2019}, etc.).%%%%%%%%%%%%到最后看投哪个期刊，再加相关期刊的文献。

However, the majority of these researches study the reinsurance and investment optimization problem only from the unilateral perspective of insurers, while the interest of the reinsurer is generally ignored.
Since the setting of the reinsurance contract depends on the mutual agreement between the insurer and the reinsurer, a reinsurance contract that only considers the interest of one party may be unacceptable to the other party.
That is, the reinsurance contract should be designed to take into account the interests of both the insurer and the reinsurer.
In view of the monopoly position of the reinsurer and the competitive relationship between insurers in the market, we consider a reinsurer and two insurers as the leader and the followers of a stochastic Stackelberg differential game respectively, and use the non-zero-sum stochastic differential game to describe the competitive relationship.
We investigate the reinsurer's premium pricing and investment optimization problem as well as insurers' reinsurance and investment optimization problem.
That is, we consider the mutual interests of the reinsurer and two insurers as well as the competition between insurers.

As far as we know, the game problem in the insurance market has attracted some scholars' attention.
%\cite{Zeng2010-47}, \cite{TaksarZeng2011Optimal} and \cite{LiRongZhao2015} studied zero-sum stochastic differential game problems between two insurers under different control constraints.
With respect to maximizing the expected utility of the relative performance, \cite{Bensoussan-2014-50} studied a non-zero-sum stochastic differential investment and reinsurance game between two insurers whose surplus processes were modulated by continuous-time Markov chains;
\cite{Deng-2018-264} investigated the implications of strategic interaction between two constant absolute risk aversion (CARA) insurers on their reinsurance-investment policies with default risk under the framework of non-zero-sum stochastic differential game.
There are still many studies on the non-zero-sum stochastic differential reinsurance-investment game problem, such as \cite{Meng-2015-62}, \cite{Pun-2016-68}, \cite{Guan-2016-70}, \cite{Yan-2017-75}, \cite{ZHUCao2018}, etc.

Obviously, the above mentioned stochastic differential game models about reinsurance-investment problem do not consider the interest of the reinsurer.
\cite{Chen2018} first proposed a stochastic Stackelberg differential reinsurance game model to  depict the leader-follower relationship between the reinsurer and insurer in the insurance market, and analyzed optimal reinsurance strategy from joint interests of the insurer and the reinsurer.
\cite{CHEN2019120} studied stochastic Stackelberg differential reinsurance games under
time-inconsistent mean-variance framework.
%In addition, the theory of stochastic Stackelberg differential games is considered by \cite{Yong-2002-41}, \cite{Bensoussan-2015-53} and \cite{SHI201660}.
Inspired by these insights, in this paper, we will build a more realistic stochastic differential reinsurance and investment game model, i.e., the hybrid stochastic differential reinsurance and investment game model, which takes into account the leader-follower relationship between the reinsurer and the insurers and the competitive relationship between the insurers.

Traditionally, most of the researches on optimal reinsurance-investment decision making are based on current information, ignoring the performance of past wealth.
However, decisions often rely on the past information in real systems, and delays arise naturally.
This feature is commonly referred to as the delay feature or bounded memory feature (see \cite{Chang-2011-36} and \cite{Federico-2011-15}).
\cite{Shen-2014-57} first introduced the bounded memory feature of wealth into the optimal investment-reinsurance problem for mean-variance insurers and obtained the optimal strategy under certain conditions.
Since then, \cite{A-2015-61} considered an optimal investment and excess-of-loss reinsurance problem with delay for an insurer under Heston's stochastic volatility model.
\cite{Yang-2017-47} researched an optimal proportional reinsurance problem for the compound Poisson risk model with delay under the mean-variance criterion.
\cite{Chunxiang-2018-342} studied an optimal investment and excess-of-loss reinsurance problem with delay and jump-diffusion risk process.
The delay feature of the wealth process has influence on decision making of the reinsurer and insurers.
It would be more practical to consider such a delay period.
Therefore, this paper also considers the delay feature of the wealth processes under the framework of the hybrid stochastic differential game.

The main work of this article is summarized as follows.
We build a hybrid stochastic differential reinsurance and investment game model, including stochastic Stackelberg differential subgame and non-zero-sum stochastic differential subgame.
One reinsurer and two insurers are three players in the hybrid game and are the leader and the followers in the Stackelberg game respectively.
The competitive relationship between two insurers is modeled by the non-zero-sum game, and their decision making will consider the relative performance measured by the difference in their terminal wealth.
Furthermore, the effects of delay on wealth processes are considered.
By using the idea of backward induction and the dynamic programming approach, we derive the equilibrium strategy and value functions explicitly.
Then, we establish the corresponding verification theorem for the optimality of the given strategy.
The equilibrium strategy indicates that the optimal reinsurance-investment strategies of two insurers interact with each other and reflect the herd effect.
Moreover, the optimal reinsurance strategies of insurers depend on the reinsurer's optimal premium strategy.
We study several special cases of the model and find that the delay factor discourages or stimulates investment depending on the length of the delay, and the optimal reinsurance premium in the intermediate case follows the principle of variance premium when there is only one insurer and one reinsurer in the insurance market.
Finally, we present some numerical examples and sensitivity analysis to demonstrate the effects of the model parameters on the equilibrium strategy.
Through the analysis and numerical simulation of equilibrium strategy, we find that
competitive factors between two insurers reduce the demand for reinsurance and the price of reinsurance premium.
In addition, we find that the effect of delay weight on the equilibrium strategy is related to the length of delay.

Different from the existing literature, our work has the following four contributions.
(1) We first construct a hybrid stochastic differential reinsurance and investment game model including a stochastic Stackelberg differential subgame and a non-zero-sum stochastic differential subgame.
That is, we first consider the leader-follower relationship between the reinsurer and insurers and the competitive relationship between insurers at the same time, which is closer to the actual situation.
(2) We consider the tripartite game between one reinsurer and two insurers, and the joint interests of the reinsurer and insurers are considered, while the majority of the existing researches on non-zero-sum reinsurance and investment game only focus on the insurer's interest, ignoring the reinsurer's interest generally.
(3) Since reinsurance and investment are important risk management tools, we study the stochastic differential reinsurance and investment game problem, while \cite{Chen2018} and \cite{CHEN2019120} only studied the reinsurance optimization problem under the framework of the Stackelberg game.
(4) We consider the effect of the bounded memory feature of wealth processes under the framework of the hybrid stochastic differential game, which has rarely been studied in the past.

The remainder of this paper is organized as follows.
Section \ref{section 2} formulates the hybrid stochastic differential reinsurance and investment game between one reinsurer and two insurers with delay.
In Section \ref{section 3}, we derive the equilibrium strategy and value functions for the hybrid game problem by  using the idea of backward induction and the dynamic programming approach.
Then, we establish a verification theorem for the optimality of the equilibrium strategy and study some special cases of our model.
Section \ref{section 4} provides some numerical examples and sensitivity analysis to demonstrate the effects of the model parameters on the equilibrium strategy.
Section \ref{section 5} concludes the paper.

\section{Model setup}\label{section 2}

In this section, we describe the model in details. Let $[0,T]$ be a continuous-time finite
horizon, over which reinsurance and investment behavior can occur. The uncertainty in
markets is represented by a complete probability space $\left(\Omega,\mathcal{F},P\right)$, which is equipped with a filtration $\mathcal{F}=\left\{\mathcal{F}_t\right\}_{0\leq t\leq T}$ satisfying the usual conditions.

\subsection{Dynamics of the surplus processes}\label{section 2.1}

We consider an insurance market containing two competing insurers and one reinsurer.
The surplus process of insurer $i\in\{1,2\}$, denoted by $\{X_i(t)\}_{t\geq0}$, is depicted by the classic Cram\'{e}r-Lundberg risk model:
\begin{align}\label{equ:2.1}
X_i(t)=x_i^0+c_it-\sum_{n=1}^{N_i(t)+N(t)}\tilde{Y}_i^n, \quad i\in\{1,2\},
\end{align}
where $x_i^0>0$ and $c_i\geq0$ are the initial surplus and the premium rate of insurer $i$, respectively;
$N_i(t)+N(t)$ represents the number of claims up to time $t$;
$\{\tilde{Y}_i^n\geq0,n\geq1\}$ is a list of independent identically distributed (i.i.d.) random variables with distribution function $F_i(\tilde{Y})$;
$\tilde{Y}_i^n$ represents the amount of the $n$-th claim of insurer $i$.
We assume that
\begin{enumerate}[(i)]
  \item $\{\tilde{Y}_i^n\geq0,n\geq1\}$ has finite first moment $\mu_i$ ($0<\mu_i<+\infty$) and finite secondary moment $(\tilde{\sigma}_i)^2<+\infty$;
  \item $\{\tilde{Y}_i^n\geq0,n\geq1\}$ is independent of $N_i(t)$ and $N(t)$;
  \item $N_{1}(t)$, $N_{2}(t)$ and $N(t)$ are three mutually independent Poisson processes with intensity $\tilde{\lambda}_1> 0$, $\tilde{\lambda}_2> 0$ and $\tilde{\lambda}> 0$, respectively.
\end{enumerate}

Then, the surplus process $\{X_i(t)\}_{t\geq0}$ indicates that insurer 1 and insurer 2 are subject to common impact that is represented by $\{N(t)\}_{t\geq0}$.
Refer to  \cite{Grandell-1977-1977}, \cite{Browne-1995-20}, \cite{Gerber-2006-186},
\cite{Bai-2008-42}, \cite{Chen-2018-48} etc., the classic Cram\'{e}r-Lundberg model \eqref{equ:2.1} can be approximated by the following diffusion process:
\begin{align}\label{equ:2.2}
dX_i(t)=c_idt-\lambda_i\mu_idt+\sqrt{\lambda_i(\tilde{\sigma}_i)^2}dW_i(t),\quad X_i(0)=x_i^0, \quad i\in\{1,2\},
\end{align}
where $\lambda_i=\tilde{\lambda}_i+\tilde{\lambda}$; $\{W_i(t),t\geq0\}$ is a standard $\mathcal{F}$-Brownian motion. $d\langle W_{1}(t),W_{2}(t) \rangle$ $=\rho dt$, where $\rho
=\frac{\tilde{\lambda} \mu_{1}\mu_{2}}{\sqrt{\lambda_{1}\lambda_{2}}\tilde{\sigma}_{1}\tilde{\sigma}_{2}}$.

According to expected value premium principle, we know that $c_i=(1+\theta_i)\lambda_i\mu_i$, where $\theta_i>0$ is the safety loading of insurer $i$.
Both two insurers can manage its claim risk through purchasing proportional reinsurance continuously from the reinsurer.
The reinsurance strategy of insurer $i$ is characterized by $\{q_i(t),t\geq0\}$ satisfying $q_i(t)\in[0,1]$. Then the reinsurer will cover $(1-q_i(t))100\%$ of the claims of insurer $i$ while insurer $i$ will cover remaining at time $t$.
The price of the reinsurance premium at time $t$ is $p(t)\in[c_F,\bar{c}]$, where $c_F=\max\{c_{1},c_{2}\}$, $\bar{c}=(1+\bar{\theta})\lambda_F\mu_F$, $\bar{\theta}$ is an upper bound of the reinsurer's relative safety loading, $\bar{\theta}>\max\{\theta_1,\theta_2\}$ and $\lambda_F\mu_F=\max\{\lambda_1\mu_1,\lambda_2\mu_2\}$.  Introducing
proportional reinsurance strategy $q_i(t)$ into equation \eqref{equ:2.2}, then
\begin{align}
dX_i(t)
=[\theta_ia_i-(p(t)-a_i)(1-q_i(t))]dt+q_i(t)\sigma_idW_i(t), \quad i\in\{1,2\},
\end{align}
where $a_i=\lambda_i\mu_i$, $\sigma_i=\sqrt{\lambda_i(\tilde{\sigma}_i)^2}$. The surplus process of the reinsurer is as following:
\begin{align}
dX_L(t)
=&[(p(t)-a_1)(1-q_1(t))+(p(t)-a_2)(1-q_2(t))]dt\nonumber\\
&+(1-q_1(t))\sigma_{1}dW_{1}(t)
+(1-q_2(t))\sigma_{2}dW_{2}(t),\quad X_L(0)=x_L^0.
\end{align}

\subsection{Dynamics of the financial assets}\label{section 2.2}

Assuming that both the reinsurer and two insurers can invest in a financial market that contains one risk-free asset and one risky asset. The price process of the risk-free asset, $\{S_0(t)\}_{t\geq0}$, is given by the following ordinary differential equation (ODE):
\begin{align}
dS_0(t)=r_0S_0(t)dt,\quad S_0(0)=1,
\end{align}
where $r_0$ is the constant risk-free interest rate. The price process of the risky asset, $\{S(t)\}_{t\geq0}$, is described by the constant elasticity of variance (CEV) model:
\begin{align}\label{equ:St}
dS(t)=S(t)\left[rdt+\sigma S^\beta(t) dW(t)\right],\quad S(0)=s_0,
\end{align}
where $r$, $\sigma S^\beta(t)$ and $\beta$ denote the expected return rate, the volatility and the constant elasticity parameter of the risky asset, respectively;
$r>r_0>0$, $\sigma>0$; $\{W(t),t\geq0\}$ is a standard $\mathcal{F}$-Brownian motion and is independent of $\{W_1(t),t\geq0\}$ and $\{W_2(t),t\geq0\}$.
The CEV model can reduce to a geometric Brownian motion (GBM) when $\beta=0$.
If $\beta<0$, the volatility $\sigma S^\beta(t)$ increases as the stock price decreases, and a distribution with a fatter left tail can be generated.
If $\beta>0$, the volatility $\sigma S^\beta(t)$ increases as the stock price increases.

\subsection{Wealth processes with delay}\label{section 2.3}

Suppose that there are no transaction costs or taxes for investment and reinsurance, and short-selling of the risky asset is allowed. Let $\{b_L(t),t\geq0\}$, $\{b_{1}(t),t\geq0\}$ and $\{b_{2}(t),t\geq0\}$ be measurable processes valued in $\mathbb{R}$ representing the amount invested in the risky asset by the reinsurer, insurer $1$ and insurer $2$ at time $t$, respectively.
Then, the remaining wealth $X_L(t)-b_L(t)$, $X_{1}(t)-b_{1}(t)$ and $X_{2}(t)-b_{2}(t)$ are invested in the risk-free asset.
Let $\pi_L(t)=(p(t),b_L(t))$, $\pi_{1}(t)=(q_1(t),b_{1}(t))$ and $\pi_{2}(t)=(q_2(t),b_{2}(t))$.
Then, with  strategy $\{\pi_L(t),t\geq 0\}$, the wealth
process of the reinsurer, denoted by  $\{X_L^{\pi_L}(t)\}_{t\geq0}$,  can be expressed as:
\begin{align}
dX_L^{\pi_L}(t)
=&[(p(t)-a_1)(1-q_1(t))+(p(t)-a_2)(1-q_2(t))+r_0X_L^{\pi_L}(t)+(r-r_0)b_L(t)]dt\nonumber\\
&+(1-q_1(t))\sigma_{1}dW_{1}(t)
+(1-q_2(t))\sigma_{2}dW_{2}(t)+b_L(t)\sigma S^\beta(t) dW(t).
\end{align}
With strategy $\{\pi_i(t),t\geq0\}$, the wealth process of
insurer $i$, denoted by $\{X_i^{\pi_i}(t)\}_{t\geq0}$, can be expressed as:
\begin{align}
dX_i^{\pi_i}(t)
=&[\theta_ia_i-(p(t)-a_i)(1-q_i(t))+r_0X_i^{\pi_i}(t)+(r-r_0)b_i(t)]dt\nonumber\\
&+q_i(t)\sigma_idW_i(t)+b_i(t)\sigma S^\beta(t) dW(t), \quad i\in\{1,2\}.
\end{align}

In fact, due to the bounded memory feature, the reinsurer's and insurers' strategies depend on the exogenous capital instantaneous inflow into or outflow from current wealth.
Refer to \cite{A-2015-61}, let $Y_L(t)$ and $Z_L(t)$ be the integrated and pointwise delayed information of the reinsurer's wealth process in the past horizon $[t-h_L,t]$, respectively. Correspondingly, let $Y_i(t)$ and $Z_i(t)$ be the integrated and pointwise delayed information of insurer $i$'s wealth process in the past horizon $[t-h_i,t]$, respectively. That is, for $\forall t\in[0,T]$,
\begin{align}
Y_L(t)&=\int_{-h_L}^0e^{\alpha_L s}X_L^{\pi_L}(t+s)ds,\quad Z_L(t)=X_L^{\pi_L}(t-h_L),\\
Y_i(t)&=\int_{-h_i}^0e^{\alpha_i s}X_i^{\pi_i}(t+s)ds,\quad Z_i(t)=X_i^{\pi_i}(t-h_i), \quad i\in\{1,2\},
\end{align}
where $\alpha_L>0$ and $\alpha_i>0$ are the average parameters; $h_L>0$ and $h_i>0$ are
the delay time parameters.
Let $f_L(t,X_L(t)-Y_L(t),X_L(t)-Z_L(t))$ and $f_i(t,X_i(t)-Y_i(t),X_i(t)-Z_i(t))$ represent the capital inflow/outflow amount of the reinsurer and insurer $i$, respectively; where $X_L(t)-Y_L(t)$, $X_i(t)-Y_i(t)$ represent the
average performance and  $X_L(t)-Z_L(t)$, $X_i(t)-Z_i(t)$ represent the absolute performance.
Such capital inflow/outflow, which is related to the past performance of the wealth, may come
out in various situations.
For example, a good past performance may bring the company more gain and further the company can pay a part of the gain as dividend to its shareholders. Contrarily, a poor past performance forces the company to seek further capital injection for covering the loss so that the final performance objective is still achievable.
To make the problem solvable, we assume
\begin{align}
f_L(t,X_L(t)-Y_L(t),X_L(t)-Z_L(t))&=B_L(X_L(t)-Y_L(t))+C_L(X_L(t)-Z_L(t)),\\
f_i(t,X_i(t)-Y_i(t),X_i(t)-Z_i(t))&=B_i(X_i(t)-Y_i(t))+C_i(X_i(t)-Z_i(t)), \quad i\in\{1,2\},
\end{align}
where $B_L$, $C_L$, $B_i$ and $C_i$ are nonnegative constants. In other words, the amount of the capital inflow/outflow is the linear weighted sum of the average performance and the absolute performance. Then, considering capital inflow/outflow functions $f_L(t,X_L(t)-Y_L(t),X_L(t)-Z_L(t))$ and $f_i(t,X_i(t)-Y_i(t),X_i(t)-Z_i(t))$, the wealth processes of the reinsurer and insurer $i$ are governed by the following stochastic differential delay equations (SDDEs), respectively:
\begin{align}
dX_L^{\pi_L}(t)
=&\big[(p(t)-a_1)(1-q_1(t))+(p(t)-a_2)(1-q_2(t))+A_LX_L^{\pi_L}(t)+(r-r_0)b_L(t)+B_LY_L(t)+C_LZ_L(t)\big]dt\nonumber\\
&+(1-q_1(t))\sigma_{1}dW_{1}(t)+(1-q_2(t))\sigma_{2}dW_{2}(t)
+b_L(t)\sigma S^\beta(t) dW(t),\label{equ:XL}\\
dX_i^{\pi_i}(t)
=&\big[\theta_ia_i-(p(t)-a_i)(1-q_i(t))
+A_iX_i^{\pi_i}(t)+B_iY_i(t)+C_iZ_i(t)+(r-r_0)b_i(t)\big]dt\nonumber\\
&+q_i(t)\sigma_idW_i(t)+b_i(t)\sigma S^\beta(t) dW(t), \quad i\in\{1,2\},\label{equ:XF}
\end{align}
where $A_L=r_0-B_L-C_L$, $A_i=r_0-B_i-C_i$.
In addition, we assume that insurer $i$, $i\in\{1,2\}$, is endowed with the initial
wealth $x_i^0$ at time $-h_i$ and does not start the business (insurance/reinsurance/investment) until time $0$, i.e., $X_i(t)=x_i^0>0,\forall t\in[-h_i,0]$.
Correspondingly, suppose that $X_L(t)=x_L^0>0,\forall t\in[-h_L,0]$.
Then, $Y_L(0)=\frac{x_L^0}{\alpha_L}(1-e^{-\alpha_Lh_L})$ and $Y_i(0)=\frac{x_i^0}{\alpha_i}(1-e^{-\alpha_ih_i})$.
For any fixed $t\in[0,T]$, denote $X_L^{\pi_{L}}(t)=x_L$, $Y_L(t)= y_L$, $Z_L(t)=z_L$, $X_i^{\pi_i}(t)=x_i$, $Y_i(t)=y_i$, $Z_i(t)=z_i$ and $S(t)=s$. Then, we define the admissible strategy as follows.

\begin{definition}\label{def1}
(Admissible strategy) $\pi(\cdot)=\pi_L(\cdot)\times\pi_{1}(\cdot)\times\pi_{2}(\cdot)
=(p(\cdot),b_L(\cdot))\times(q_1(\cdot),b_{1}(\cdot))\times(q_2(\cdot),b_{2}(\cdot))$ is said to be admissible, if
\begin{enumerate}[(i)]
\item $\{\pi_L(t)\}_{t\in[0,T]}$, $\{\pi_{1}(t)\}_{t\in[0,T]}$ and $\{\pi_{2}(t)\}_{t\in[0,T]}$ are $\mathcal{F}$-progressively measurable processes, such that $p(t)\in [c_F,\bar{c}]$, $q_1(t)\in[0,1]$ and $q_2(t)\in[0,1]$ for any $t\in[0,T]$;
\item $E\left[\int_t^T[(b_L(\ell))^2+(p(\ell))^2]d\ell\right]<+\infty$ and
    $E\left[\int_t^T[(b_i(\ell))^2+(q_i(\ell))^2]d\ell\right]<+\infty$,
    $\forall\ell\in[t,T]$, $i\in\{1,2\}$;
\item the equation \eqref{equ:XL} associated with $\pi(\cdot)$ has a unique solution $X_L^{\pi_L}(\cdot)$, which satisfies $\{E_{t,x_L,y_L,s}\left[\sup|X_L^{\pi_L}(\ell)|^2\right]\}^{\frac{1}{2}}\newline<+\infty$, for $\forall(t,x_L,y_L,s)\in[0,T]\times\mathbb{R}\times\mathbb{R}\times\mathbb{R}$, $\forall\ell\in[t,T]$;
\item the equation \eqref{equ:XF} associated with $\pi(\cdot)$ has a unique solution $X_i^{\pi_i}(\cdot)$, which satisfies $\{E_{t,x_i,y_i,s}\left[\sup|X_i^{\pi_i}(\ell)|^2\right]\}^{\frac{1}{2}}\newline<+\infty$, for $\forall(t,x_i,y_i,s)\in[0,T]
    \times\mathbb{R}\times\mathbb{R}\times\mathbb{R}$,  $\forall\ell\in[t,T]$, $i\in\{1,2\}$.
\end{enumerate}
\end{definition}

Let $\Pi=\Pi_L\times\Pi_{1}\times\Pi_{2}$ be the set of all admissible strategies, where  $\Pi_L$, $\Pi_1$ and $\Pi_2$ denote the set of all admissible strategies of the reinsurer, insurer $1$ and insurer $2$, respectively.

\subsection{The hybrid game problem}\label{section 2.4}

In view of the monopoly position of the reinsurer and the competitive relationship between insurers in the market, we investigate a hybrid stochastic differential reinsurance and investment game with delay.
The reinsurer and the two insurers are three players in the hybrid game.
The game between the reinsurer and two insurers is a stochastic Stackelberg differential game, in which the reinsurer is the leader and two insurers are the followers.
The game between the two insurers is a non-zero-sum stochastic differential game, and their status is equal.
More intuitively, the relationships between these three companies are shown in Figure \ref{fig:relation}.
\footnote{The non-zero-sum game between two insurers can also be regarded as the subgame of the Stackelberg game.}
\begin{figure}[htp]
\begin{center}
  \includegraphics[width=5in]{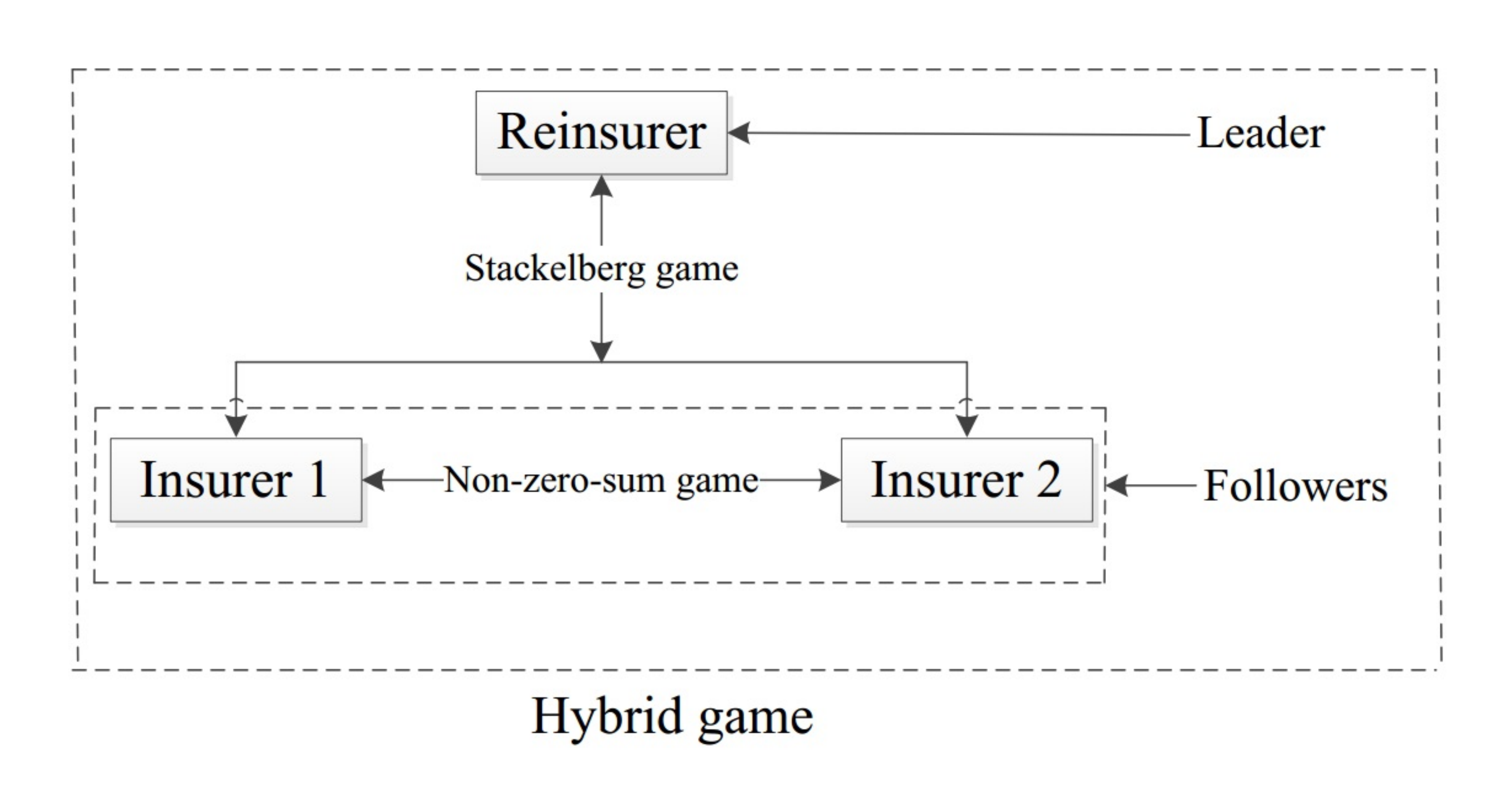}\\
  \caption{Relationships between three companies.}
  \label{fig:relation}
  \end{center}
\end{figure}

The goal of the hybrid game is to seek the equilibrium by solving the optimization problems of three parties.
Refer to \cite{Chen2018}, \cite{CHEN2019120} and \cite{Asmussen2019-7-2}, the procedure of solving the Stackelberg game is to solve the leader's and followers' optimization problems sequentially, based on the idea of backward induction.
To be more specific, the procedure can be divided into the following three steps:
\begin{itemize}
\item Step 1: The leader (i.e., the reinsurer) moves first by announcing its any admissible strategy $(p(\cdot),b_L(\cdot))\in\Pi_L$;
\item Step 2: The followers (i.e., the two insurers) observe the reinsurer's strategy and obtain their optimal strategies $q_1^{\ast}(\cdot)=\alpha_1^{\ast}(\cdot,p(\cdot),b_L(\cdot))$,
    $b_{1}^{\ast}(\cdot)=\beta_1^{\ast}(\cdot,p(\cdot),b_L(\cdot))$, $q_2^{\ast}(\cdot)=\alpha_2^{\ast}(\cdot,p(\cdot),b_L(\cdot))$,  $b_{2}^{\ast}(\cdot)=\beta_2^{\ast}(\cdot,p(\cdot),b_L(\cdot))$  by solving their optimization problems;
\item Step 3: Knowing that the two insurers would execute $\alpha_1^{\ast}(\cdot,p(\cdot),b_L(\cdot))$, $\beta_1^{\ast}(\cdot,p(\cdot),b_L(\cdot))$, $\alpha_2^{\ast}(\cdot,p(\cdot),b_L(\cdot))$ and $\beta_2^{\ast}(\cdot,p(\cdot),b_L(\cdot))$, the reinsurer then decides on its optimal strategy $(p^{\ast}(\cdot),b^{\ast}_L(\cdot))$ by solving its own optimization problem.
\end{itemize}

Due to the reinsurer's bounded memory feature, we suppose that the reinsurer is concerned about not only the terminal wealth $X_L^{\pi_L}(T)$, but also the integrated delayed information over the period $[T-h_L,T]$, i.e., $Y_L(T)$.
In other words, the objective of the reinsurer is to find the premium strategy and investment strategy to maximize the expected utility of $X_L^{\pi_L}(T)+\eta_LY_L(T)$, where the constant $\eta_L\in(0,1)$ represents the sensitivity of the reinsurer to past wealth.

Due to the fierce competition in the insurance market, insurer $i$ ($i\in\{1,2\}$) should consider not only the bounded memory feature of its own wealth, but also the wealth gap between itself and insurer $j$ ($j\neq i\in\{1,2\}$) at the terminal time $T$.
In other words, the objective of insurer $i$ is to find the optimal reinsurance strategy and investment strategy such that the expected utility of the combination of its terminal wealth and the relative performance with delay is maximized.
That is, insurer $i$ will choose a reinsurance-investment strategy
$\pi_{i}(\cdot)=(q_i(\cdot),b_{i}(\cdot))\in \Pi_{i}$ such that
\begin{align}
&E\Big[U_i\big((1-k_i)(X_i^{\pi_i}(T)+\eta_iY_i(T))+k_i((X_i^{\pi_i}(T)+\eta_iY_i(T))-(X_j^{\pi_j}(T)+\eta_jY_j(T)))\big)\Big]\nonumber\\
=&E\Big[U_i\big((X_i^{\pi_i}(T)+\eta_iY_i(T))-k_i(X_j^{\pi_j}(T)+\eta_jY_j(T))\big)\Big],\quad i\neq j\in\{1,2\},
\end{align}
is maximized.
Here, $U_i~(i\in\{1,2\})$ is the utility function of insurer $i$,
$\eta_i\in(0,1)$ values the weight of $Y_i(T)$, $k_i\in[0,1]$ measures the sensitivity of insurer $i$ to the performance of insurer $j$ ($j\neq i\in\{1,2\}$).

Refer to \cite{Bensoussan-2014-50}, \cite{Yan-2017-75} and \cite{Deng-2018-264},
the non-zero-sum game problem is to find an equilibrium reinsurance-investment strategy $(\pi_1^\ast,\pi_2^\ast)\in\Pi_1\times\Pi_2$ such that for any $(\pi_1,\pi_2)\in\Pi_1\times\Pi_2$, the following inequations are simultaneously established.
\begin{align}
&E\left[U_1\big((X_1^{\pi_1}(T)+\eta_1Y_1(T))-k_1(X_2^{\pi_2^\ast}(T)+\eta_2Y_2(T))\big)\right]\leq E\left[U_1\big((X_1^{\pi_1^\ast}(T)+\eta_1Y_1(T))-k_1(X_2^{\pi_2^\ast}(T)+\eta_2Y_2(T))\big)\right],\\
&E\left[U_2\big((X_2^{\pi_2}(T)+\eta_2Y_2(T))-k_2(X_1^{\pi_1^\ast}(T)+\eta_1Y_1(T))\big)\right]\leq E\left[U_2\big((X_2^{\pi_2^\ast}(T)+\eta_2Y_2(T))-k_2(X_1^{\pi_1^\ast}(T)+\eta_1Y_1(T))\big)\right].
\end{align}

The way of solving the non-zero-sum game is to solve the optimization problems of both two insurers at the same time.
That is, in Step 2, we have to solve the optimization problems for both two insurers simultaneously.
For convenience, we denote $\hat{X}_i^{\pi_i}(t)=X_i^{\pi_i}(t)-k_iX_j^{\pi_j}(t)$, for $i\neq j\in\{1,2\}$. Then, from \eqref{equ:XF}, we have
\begin{align}\label{equ:X^F}
d\hat{X}_i^{\pi_i}(t)
=&\Big[\theta_ia_i-k_i\theta_ja_j-(p(t)-a_i)(1-q_i(t))+k_i(p(t)-a_j)(1-q_j(t))
+A_iX_i^{\pi_i}(t)
-k_iA_jX_j^{\pi_j}(t)+B_iY_i(t)\nonumber\\
&-k_iB_jY_j(t)
+C_iZ_i(t)-k_iC_jZ_j(t)+(r-r_0)(b_i(t)-k_ib_j(t))\Big]dt
+q_i(t)\sigma_idW_i(t)-k_iq_j(t)\sigma_jdW_j(t)\nonumber\\
&
+(b_i(t)-k_ib_j(t))\sigma S^\beta(t) dW(t), \quad i\neq j\in\{1,2\},
\end{align}
with $\hat{X}_i^{\pi_i}(0)=X_i^{\pi_i}(0)-k_iX_j^{\pi_j}(0)
=x_i^0-k_ix_j^0\doteq\hat{x}_i^0$.
For any fixed $t\in[0,T]$,
let $\hat{X}_i^{\pi_i}(t)=X_i^{\pi_i}(t)-k_iX_j^{\pi_j}(t)
=x_i-k_ix_j\doteq\hat{x}_i$.

Then, the hybrid stochastic differential reinsurance and investment game problem can be described as the following problem.
\begin{problem} \label{problem}
The problem of insurer $i$ ($i\in\{1,2\}$) is the following optimization problem: for any $\pi_L(\cdot)=(p(\cdot),b_L(\cdot))\in\Pi_L$, find a map $(q_i^{\ast}(\cdot),b_i^{\ast}(\cdot))
=(\alpha_i^{\ast}(\cdot,p(\cdot),b_L(\cdot)),
\beta_i^{\ast}(\cdot,p(\cdot),b_L(\cdot)))
:[0,T]\times\Omega\times\Pi_L\rightarrow\Pi_i$ such that the following value function holds.
\begin{align}\label{equ:218}
&V^{F_i}(t,\hat{x}_i,y_i,y_j,s;p(\cdot),b_L(\cdot),
\alpha_i^{\ast}(\cdot,p(\cdot),b_L(\cdot)),\beta_i^{\ast}(\cdot,p(\cdot),b_L(\cdot)))\nonumber\\
&=\sup_{(q_i(\cdot),b_i(\cdot))\in\Pi_i}V^{F_i}(t,\hat{x}_i,y_i,y_j,s;
p(\cdot),b_L(\cdot),q_i(\cdot),b_i(\cdot))\nonumber\\
&=\sup_{(q_i(\cdot),b_i(\cdot))\in\Pi_i}E_{t,\hat{x}_i,y_i,y_j,s}
\left[U_i\big(X_i^{\pi_i}(T)+\eta_iY_i(T)-k_i(X_j^{\pi_j^\ast}(T)+\eta_jY_j(T))\big)\right],  i\neq j\in\{1,2\}.
\end{align}

The reinsurer's problem is the following optimization problem: find the optimal strategy $(p^{\ast}(\cdot),b_L^{\ast}(\cdot))\in\Pi_L$ such that the following value function holds.
\begin{align}
&V^L\big(t,x_L,y_L,s;p^{\ast}(\cdot),b_L^{\ast}(\cdot),
\alpha_1^{\ast}(\cdot,p^{\ast}(\cdot),b_L^{\ast}(\cdot)),
\beta_1^{\ast}(\cdot,p^{\ast}(\cdot),b_L^{\ast}(\cdot)),
\alpha_2^{\ast}(\cdot,p^{\ast}(\cdot),b_L^{\ast}(\cdot)),
\beta_2^{\ast}(\cdot,p^{\ast}(\cdot),b_L^{\ast}(\cdot))\big)\nonumber\\
&=\sup_{(p(\cdot),b_L(\cdot))\in\Pi_L}V^L\big(t,x_L,y_L,s;p(\cdot),b_L(\cdot),
\alpha_1^{\ast}(\cdot,p(\cdot),b_L(\cdot)),
\beta_1^{\ast}(\cdot,p(\cdot),b_L(\cdot)),\alpha_2^{\ast}(\cdot,p(\cdot),b_L(\cdot)),
\beta_2^{\ast}(\cdot,p(\cdot),b_L(\cdot))\big)\nonumber\\
&=\sup_{(p(\cdot),b_L(\cdot))\in\Pi_L}E_{t,x_L,y_L,s}
\left[U_L(X_L^{\pi_L}(T)+\eta_LY_L(T))\right],
\end{align}
where $U_L$ is the utility function of the reinsurer.
\end{problem}

\begin{definition}\label{def2}
The pair $\big(p^{\ast}(\cdot),b_L^{\ast}(\cdot),
\alpha_1^{\ast}(\cdot,p^{\ast}(\cdot),b_L^{\ast}(\cdot)),
\beta_1^{\ast}(\cdot,p^{\ast}(\cdot),b_L^{\ast}(\cdot)),
\alpha_2^{\ast}(\cdot,p^{\ast}(\cdot),b_L^{\ast}(\cdot)),
\beta_2^{\ast}(\cdot,p^{\ast}(\cdot),b_L^{\ast}(\cdot))\big)$ is called an equilibrium strategy of the hybrid game.
\end{definition}

Furthermore, if there is no risk of confusion, when the equilibrium strategy of the hybrid game is adopted,
$V^{F_i}(t,\hat{x}_i,y_i,\newline y_j,s;p^{\ast}(\cdot),b_L^{\ast}(\cdot),
\alpha_i^{\ast}(\cdot,p^{\ast}(\cdot),b_L^{\ast}(\cdot)),
\beta_i^{\ast}(\cdot,p^{\ast}(\cdot),b_L^{\ast}(\cdot)))$
is also called the value function of the insurer $i$.

\section{Solution to the hybrid game for CARA preference}\label{section 3}

Assume that both the reinsurer and insurers are constant absolute risk aversion (CARA) agents, i.e., the reinsurer and insurer $i$ ($i\in\{1,2\}$) have exponential utility functions:
\begin{align}
U_L(x_L+\eta_Ly_L)&=-\frac{1}{\gamma_L}\exp(-\gamma_L(x_L+\eta_Ly_L)),\\
U_i(\hat{x}_i+\eta_iy_i-k_i\eta_jy_j)&
=-\frac{1}{\gamma_i}\exp(-\gamma_i(\hat{x}_i+\eta_iy_i-k_i\eta_jy_j)),\quad i\neq j\in\{1,2\},
\end{align}
where $\gamma_L>0$ and $\gamma_i>0$ are the constant absolute risk aversion coefficients of the reinsurer and insurer $i$, respectively.
According to existing literatures, the optimal control problem with delay is infinite-dimensional in general. To make the problem solvable and finite-dimensional, we assume the following conditions on parameters  $C_L=\eta_Le^{-\alpha_Lh_L}$, $B_Le^{-\alpha_Lh_L}=(\alpha_L+A_L+\eta_L)C_L$, $C_i=\eta_ie^{-\alpha_ih_i}$, $B_ie^{-\alpha_ih_i}=(\alpha_i+A_i+\eta_i)C_i$, $i\in\{1,2\}$.

\subsection{Equilibrium strategy and value functions}\label{subsection3.1}

By using the idea of backward induction mentioned in Section \ref{section 2.4} and dynamic programming techniques, we solve the hybrid game problem and obtain the following theorem.

\begin{theorem}\label{Theorem1}
Suppose that $A_1+\eta_1=A_2+\eta_2$, $k_1k_2<1$ and $k_1k_2\rho^2<1$.
The equilibrium strategy of the Stackelberg game problem is $(p^{\ast}(t),b_L^{\ast}(t),q_1^{\ast}(t),b_1^{\ast}(t),q_2^{\ast}(t),b_2^{\ast}(t))$,
where
$b_L^{\ast}(t)$, $b_1^{\ast}(t)$ and $b_2^{\ast}(t)$ are given by
\begin{align}
b^{\ast}_L(t)
=&\frac{s^{-2\beta}}{\gamma_L\varphi^L(t)}\big[\frac{(r-r_0)}{\sigma^2}-2\beta g_1(t)\big],\label{equ:bL*}\\
b^{\ast}_{1}(t)=&\frac{s^{-2\beta}}{(1-k_1k_2)\varphi^{F_1}(t)}
\big(\frac{1}{\gamma_{1}}+\frac{k_1}{\gamma_{2}}\big)
\big[\frac{r-r_0}{\sigma^2}-2\beta g_1(t)\big],\label{equ:b1*}\\
b^{\ast}_{2}(t)=&\frac{s^{-2\beta}}{(1-k_1k_2)\varphi^{F_2}(t)}
\big(\frac{1}{\gamma_{2}}+\frac{k_2}{\gamma_{1}}\big)
\big[\frac{r-r_0}{\sigma^2}-2\beta g_1(t)\big];\label{equ:b2*}
\end{align}
$p^{\ast}(t)$, $q_1^{\ast}(t)$ and $q_2^{\ast}(t)$ under different cases are
given by the following:

\begin{enumerate}[{Case} (1):]
  \item If $N^{cF_1}(t)+\frac{k_1\rho\sigma_2}{\sigma_1}\geq 1$, $N^{cF_2}(t)+\frac{k_2\rho\sigma_1}{\sigma_2}\geq 1$, then
      $p^{\ast}(t)=p$,  $q_1^{\ast}(t)=1$,  $q_2^{\ast}(t)=1$,
      where $\forall p\in[c_F,\bar{c}]$;
  \item If $K[N^{\bar{c}F_1}(t)+\frac{k_1\rho\sigma_2}{\sigma_1}N^{\bar{c}F_2}(t)]\geq 1$,
    $N^{\bar{c}F_2}(t)\leq(1-\frac{k_2\rho\sigma_1}{\sigma_2})M^{F_2}(t)$, then
     $p^{\ast}(t)=\bar{c}$,  $q_1^{\ast}(t)=1$,  $q_2^{\ast}(t)=N^{\bar{c}F_2}(t)+\frac{k_2\rho \sigma_{1}}{\sigma_{2}}$;
  \item If $K[N^{cF_1}(t)+\frac{k_1\rho\sigma_2}{\sigma_1}N^{cF_2}(t)]\geq 1$,
    $(1-\frac{k_2\rho\sigma_1}{\sigma_2})M^{F_2}(t)\leq N^{cF_2}(t)< 1-\frac{k_2\rho\sigma_1}{\sigma_2}$, then
     $p^{\ast}(t)=c_F$,  $q_1^{\ast}(t)=1$,
     $q_2^{\ast}(t)=N^{cF_2}(t)+\frac{k_2\rho \sigma_{1}}{\sigma_{2}}$;
  \item If $K[N^{aF_1}(t)+K^{\tilde{F}_1}M^{F_2}(t)] \geq 1$,
    $N^{cF_2}(t)<(1-\frac{k_2\rho\sigma_1}{\sigma_2})M^{F_2}(t)<\min\{N^{\bar{c}F_2}(t),1-\frac{k_2\rho\sigma_1}{\sigma_2}\}$, then $p^{\ast}(t)=a_2+\gamma_2\sigma_2^2\varphi^{F_2}(t)
    (1-\frac{k_2\rho\sigma_1}{\sigma_2})M^{F_2}(t)$, $q_1^{\ast}(t)=1$,
    $q_2^{\ast}(t)=(1-\frac{k_2\rho\sigma_1}{\sigma_2})M^{F_2}(t)+\frac{k_2\rho \sigma_{1}}{\sigma_{2}}$;
    \item If $K[N^{\bar{c}F_2}(t)+\frac{k_2\rho\sigma_1}{\sigma_2}N^{\bar{c}F_1}(t)]\geq 1$,
    $N^{\bar{c}F_1}(t)\leq(1-\frac{k_1\rho\sigma_2}{\sigma_1})M^{F_1}(t)$, then
     $p^{\ast}(t)=\bar{c}$,
     $q_1^{\ast}(t)=N^{\bar{c}F_1}(t)+\frac{k_1\rho \sigma_{2}}{\sigma_{1}}$, $q_2^{\ast}(t)=1$;
  \item If $K[N^{cF_2}(t)+\frac{k_2\rho\sigma_1}{\sigma_2}N^{cF_1}(t)]\geq 1$,
     $(1-\frac{k_1\rho\sigma_2}{\sigma_1})M^{F_1}(t)\leq N^{cF_1}(t)< 1-\frac{k_1\rho\sigma_2}{\sigma_1}$, then
     $p^{\ast}(t)=c_F$,
     $q_1^{\ast}(t)=N^{cF_1}(t)+\frac{k_1\rho \sigma_{2}}{\sigma_{1}}$,  $q_2^{\ast}(t)=1$;
  \item If $K[N^{aF_2}(t)+K^{\tilde{F}_2}M^{F_1}(t)] \geq 1$,
    $N^{cF_1}(t)<(1-\frac{k_1\rho\sigma_2}{\sigma_1})M^{F_1}(t)<\min\{N^{\bar{c}F_1}(t),1-\frac{k_1\rho\sigma_2}{\sigma_1}\}$, then $p^{\ast}(t)=a_1+\gamma_1\sigma_1^2\varphi^{F_1}(t)
    (1-\frac{k_1\rho\sigma_2}{\sigma_1})M^{F_1}(t)$,
    $q_1^{\ast}(t)=(1-\frac{k_1\rho\sigma_2}{\sigma_1})M^{F_1}(t)+\frac{k_1\rho \sigma_{2}}{\sigma_{1}}$,
     $q_2^{\ast}(t)=1$;
  \item If $K[N^{\bar{c}F_1}(t)+\frac{k_1\rho\sigma_2}{\sigma_1}N^{\bar{c}F_2}(t)]< 1$, $K[N^{\bar{c}F_2}(t)+\frac{k_2\rho\sigma_1}{\sigma_2}N^{\bar{c}F_1}(t)]< 1$,
      $\frac{P^N(t)}{P^D(t)}\geq \bar{c}$, then  $p^{\ast}(t)=\bar{c}$,
    $q_1^{\ast}(t)=K[N^{\bar{c}F_1}(t)+\frac{k_1\rho\sigma_2}{\sigma_1}N^{\bar{c}F_2}(t)]$,
   $q_2^{\ast}(t)=K[N^{\bar{c}F_2}(t)+\frac{k_2\rho\sigma_1}{\sigma_2}N^{\bar{c}F_1}(t)]$;
  \item If $K[N^{cF_1}(t)+\frac{k_1\rho\sigma_2}{\sigma_1}N^{cF_2}(t)]< 1$, $K[N^{cF_2}(t)+\frac{k_2\rho\sigma_1}{\sigma_2}N^{cF_1}(t)]< 1$,
       $\frac{P^N(t)}{P^D(t)}\leq c_F$, then $p^{\ast}(t)=c_F$,
    $q_1^{\ast}(t)=K[N^{cF_1}(t)+\frac{k_1\rho\sigma_2}{\sigma_1}N^{cF_2}(t)]$, $q_2^{\ast}(t)=K[N^{cF_2}(t)+\frac{k_2\rho\sigma_1}{\sigma_2}N^{cF_1}(t)]$;
  \item If $K[\frac{\frac{P^N(t)}{P^D(t)}-a_1}{\gamma_{1}\sigma_{1}^2\varphi^{F_1}(t)} +\frac{k_1\rho(\frac{P^N(t)}{P^D(t)}-a_2)}{\gamma_{2}\sigma_{2}\sigma_{1}\varphi^{F_2}(t)}]<1$,
    $K[\frac{\frac{P^N(t)}{P^D(t)}-a_2}{\gamma_{2}\sigma_{2}^2\varphi^{F_2}(t)} +\frac{k_2\rho(\frac{P^N(t)}{P^D(t)}-a_1)}{\gamma_{1}\sigma_{1}\sigma_{2}\varphi^{F_1}(t)}]<1$, $c_F<\frac{P^N(t)}{P^D(t)}<\bar{c}$, then
    $p^{\ast}(t)=\frac{P^N(t)}{P^D(t)}$,
     $q_1^{\ast}(t)=K [\frac{\frac{P^N(t)}{P^D(t)}-a_1}{\gamma_{1}\sigma_{1}^2\varphi^{F_1}(t)} +\frac{k_1\rho(\frac{P^N(t)}{P^D(t)}-a_2)}{\gamma_{2}\sigma_{2}\sigma_{1}\varphi^{F_2}(t)}]$,
     $q_2^{\ast}(t)=K[\frac{\frac{P^N(t)}{P^D(t)}-a_2}{\gamma_{2}\sigma_{2}^2\varphi^{F_2}(t)} +\frac{k_2\rho(\frac{P^N(t)}{P^D(t)}-a_1)}{\gamma_{1}\sigma_{1}\sigma_{2}\varphi^{F_1}(t)}]$.
\end{enumerate}
where
$\varphi^L(t)$, $\varphi^{F_i}(t)$ and $g_1(t)$ are given by \eqref{equ:varphiL}, \eqref{equ:varphiF} and \eqref{equ:g1t}, respectively; $K$, $K^{\tilde{F}_i}$, $K^{F_i}$, $N^{cF_i}(t)$, $N^{\bar{c}F_i}(t)$, $N^{aF_i}(t)$ and $M^{F_i}(t)$ are given by \eqref{equ:KKK}; $P^{N}(t)$ and $P^{D}(t)$ are given by \eqref{equ:PNt} and \eqref{equ:PDt}, respectively.

The value function of the reinsurer is given by
\begin{align}
V^L(t,x_L,y_L,s)
=-\frac{1}{\gamma_L}\exp\{-\gamma_L\varphi^L(t)(x_L+\eta_Ly_L)+g_1(t)s^{-2\beta}+g^L_2(t)\},
\end{align}
and, the value function of insurer $1$ is given by
\begin{align}
V^{F_1}(t,\hat{x}_{1},y_{1},y_{2},s)
=&-\frac{1}{\gamma_{1}}\exp\{-\gamma_{1}\varphi^{F_1}(t)(\hat{x}_{1}
+\eta_{1}y_{1}-k_1\eta_{2}y_{2})+g_1(t)s^{-2\beta}+g^{F_1}_2(t)\},
\end{align}
and, the value function of insurer $2$ is given by
\begin{align}
V^{F_2}(t,\hat{x}_{2},y_{2},y_{1},s)
=&-\frac{1}{\gamma_{2}}\exp\{-\gamma_{2}\varphi^{F_2}(t)(\hat{x}_{2}
+\eta_{2}y_{2}-k_2\eta_{1}y_{1})+g_1(t)s^{-2\beta}+g^{F_2}_2(t)\},
\end{align}
where $g^{L}_2(t)$, $g^{F_1}_2(t)$ and $g^{F_2}_2(t)$ under different cases are given by Table \ref{g};
\begin{table}[htp]
\setlength{\abovecaptionskip}{0.cm}
\caption{$g^{L}_2(t)$, $g^{F_1}_2(t)$ and $g^{F_2}_2(t)$ under different cases.} % title name of the table
\centering % centering table
\setlength{\tabcolsep}{8mm}
\begin{tabular}
    {@{}ccccc@{}}
    \toprule
    %\multicolumn{2}{c}{Item} \\ \cmidrule(r){1-2}
    $Cases$ & $g^{L}_2(t)$ & $g^{F_1}_2(t)$  & $g^{F_2}_2(t)$  \\ \midrule
    $Case (1)$ & $g^{La}_2(t)$ & $g^{F_1a}_2(t)$ & $g^{F_2a}_2(t)$ \\ \midrule
    $Case (2)$ & $g^{L2b1}_2(t)$ & $g^{F_1b1}_2(t)$ & $g^{\tilde{F}_2b1}_2(t)$ \\ \midrule
    $Case (3)$ & $g^{L2b2}_2(t)$ & $g^{F_1b2}_2(t)$ & $g^{\tilde{F}_2b2}_2(t)$ \\ \midrule
    $Case (4)$ & $g^{L2b3}_2(t)$ & $g^{F_1b3}_2(t)$ & $g^{\tilde{F}_2b3}_2(t)$ \\ \midrule
    $Case (5)$ & $g^{L1b1}_2(t)$ & $g^{\tilde{F}_1b1}_2(t)$ & $g^{F_2b1}_2(t)$ \\ \midrule
    $Case (6)$ & $g^{L1b2}_2(t)$ & $g^{\tilde{F}_1b2}_2(t)$ & $g^{F_2b2}_2(t)$ \\ \midrule
    $Case (7)$ & $g^{L1b3}_2(t)$ & $g^{\tilde{F}_1b3}_2(t)$ & $g^{F_2b3}_2(t)$ \\ \midrule
    $Case (8)$ & $g^{Lc1}_2(t)$ & $g^{F_1c1}_2(t)$ & $g^{F_2c1}_2(t)$ \\ \midrule
    $Case (9)$ & $g^{Lc2}_2(t)$ & $g^{F_1c2}_2(t)$ & $g^{F_2c2}_2(t)$ \\ \midrule
    $Case (10)$ & $g^{Lc3}_2(t)$ & $g^{F_1c3}_2(t)$ & $g^{F_2c3}_2(t)$ \\
\bottomrule
\end{tabular}
\label{g}
\end{table}
$g^{La}_2(t)$ is given by equation \eqref{equ:gLa}; for $j\in\{1,2\}$, $g^{Ljb1}_2(t)$, $g^{Ljb2}_2(t)$ and $g^{Ljb3}_2(t)$ are given by equations \eqref{equ:gLjb1}, \eqref{equ:gLjb2} and \eqref{equ:gLjb3}, respectively; $g^{Lc1}_2(t)$, $g^{Lc2}_2(t)$ and $g^{Lc3}_2(t)$ are given by equations \eqref{equ:gLc1}, \eqref{equ:gLc2} and \eqref{equ:gLc3}, respectively;
$g^{F_ia}_2(t)$ is given by equation \eqref{equ:g2Fia};
for $i\in\{1,2\}$, $g^{F_ib1}_2(t)$, $g^{F_ib2}_2(t)$ and $g^{F_ib3}_2(t)$ are given by equations \eqref{equ:g2Fib1}, \eqref{equ:g2Fib2} and \eqref{equ:g2Fib3},respectively;
$g^{\tilde{F}_jb1}_2(t)$, $g^{\tilde{F}_jb2}_2(t)$ and $g^{\tilde{F}_jb3}_2(t)$ are given by equations \eqref{equ:g2Fjb1}, \eqref{equ:g2Fjb2} and \eqref{equ:g2Fjb3}, respectively;
$g^{F_ic1}_2(t)$, $g^{F_ic2}_2(t)$ and $g^{F_ic3}_2(t)$ are given by equations \eqref{equ:g2Fic1}, \eqref{equ:g2Fic2} and \eqref{equ:g2Fic3},respectively.
\end{theorem}

\begin{proof}
  See \ref{Appendix A}.
\end{proof}

\begin{remark}
Case (1) in Theorem \ref{Theorem1} corresponds to Case (La) of \ref{Appendix A}.
In this situation, both two insurers do not sign the reinsurance contract and bear all claims themselves.
Case (2), Case (3) and Case (4) in Theorem \ref{Theorem1} correspond to the case of $i=1,j=2$ in Case (Lb) of \ref{Appendix A}.
This situation can be understood as that insurer $1$ bears all the claim risk by itself, and insurer $2$ adopts the reinsurance strategy to spread its claim risk.
Case (5), Case (6) and Case (7) in Theorem \ref{Theorem1} correspond to the case of $i=2,j=1$ in Case (Lb) of \ref{Appendix A}.
That is, insurer $2$ bears all the claim risk by itself, and insurer $1$ adopts the reinsurance strategy to spread its claim risk.
Case (8), Case (9) and Case (10) in Theorem \ref{Theorem1} correspond to the Case (Lc) of \ref{Appendix A}.
In this case, both two insurers will sign the reinsurance contract to spread their claims risk.
\end{remark}

\begin{remark}
More generally, if there's one reinsurer and $n$ insurers in the insurance market,
the optimal premium strategy and the optimal reinsurance strategies will have  $C_n^0+3(C_n^1+C_n^2+\cdots+C_n^n)=1+3(2^n-1)$ situations.
\end{remark}

\begin{corollary}\label{corollary1}
The insurer $i$'s ($i\in\{1,2\}$) optimal reinsurance strategy can be expressed by the optimal premium price strategy and the insurer $j$'s ($j\neq i\in\{1,2\}$) optimal reinsurance strategy. That is,
\begin{align}
q_i^{\ast}(t)
=\Big[\frac{p^{\ast}(t)-a_i}{\gamma_i\sigma_i^2\varphi^{F_i}(t)}+\frac{k_i\rho \sigma_jq_j^{\ast}(t)}{\sigma_i}\Big]\wedge 1, \quad i\neq j\in\{1,2\}.\label{equ:qiqj}
\end{align}

The optimal investment strategy of insurer $i$'s ($i\in\{1,2\}$) can be expressed by the optimal reinsurance strategy of insurer $j$'s ($j\neq i\in\{1,2\}$). That is
\begin{align}
b^{\ast}_i(t)
=k_ib_j^{\ast}(t)+\frac{1}{\gamma_i\varphi^{F_i}(t)s^{2\beta}}
\big[\frac{r-r_0}{\sigma^2}-2\beta g_1(t)\big], \quad i\neq j\in\{1,2\}.\label{equ:bibj}
\end{align}

Moreover, we have
\begin{align}
\frac{\partial q_i^{\ast}(t)}{\partial p^{\ast}(t)}>0,\quad
\frac{\partial q_i^{\ast}(t)}{\partial q_j^{\ast}(t)}>0,\quad
\frac{\partial b^{\ast}_i(t)}{\partial b^{\ast}_j(t)}>0,
\quad i\neq j\in\{1,2\}.\label{equ:weifen}
\end{align}
\end{corollary}

\begin{proof}
\eqref{equ:qiqj},\eqref{equ:bibj} and \eqref{equ:weifen} are obviously established, and we omit the proof here.
\end{proof}

Through Corollary \ref{corollary1}, we find that the equilibrium reinsurance-investment strategies of the two insurers interact with each other and exhibit herd effect.
That is insurer $i$'s ($i\in\{1,2\}$) will imitate insurer $j$'s ($j\neq i\in\{1,2\}$) reinsurance-investment strategy.
More specifically, the amount of insurer $i$ invested in the risky asset will increase with the amount that insurer $j$ invested in the risky asset; the reserve proportion of insurer $i$ will also increase with the increase of the reserve proportion of insurer $j$.
Furthermore, the optimal reinsurance strategies of insurers depend on the optimal premium strategy.
The reserve proportion of insurers will increase with the increase of reinsurance premium price, that is, a high reinsurance premium price will reduce the reinsurance demand.
These findings illustrate considering leader-follower relationship and competitive relationship at the same time will make decisions more rational and more realistic.

%From Theorem \ref{Theorem1}, we find that the equilibrium reinsurance-investment strategy is independent of the current wealth. In addition, the investment strategy of insurer $i$ ($i\in\{1,2\}$) is independent of its reinsurance strategy. Similarly, reinsurer's investment strategy is not related to the reinsurance premium price.

Similar to \cite{Lin-2015-2016}, we find that the reinsurer's investment strategy $b_L^\ast(t)$ in Theorem \ref{Theorem1} contains two parts.
The first part $\frac{1}{\gamma_L\varphi^L(t)s^{2\beta}}\frac{(r-r_0)}{\sigma^2}$ has an updated instantaneous volatility at the current time $t$, while the second part $\frac{-2\beta}{\gamma_L\varphi^L(t)s^{2\beta}}g_1(t)$ results from the fact that the reinsurer tries to hedge its portfolio against the additional volatility risk.
When $\beta>0$, we have $e^{-2r_0\beta(T-t)}<1$ and $\frac{-2\beta}{\gamma_L\varphi^L(t)s^{2\beta}}g_1(t)>0$, it will cause positive deviation from the classical result (i.e., the investment strategy when the price of the risky asset obeys GBM model
\footnote{If $\beta=0$, the CEV model reduces to the GBM model. Then, the optimal investment strategies of the reinsurer and the insurers are given by
$b^{\ast}_L(t)=\frac{1}{\gamma_L\varphi^L(t)}\frac{(r-r_0)}{\sigma^2}, b^{\ast}_{i}(t)=\frac{1}{(1-k_1k_2)\varphi^{F_i}(t)}
\big(\frac{1}{\gamma_{i}}+\frac{k_i}{\gamma_{j}}\big)\frac{r-r_0}{\sigma^2},
i\neq j\in\{1,2\}.$}).
Conversely, when $\beta<0$, we have $e^{-2r_0\beta(T-t)}>1$ and $\frac{-2\beta}{\gamma_L\varphi^L(t)s^{2\beta}}g_1(t)<0$, it will cause negative deviation from the classical result.
Correspondingly, the investment strategies of two insurers have similar analysis.

\begin{corollary}\label{corollary2}
If $\beta\geq 0$, some properties of $b_L^{\ast}(t)$ and $b_i^{\ast}(t)$ ($i=1,2$) are given in Table \ref{TablebLbF}, \eqref{equ:balpha} and \eqref{equ:beta}.
\begin{table}[htp]
\small
\caption{The properties of $b_L^{\ast}(t)$ and $b_i^{\ast}(t)$.\\}
\centering
\renewcommand\arraystretch{1}
\begin{tabular}{p{1.5cm}<{\centering}|p{1.5cm}<{\centering}|p{1.5cm}<{\centering}|
p{1.5cm}<{\centering}|p{1.5cm}<{\centering}}% 通过添加 |来表示是否需要绘制竖线
$\frac{\partial b_L^{\ast}(t)}{\partial\gamma_L}$
& $\frac{\partial b_L^{\ast}(t)}{\partial h_L}$
& $\frac{\partial b_i^{\ast}(t)}{\partial \gamma_i}$
& $\frac{\partial b_i^{\ast}(t)}{\partial k_i}$
& $\frac{\partial b_i^{\ast}(t)}{\partial h_i}$\\
\hline  %在第一行和第二行之间绘制横线
$-$&$-$&$-$&$+$&$-$\\
\end{tabular}
\label{TablebLbF}
\end{table}
\begin{footnotesize}
\begin{equation}\label{equ:balpha}
\frac{\partial b_L^{\ast}(t)}{\partial\alpha_L}=\left\{
\begin{array}{rcl}
>0, &  {\alpha_L>-\frac{1}{h_L}ln\frac{1}{h_L};}\\
=0,&  {\alpha_L=-\frac{1}{h_L}ln\frac{1}{h_L};}\\
<0, &  {\alpha_L<-\frac{1}{h_L}ln\frac{1}{h_L}.}
\end{array} \right.
\frac{\partial b_i^{\ast}(t)}{\partial\alpha_i}=\left\{
\begin{array}{rcl}
>0, &  {\alpha_i>-\frac{1}{h_i}ln\frac{1}{h_i};}\\
=0,&  {\alpha_i=-\frac{1}{h_i}ln\frac{1}{h_i};}\\
<0, &  {\alpha_i<-\frac{1}{h_i}ln\frac{1}{h_i}.}
\end{array} \right.
\end{equation}
\end{footnotesize}
Furthermore, if $r_0+\alpha_L<1$ and $r_0+\alpha_i<1$, $i=1,2$, then
\begin{footnotesize}
\begin{equation}\label{equ:beta}
\frac{\partial b_L^{\ast}(t)}{\partial\eta_L}=\left\{
\begin{array}{rcl}
>0, & {h_L<-\frac{1}{\alpha_L}ln(1-r_0-\alpha_L);}\\
=0,&  {h_L=-\frac{1}{\alpha_L}ln(1-r_0-\alpha_L);}\\
<0, &  {h_L>-\frac{1}{\alpha_L}ln(1-r_0-\alpha_L).}
\end{array} \right.
\frac{\partial b_i^{\ast}(t)}{\partial\eta_i}=\left\{
\begin{array}{rcl}
 >0,&  {h_i<-\frac{1}{\alpha_i}ln(1-r_0-\alpha_i);}\\
 =0,&  {h_i=-\frac{1}{\alpha_i}ln(1-r_0-\alpha_i);}\\
 <0,&  {h_i>-\frac{1}{\alpha_i}ln(1-r_0-\alpha_i).}
\end{array} \right.
\end{equation}
\end{footnotesize}
\end{corollary}

\begin{proof}
See \ref{Appendix B}.
\end{proof}

Corollary \ref{corollary2} states that the effect of the delay weight on the optimal investment strategy depends on the delay time.
For the reinsurer, when the delay time $h_L>-\frac{1}{\alpha_L}ln(1-r_0-\alpha_L)$,
the greater the delay weight, the less money is invested in the risky asset. On the contrary, when the delay time $h_L<-\frac{1}{\alpha_L}ln(1-r_0-\alpha_L)$, the greater the delay weight, the more money is invested in the risky asset.
The optimal investment strategies of the insurers have similar rules.
In other words, when the delay time selected by the reinsurer (insurer $i$) is relatively long, the greater the weight of the integrated performance in the past, the more conservative the strategy made by the reinsurer (insurer $i$).
On the contrary, when the memory time is relatively short, the smaller the weight of  integrated performance in the past, the more conservative investment strategy will be adopted. It also illustrates that the reinsurer and insurers manage investment risk according to the relevant parameters of delay.

\subsection{Verification theorem}\label{subsection3.2}

In order to prove that the equilibrium strategy given in Theorem \ref{Theorem1} is indeed optimal for all three parties of the hybrid game, we give a verification theorem in this section. For $i\neq j\in\{1,2\}$, let
\begin{align}
&\mathcal{A}^{F_i}V^{F_i}(t,\hat{x}_i,y_i,y_j,s)\nonumber\\
&=V^{F_i}_t+V^{F_i}_{\hat{x}_i}
\big[\theta_ia_i-k_i\theta_ja_j-(p(t)-a_i)(1-q_i(t))+k_i(p(t)-a_j)(1-q_j^{\ast}(t))
+A_ix_i-k_iA_jx_j+B_iy_i-k_iB_jy_j\nonumber\\
&
+C_iz_i-k_iC_jz_j+(r-r_0)(b_i(t)-k_ib_j^{\ast}(t))\big]
+\frac{1}{2}\big[(q_i(t)\sigma_i)^2+(k_iq_j^{\ast}(t)\sigma_j)^2
-2q_i(t)\sigma_ik_iq_j^{\ast}(t)\sigma_j\rho\nonumber\\
&
+(b_i(t)-k_ib_j^{\ast}(t))^2\sigma^2s^{2\beta}\big]V^{F_i}_{\hat{x}_i\hat{x}_i}
+(x_i-\alpha_iy_i-e^{-\alpha_ih_i}z_i)V^{F_i}_{y_i}
+(x_j-\alpha_jy_j-e^{-\alpha_jh_j}z_j)V^{F_i}_{y_j}\nonumber\\
&+rsV^{F_i}_{s}
+\frac{1}{2}\sigma^2s^{2\beta+2}V^{F_i}_{ss}
+(b_i(t)-k_ib_j^{\ast}(t))\sigma^2s^{2\beta+1} V^{F_i}_{\hat{x}_is},\\
&\mathcal{A}^{L}V^{L}(t,x_L,y_L,s)\nonumber\\
&=V^L_t
+V^L_{x_L}\big[(p(t)-a_1)(1-q_1^{\ast}(t))+(p(t)-a_2)(1-q_2^{\ast}(t))+(r-r_0)b_L(t)
+A_Lx_L+B_Ly_L+C_Lz_L\big]\nonumber\\
&
+\frac{1}{2}\big[(1-q_1^{\ast}(t))^2(\sigma_1)^2+(1-q_2^{\ast}(t))^2(\sigma_2)^2
+(b_L(t))^2\sigma^2s^{2\beta}
+2(1-q_1^{\ast}(t))(1-q_2^{\ast}(t))\sigma_1\sigma_2\rho\big]V^L_{x_Lx_L}\nonumber\\
&
+(x_L-\alpha_Ly_L-e^{-\alpha_Lh_L}z_L)V^L_{y_L}
+rsV^L_{s}+\frac{1}{2}\sigma^2s^{2\beta+2}V^L_{ss}
+b_L(t)\sigma^2s^{2\beta+1} V^L_{x_Ls}.
\end{align}

We first give the following lemmas:
\begin{lemma}\label{lemma1}
Let $\mathcal{M}^{F_i}=\mathbb{R}\times\mathbb{R}^{+}\times\mathbb{R}^{+}\times\mathbb{R}^{+}$, $i\in\{1,2\}$. Take a sequence of bounded open sets $\mathcal{M}_1^{F_i},\mathcal{M}_2^{F_i},\mathcal{M}_3^{F_i},\cdots$, with $\mathcal{M}_n^{F_i}\subset\mathcal{M}_{n+1}^{F_i}\subset\mathcal{M}^{F_i},n=1,2,\cdots$, and $\mathcal{M}^{F_i}=\cup_n\mathcal{M}_{n}^{F_i}$. For $(\hat{x}_i,y_i,y_j,s)\in \mathcal{M}_{n}^{F_i}$, let $\tau_n$ be the exit time of $(\hat{X}_i(t),Y_i(t),Y_j(t),S(t))$ from $\mathcal{M}_{n}^{F_i}$. Then, for $n=1,2,\cdots$,
$E_{t,\hat{x}_i,y_i,y_j,s}\big\{\big[V^{F_i}(\tau_n\wedge T,\hat{X}_i(\tau_n\wedge T),Y_i(\tau_n\wedge T),Y_j(\tau_n\wedge T),S(\tau_n\wedge T))\big]^2\big\}<+\infty$.
\end{lemma}

\begin{proof}
See \ref{Appendix C}.
\end{proof}

\begin{lemma}\label{lemma2}
Let $\mathcal{M}^L=\mathbb{R}^{+}\times\mathbb{R}^{+}\times\mathbb{R}^{+}$. Take a sequence of bounded open sets $\mathcal{M}_1^L,\mathcal{M}_2^L,\mathcal{M}_3^L,\cdots$, with $\mathcal{M}_n^L\subset\mathcal{M}_{n+1}^L\subset\mathcal{M}^L,n=1,2,\cdots$, and $\mathcal{M}^L=\cup_n\mathcal{M}_{n}^L$. For $(x_L,y_L,s)\in \mathcal{M}_{n}^L$, let $\tau_n$ be the exit time of $(X_L(t),Y_L(t),S(t))$ from $\mathcal{M}_{n}^L$. Then, for $n=1,2,\cdots$,
 $E_{t,x_L,y_L,s}\big\{\big[V^L(\tau_n\wedge T,X_L(\tau_n\wedge T),Y_L(\tau_n\wedge T),S(\tau_n\wedge T))\big]^2\big\}<+\infty$.
\end{lemma}

\begin{proof}
The proof of this lemma is similar to that of Lemma \ref{lemma1}.
\end{proof}

\begin{theorem}\label{Theorem2}
(Verification theorem)
The equilibrium strategy $(p^{\ast}(t),b_L^{\ast}(t),q_1^{\ast}(t),b_1^{\ast}(t), q_2^{\ast}(t),b_2^{\ast}(t))$ described in Theorem \ref{Theorem1} achieves optimality in $\Pi_L\times\Pi_{1}\times\Pi_{2}$.
\end{theorem}

\begin{proof}
See \ref{Appendix D}.
\end{proof}

\subsection{Special cases}\label{subsection3.3}

In what follows, we present several special cases of our model.

{\bfseries {Special case 1:}}
If the reinsurer and two insurers do not consider the effect of the bounded memory, i.e., $\eta_L=\eta_i=h_L=h_i=\alpha_L=\alpha_i=0$, then $B_L=B_i=C_L=C_i=0$ and $A_L=A_i=r_0$, $i\in\{1,2\}$. The optimal investment strategies, denoted as
$\hat{b}_L^{\ast}(t)$, $\hat{b}_i^{\ast}(t),i\in\{1,2\}$, are given by
\begin{align}
\hat{b}^{\ast}_L(t)
=&\frac{e^{-r_0(T-t)}s^{-2\beta}}{\gamma_L}\big[\frac{(r-r_0)}{\sigma^2}-2\beta g_1(t)\big]
=b^{\ast}_L(t)e^{\frac{\eta_L}{1+\eta_L}(1-r_0-\alpha_L-e^{-\alpha_Lh_L})(T-t)},\label{equ:bLhat}\\
\hat{b}^{\ast}_{i}(t)
=&\frac{e^{-r_0(T-t)}s^{-2\beta}}{(1-k_1k_2)}
\big(\frac{1}{\gamma_{i}}+\frac{k_i}{\gamma_{j}}\big)
\big[\frac{r-r_0}{\sigma^2}-2\beta g_1(t)\big]
=b^{\ast}_i(t)e^{\frac{\eta_i}{1+\eta_i}(1-r_0-\alpha_i-e^{-\alpha_ih_i})(T-t)},\label{equ:bihat}
\end{align}
where $g_1(t)$ is given by \eqref{equ:g1t}.
The optimal reinsurance premium strategy and the optimal reinsurance strategies in the interior case (i.e., Case (10) in Theorem 1) become
\begin{align}
\hat{p}^{\ast}(t)=\frac{\hat{P}^N(t)}{\hat{P}^D(t)},\quad
\hat{q}_i^{\ast}(t)=\frac{e^{-r_0(T-t)}}{1-k_1k_2\rho^2} \Big[\frac{\frac{\hat{P}^N(t)}{\hat{P}^D(t)}-a_i}{\gamma_{i}\sigma_{i}^2} +\frac{k_i\rho(\frac{\hat{P}^N(t)}{\hat{P}^D(t)}-a_j)}{\gamma_{j}\sigma_{j}\sigma_{i}}\Big],\quad i\neq j\in\{1,2\},\label{equ:pqhat}
\end{align}
where
\begin{align*}
\hat{P}^N(t)=&e^{r_0(T-t)}(\sigma_1\sigma_2)^2(\gamma_2\hat{D}^{F_1}+\gamma_1\hat{D}^{F_2})
+2a_1\sigma_2^2\gamma_2\hat{D}^{\tilde{F}_1}+2a_2\sigma_1^2\gamma_1\hat{D}^{\tilde{F}_2}
+(a_1+a_2)\rho\sigma_1\sigma_2\hat{D}^{F_{12}},\\
\hat{P}^D(t)=&2\sigma_2^2\gamma_2\hat{D}^{\tilde{F}_1}
+2\sigma_1^2\gamma_1\hat{D}^{\tilde{F}_2}+2\rho\sigma_1\sigma_2\hat{D}^{F_{12}},
\hat{D}^{F_i}=\gamma_i(1-k_1k_2\rho^2)+\gamma_L[1+k_j\rho^2+\frac{\sigma_j\rho}{\sigma_i}(1+k_j)], i\neq j\in\{1,2\},\nonumber\\
\hat{D}^{\tilde{F}_i}=&1+\frac{\gamma_L(1+(k_j\rho)^2+2k_j\rho^2)}{2\gamma_i(1-k_1k_2\rho^2)}, i\neq j\in\{1,2\},
\hat{D}^{F_{12}}=k_1\gamma_1+k_2\gamma_2+\frac{\gamma_L(1+k_1+k_2+k_1k_2\rho^2)}{1-k_1k_2\rho^2}.
\end{align*}

From  \eqref{equ:bLhat} and \eqref{equ:bihat}, we can find that
when $1-r_0-\alpha_L-e^{-\alpha_Lh_L}\geq 0$, then $\hat{b}^{\ast}_L(t)\geq b^{\ast}_L(t)$.
That is, when the delay time satisfies $h_L\geq-\frac{1}{\alpha_L}ln(1-r_0-\alpha_L)$, the amount of investment in the risky asset without delay is greater than that with delay, that is to say, delay makes the investment strategy more conservative in this case.
On the contrary, when $1-r_0-\alpha_L-e^{-\alpha_Lh_L}< 0$, then $\hat{b}^{\ast}_L(t)< b^{\ast}_L(t)$.
That is to say, when the reinsurer's delay time $h_L$ is less than $-\frac{1}{\alpha_L}ln(1-r_0-\alpha_L)$, the amount invested in the risky asset with delay is larger than that without delay, i.e., the delay factor stimulates the investment in this case.
Accordingly, the investment strategies of insurers have similar analysis.
Generally speaking, delay factor discourages or stimulates investment depending on the length of the delay.

\begin{corollary}\label{corollary3}
For $i\in\{1,2\}$, we have
\begin{align}
\lim_{t\rightarrow T}[\hat{b}^{\ast}_L(t)-b^{\ast}_L(t)]=0,\quad
\lim_{t\rightarrow T}[\hat{b}^{\ast}_{i}(t)-b^{\ast}_i(t)]=0, \quad
\lim_{t\rightarrow T}[\hat{p}^{\ast}(t)-p^{\ast}(t)]=0,\quad
\lim_{t\rightarrow T}[\hat{q}_i^{\ast}(t)-q_i^{\ast}(t)]=0,
\end{align}
where $p^{\ast}(t)$ and $q_i^{\ast}(t),i\in\{1,2\}$ are the optimal reinsurance premium strategy and the optimal reinsurance strategies respectively in Case (10) of Theorem \ref{Theorem1}.
\end{corollary}

\begin{proof}
  This corollary is easily obtained by \eqref{equ:bL*}, \eqref{equ:b1*}, \eqref{equ:b2*}, \eqref{equ:pq*}, \eqref{equ:bLhat}, \eqref{equ:bihat} and  \eqref{equ:pqhat}, and we omits the proof here.
\end{proof}

Corollary \ref{corollary3} indicates that when time $t$ tends to the terminal time $T$,
the equilibrium strategy with delay and without delay will tend to be consistent.
In particular,
the equilibrium strategy with delay is equal to that without delay at the terminal time $T$.

{\bfseries {Special case 2:}}
We study a stochastic differential reinsurance-investment game between one reinsurer and one insurer, i.e., $i=j=1$. Then, $\rho=1$, $c_F=c_1=(1+\theta_1)a_1$, $k_i=k_j=0$.
At this point, the hybrid game becomes a pure Stackelberg game problem.
Using the method similar to that in Section \ref{subsection3.1}, we can get the equilibrium strategy $(\tilde{p}^{\ast}(t),\tilde{b}^{\ast}_L(t),\tilde{q}^{\ast}_1(t),\tilde{b}^{\ast}_1(t))$ and value functions $\tilde{V}^L(t,x_L,y_L,s)$, $\tilde{V}^{F_1}(t,x_1,y_1,s)$.
$\tilde{b}^{\ast}_L(t)$ and $\tilde{b}^{\ast}_1(t)$ are given by
\begin{align*}
\tilde{b}^{\ast}_L(t)
=\frac{1}{\gamma_L\varphi^L(t)s^{2\beta}}\Big[\frac{(r-r_0)}{\sigma^2}-2\beta g_1(t)\Big],
\quad
\tilde{b}^{\ast}_1(t)
=\frac{1}{\gamma_1\varphi^{F_1}(t)s^{2\beta}}\Big[\frac{(r-r_0)}{\sigma^2}-2\beta g_1(t)\Big].
\end{align*}
$\tilde{p}^{\ast}(t)$ and $\tilde{q}^{\ast}_1(t)$ under different cases are given by Table \ref{Tablestrategy0},
\begin{table}[htp]
\setlength{\abovecaptionskip}{0.cm}
\caption{ The optimal premium strategy and the optimal reinsurance strategy under different cases.} % title name of the table
\centering % centering table
\setlength{\tabcolsep}{10mm}
\begin{tabular}
    {@{}cccc@{}}
    \toprule
    %\multicolumn{2}{c}{Item} \\ \cmidrule(r){1-2}
    $Cases$ & $p^{\ast}(t)$ & $q^{\ast}_1(t)$   \\ \midrule
    $(1)~N^{cF_1}(t)\geq 1$ & $\forall p\in[c_F,\bar{c}]$ & $1$ \\ \midrule
    $(2)~N^{\bar{c}F_1}(t)\leq M^{F_1}(t)$ & $\bar{c}$ & $N^{\bar{c}F_1}(t)$ \\ \midrule
    $(3)~M^{F_1}(t)\leq N^{cF_1}(t)<1$ & $c_F$ & $N^{cF_1}(t)$ \\ \midrule
    $(4)~N^{cF_1}(t)<M^{F_1}(t)<N^{\bar{c}F_1}(t)$ & $a_1+M^{F_1}(t)\gamma_1\sigma_1^2\varphi^{F_1}(t)$ & $M^{F_1}(t)$ \\
\bottomrule
\end{tabular}
\label{Tablestrategy0}
\end{table}
where $\varphi^L(t)$, $\varphi^{F_1}(t)$ and $g_1(t)$ are given by \eqref{equ:varphiL}, \eqref{equ:varphiF} and \eqref{equ:g1t}, respectively;
$N^{cF_1}(t)$, $N^{\bar{c}F_1}(t)$ and $M^{F_1}(t)$ are given by \eqref{equ:KKK}.
The value function of the reinsurer is given by
\begin{align*}
\tilde{V}^L(t,x_L,y_L,s)
=&-\frac{1}{\gamma_L}\exp\{-\gamma_L\varphi^L(t)(x_L+\eta_Ly_L)+g_1(t)s^{-2\beta}+g^{L}_2(t)\},
\end{align*}
the value function of the insurer is given by
\begin{align*}
\tilde{V}^{F_1}(t,x_1,y_1,s)
=&-\frac{1}{\gamma_1}\exp\{-\gamma_1\varphi^{F_1}(t)(x_1+\eta_1y_1)+g_1(t)s^{-2\beta}+g^{F_1}_2(t)\},
\end{align*}
where $g^{L}_2(t)$ and $g^{F_1}_2(t)$ under different cases are given by Table
\ref{gspecial2}, $g(t)$ is given by \eqref{equ:g(t)},
\begin{table}[htp]
\setlength{\abovecaptionskip}{0.cm}
\caption{$g^{L}_2(t)$ and $g^{F_1}_2(t)$ under different cases.} % title name of the table
\centering % centering table
\setlength{\tabcolsep}{8mm}
\begin{tabular}
    {@{}ccc@{}}
    \toprule
    %\multicolumn{2}{c}{Item} \\ \cmidrule(r){1-2}
    $Cases$ & $g^{L}_2(t)$ & $g^{F_1}_2(t)$    \\ \midrule
    $Case (1)$ & $g^{0La}_2(t)$ & $g^{0\tilde{F}_1a}_2(t)$  \\ \midrule
    $Case (2)$ & $g^{0Lb1}_2(t)$ & $g^{0\tilde{F}_1b1}_2(t)$  \\ \midrule
    $Case (3)$ & $g^{0Lb2}_2(t)$ & $g^{0\tilde{F}_1b2}_2(t)$  \\ \midrule
    $Case (4)$ & $g^{0Lb3}_2(t)$ & $g^{0\tilde{F}_1b3}_2(t)$  \\
\bottomrule
\end{tabular}
\label{gspecial2}
\end{table}
\begin{align*}
g^{0La}_2(t)&=g(t),\\
g^{0Lb1}_2(t)&=g(t)+\frac{\gamma_L^2\sigma_1^2}{4(A_L+\eta_L)}[(\varphi^L(t))^2-1]
+\int_T^t\bar{\theta}a_1\gamma_L\varphi^L(s)[1+\frac{\gamma_L\varphi^L(s)}{\gamma_1\varphi^{F_1}(s)}]ds\\
&~-\int_T^t\frac{(\bar{\theta}a_1)^2\gamma_L\varphi^L(s)}
{\sigma_1^2\gamma_1\varphi^{F_1}(s)}[1+\frac{\gamma_L\varphi^L(s)}{2\gamma_1\varphi^{F_1}(s)}]ds,\\
g^{0Lb2}_2(t)&=g(t)+\frac{\gamma_L^2\sigma_1^2}{4(A_L+\eta_L)}[(\varphi^L(t))^2-1]
+\int_T^t\theta_1a_1\gamma_L\varphi^L(s)[1+\frac{\gamma_L\varphi^L(s)}{\gamma_1\varphi^{F_1}(s)}]ds\\
&~-\int_T^t\frac{(\theta_1a_1)^2\gamma_L\varphi^L(s)}
{\sigma_1^2\gamma_1\varphi^{F_1}(s)}[1+\frac{\gamma_L\varphi^L(s)}{2\gamma_1\varphi^{F_1}(s)}]ds,\\
g^{0Lb3}_2(t)&=g(t)+\frac{\gamma_L^2\sigma_1^2}{4(A_L+\eta_L)}[(\varphi^L(t))^2-1]+\frac{\sigma_1^2}{2}\int_T^t\gamma_L\varphi^L(s)
(\gamma_1\varphi^{F_1}(s)+\gamma_L\varphi^L(s))M^{F_1}(s)ds,\\
g^{0\tilde{F}_1a}_2(t)&=g(t)-\frac{\gamma_1\theta_1a_1}{A_1+\eta_1}[\varphi^{F_1}(t)-1]
+\frac{\gamma_1^2\sigma_1^2}{4(A_1+\eta_1)}[(\varphi^{F_1}(t))^2-1],\\
g^{0\tilde{F}_1b1}_2(t)&=g(t)+\frac{\gamma_1(\bar{\theta}-\theta_1)a_1}{A_1+\eta_1}[\varphi^{F_1}(t)-1]
-\frac{(\bar{\theta}a_1)^2}{2\sigma_1^2}(T-t),\\
g^{0\tilde{F}_1b2}_2(t)&=g(t)-\frac{(\theta_1a_1)^2}{2\sigma_1^2}(T-t),\\
g^{0\tilde{F}_1b3}_2(t)&=g(t)-\frac{\gamma_1\theta_1a_1}{A_1+\eta_1}[\varphi^{F_1}(t)-1]
-\gamma_1^2\sigma_1^2\int_T^t(\varphi^{F_1}(s))^2M^{F_1}(s)ds
+\frac{1}{2}\gamma_1^2\sigma_1^2\int_T^t(\varphi^{F_1}(s))^2(M^{F_1}(s))^2ds.
\end{align*}

Similar to \cite{Chen2018}, we can get that
when the equilibrium is achieved in the interior case (i.e., Case (4) in Table \ref{Tablestrategy0}), the optimal reinsurance premium follows the variance premium principle. In other words, for every one unit of risk, the total instantaneous reinsurance premium associated with the ceded proportion $(1-\tilde{q}^{\ast}_1(t))100\%$ can be written as
\begin{align*}
\tilde{p}^{\ast}(t)(1-\tilde{q}^{\ast}_1(t))
=a_1(1-\tilde{q}^{\ast}_1(t))+[\gamma_1\varphi^{F_1}(t)+\gamma_L\varphi^L(t)]\sigma_1^2
(1-\tilde{q}^{\ast}_1(t))^2,
\end{align*}
where the first term accounts for the mean component, and the second for the variance component.

\begin{corollary}\label{corollary4}
If $\beta\geq 0$, some properties of the optimal investment strategies (i.e., $\tilde{b}^{\ast}_L(t),\tilde{b}^{\ast}_1(t)$), optimal premium strategy (i.e., $\tilde{p}^{\ast}(t)$) and optimal reinsurance strategy (i.e., $\tilde{q}^{\ast}_1(t)$) of Case 4 in Table \ref{Tablestrategy0} are given in Table \ref{Tablespecial}, \eqref{equ:balpha0}, \eqref{equ:pqalpha0}, \eqref{equ:beta0} and \eqref{equ:pqeta0}.
\begin{table}[h]
\small
\caption{The properties of $(\tilde{p}^{\ast}(t),\tilde{b}^{\ast}_L(t),\tilde{q}^{\ast}_1(t),\tilde{b}^{\ast}_1(t))$.\\}
\centering
\renewcommand\arraystretch{1}
\begin{tabular}{p{1.1cm}<{\centering}|p{1.1cm}<{\centering}|p{1.1cm}<{\centering}|
p{1.1cm}<{\centering}|p{1.1cm}<{\centering}|p{1.1cm}<{\centering}
|p{1.1cm}<{\centering}|p{1.1cm}<{\centering}}
$\frac{\partial \tilde{b}_L^{\ast}(t)}{\partial\gamma_L}$
& $\frac{\partial \tilde{b}_L^{\ast}(t)}{\partial h_L}$
& $\frac{\partial \tilde{b}_1^{\ast}(t)}{\partial \gamma_1}$
& $\frac{\partial \tilde{b}_1^{\ast}(t)}{\partial h_1}$
& $\frac{\partial \tilde{p}^{\ast}(t)}{\partial \gamma_L}$
& $\frac{\partial \tilde{p}^{\ast}(t)}{\partial h_L}$
& $\frac{\partial \tilde{q}_1^{\ast}(t)}{\partial \gamma_1}$
& $\frac{\partial \tilde{q}_1^{\ast}(t)}{\partial h_1}$\\
\hline
$-$&$-$&$-$&$-$&$+$&$+$&$-$&$-$\\
\end{tabular}
\label{Tablespecial}
\end{table}

\begin{footnotesize}
\begin{equation}\label{equ:balpha0}
\frac{\partial \tilde{b}_L^{\ast}(t)}{\partial\alpha_L}=\left\{
\begin{array}{rcl}
>0, &  {\alpha_L>-\frac{1}{h_L}ln\frac{1}{h_L};}\\
=0,&  {\alpha_L=-\frac{1}{h_L}ln\frac{1}{h_L};}\\
<0, &  {\alpha_L<-\frac{1}{h_L}ln\frac{1}{h_L}.}
\end{array} \right.
\frac{\partial \tilde{b}_1^{\ast}(t)}{\partial\alpha_1}=\left\{
\begin{array}{rcl}
>0, &  {\alpha_1>-\frac{1}{h_1}ln\frac{1}{h_1};}\\
=0,&  {\alpha_1=-\frac{1}{h_1}ln\frac{1}{h_1};}\\
<0, &  {\alpha_1<-\frac{1}{h_1}ln\frac{1}{h_1}.}
\end{array} \right.
\end{equation}
\end{footnotesize}

\begin{footnotesize}
\begin{equation}\label{equ:pqalpha0}
\frac{\partial \tilde{p}^{\ast}(t)}{\partial\alpha_L}=\left\{
\begin{array}{rcl}
<0, &  {\alpha_L>-\frac{1}{h_L}ln\frac{1}{h_L};}\\
=0,&  {\alpha_L=-\frac{1}{h_L}ln\frac{1}{h_L};}\\
>0, &  {\alpha_L<-\frac{1}{h_L}ln\frac{1}{h_L}.}
\end{array} \right.
\frac{\partial \tilde{q}_1^{\ast}(t)}{\partial\alpha_1}=\left\{
\begin{array}{rcl}
>0, &  {\alpha_1>-\frac{1}{h_1}ln\frac{1}{h_1};}\\
=0,&  {\alpha_1=-\frac{1}{h_1}ln\frac{1}{h_1};}\\
<0, &  {\alpha_1<-\frac{1}{h_1}ln\frac{1}{h_1}.}
\end{array} \right.
\end{equation}
\end{footnotesize}

If $r_0+\alpha_L<1$ and $r_0+\alpha_1<1$, then
\begin{footnotesize}
\begin{equation}\label{equ:beta0}
\frac{\partial \tilde{b}_L^{\ast}(t)}{\partial\eta_L}=\left\{
\begin{array}{rcl}
>0, &  {h_L<-\frac{1}{\alpha_L}ln(1-r_0-\alpha_L);}\\
=0,&  {h_L=-\frac{1}{\alpha_L}ln(1-r_0-\alpha_L);}\\
<0, &  {h_L>-\frac{1}{\alpha_L}ln(1-r_0-\alpha_L).}
\end{array} \right.
\frac{\partial \tilde{b}_1^{\ast}(t)}{\partial\eta_1}=\left\{
\begin{array}{rcl}
 >0,&  {h_1<-\frac{1}{\alpha_1}ln(1-r_0-\alpha_1);}\\
 =0,&  {h_1=-\frac{1}{\alpha_1}ln(1-r_0-\alpha_1);}\\
 <0,&  {h_1>-\frac{1}{\alpha_1}ln(1-r_0-\alpha_1).}
\end{array} \right.
\end{equation}
\end{footnotesize}

\begin{footnotesize}
\begin{equation}\label{equ:pqeta0}
\frac{\partial \tilde{p}^{\ast}(t)}{\partial\eta_L}=\left\{
\begin{array}{rcl}
<0, &  {h_L<-\frac{1}{\alpha_L}ln(1-r_0-\alpha_L);}\\
=0,&  {h_L=-\frac{1}{\alpha_L}ln(1-r_0-\alpha_L);}\\
>0, &  {h_L>-\frac{1}{\alpha_L}ln(1-r_0-\alpha_L).}
\end{array} \right.
\frac{\partial \tilde{q}_1^{\ast}(t)}{\partial\eta_1}=\left\{
\begin{array}{rcl}
 >0,&  {h_1<-\frac{1}{\alpha_1}ln(1-r_0-\alpha_1);}\\
 =0,&  {h_1=-\frac{1}{\alpha_1}ln(1-r_0-\alpha_1);}\\
 <0,&  {h_1>-\frac{1}{\alpha_1}ln(1-r_0-\alpha_1).}
\end{array} \right.
\end{equation}
\end{footnotesize}
\end{corollary}

\begin{proof}
The proof of this corollary is similar to that of Corollary \ref{corollary2}.
\end{proof}

Further, if the reinsurer and the insurer do not consider the effect of the delay, i.e., $\eta_L=\eta_1=h_L=h_1=\alpha_L=\alpha_1=0$, then $B_L=B_1=C_L=C_1=0$ and $A_L=A_1=r_0$.
The optimal reinsurance premium and the optimal reinsurance strategy are the same as that in the case of $\rho_L=\rho_F$ in the literature \cite{Chen2018}.

\section{Sensitivity analysis}\label{section 4}

To illustrate the sensitivities of the equilibrium strategy $(p^{\ast}(\cdot),b^{\ast}_L(\cdot);q_1^{\ast}(\cdot),b_1^{\ast}(\cdot);q_2^{\ast}(\cdot),b_2^{\ast}(\cdot))$ with respect to the model parameters, we conduct numerical experiments in this section.
Throughout this section, unless stated otherwise, the basic model parameters are given in Table \ref{Tablefinance}, Table \ref{Tablereinsurer} and Table\ref{Tableinsurer}.\footnote{For $i\neq j\in\{1,2\}$, we can get that
$A_i=\frac{1}{1+\eta_i}[r_0-(\alpha_i+\eta_i)\eta_i-\eta_ie^{-\alpha_ih_i}]$ due to  $A_i=r_0-B_i-C_i$, $C_i=\eta_ie^{-\alpha_ih_i}$ and $B_ie^{-\alpha_ih_i}=(\alpha_i+A_i+\eta_i)C_i$.
From the condition in Theorem \ref{Theorem1}, $A_1+\eta_1=A_2+\eta_2$,
we can get that $\eta_j=\frac{(r_0-1+e^{-\alpha_ih_i}+\alpha_i)\eta_i}{(r_0-1+e^{-\alpha_jh_j}+\alpha_j)
+(\alpha_j-\alpha_i+e^{-\alpha_jh_j}-e^{-\alpha_ih_i})\eta_i}$.}
\begin{table}[htp]%金融市场参数
\setlength{\abovecaptionskip}{0.cm}
\caption{The parameter values of the financial assets.\\}
\centering
\renewcommand\arraystretch{1}
\begin{tabular}
{p{0.01cm}p{1.2cm}<{\centering}p{1.2cm}<{\centering}p{1.2cm}<{\centering}p{1.2cm}<{\centering}p{1.2cm}<{\centering}p{1.2cm}<{\centering}}
%列数p{1.45cm}<{\centering}宽度居中，表示各列元素对齐方式，left-l,right-r,center-c
\toprule
\multirow{1}*{}
& $r_0$ & $r$ & $\sigma$ & $\beta$ & $s_0$ & $T$\\
\midrule
&\tabincell{c}{$0.05$} &\tabincell{c}{$0.1$}
&\tabincell{c}{$0.4$} &\tabincell{c}{$1$} &\tabincell{c}{$1$} &\tabincell{c}{$10$}\\
\bottomrule
\end{tabular}
\label{Tablefinance}
\end{table}
\begin{table}[htp]%再保险公司参数
\setlength{\abovecaptionskip}{0.cm}
\caption{The parameter values of the reinsurer.\\}
\centering
\renewcommand\arraystretch{1}
\begin{tabular}
{p{0.01cm}p{1.2cm}<{\centering}p{1.2cm}<{\centering}p{1.2cm}<{\centering}
p{1.2cm}<{\centering}p{1.2cm}<{\centering}}
%列数p{1.45cm}<{\centering}宽度居中，表示各列元素对齐方式，left-l,right-r,center-c
\toprule
\multirow{1}*{}
& $\bar{\theta}$ & $h_L$ & $\alpha_L$ & $\eta_L$ & $\gamma_L$ \\
\midrule
&\tabincell{c}{$2$} &\tabincell{c}{$2$}
&\tabincell{c}{$0.3$} &\tabincell{c}{$0.05$} &\tabincell{c}{$0.1$} \\
\bottomrule
\end{tabular}
\label{Tablereinsurer}
\end{table}
\begin{table}[htp]%保险公司参数
\setlength{\abovecaptionskip}{0.cm}
\caption{The parameter values of insurers.\\}
\centering
\renewcommand\arraystretch{1}
\begin{tabular}
{p{0.00cm}p{2.5cm}<{\centering}p{2.5cm}<{\centering}p{2.5cm}<{\centering}p{2.5cm}<{\centering}}
%列数p{1.45cm}<{\centering}宽度居中，表示各列元素对齐方式，left-l,right-r,center-c
\toprule
& \multicolumn{2}{c}{The parameter values of insurer 1} & \multicolumn{2}{c}{The parameter values of insurer 2}\\
\cmidrule(lr){2-3}\cmidrule(lr){4-5}%\cmidrule(lr){2-5}实现2到5列的短划线
& Parameter & Value & Parameter & Value \\
\midrule
&\tabincell{c}{$\lambda_1$\\$\mu_1$\\$\sigma_1$\\$\theta_1$\\$h_1$\\$\alpha_1$\\$\eta_1$
\\$\gamma_1$\\$k_1$\\$\rho$} &\tabincell{c}{$0.8$\\$5$\\$3$\\$1.2$\\$2$\\$0.5$\\$0.05$\\$2$\\$0.4$\\$0.3$}
&\tabincell{c}{$\lambda_2$\\$\mu_2$\\$\sigma_2$\\$\theta_2$\\$h_2$\\$\alpha_2$\\$\eta_2$
\\$\gamma_2$\\$k_2$\\$/$}
&\tabincell{c}{$1$\\$4$\\$2$\\$1$\\$3$\\$0.3$\\$ / $\\$3$\\$0.3$\\$/$} \\
\bottomrule
\end{tabular}
\label{Tableinsurer}
\end{table}

\subsection{Sensitivity analysis of the equilibrium investment strategy}

Figure \ref{fig:bdelayt} shows the change in the risky asset price over time and the optimal investment strategies with and without delay over time.
From Theorem \ref{Theorem1}, we can find that $\frac{\partial b_L^{\ast}(t)}{\partial s}=-\frac{2\beta}{s} b_L^{\ast}(t)$ and $\frac{\partial b_i^{\ast}(t)}{\partial s}=-\frac{2\beta}{s} b_i^{\ast}(t), i\in\{1,2\}$.
It is easy to find that $\frac{\partial b_L^{\ast}(t)}{\partial s}<0$ and $\frac{\partial b_i^{\ast}(t)}{\partial s}<0$, when $\beta=1$.
That is, the wealth invested in the risky asset is negatively correlated with the price of the risky asset, which is consistent with the trend of curves in Figure \ref{fig:bdelayt}.
Under the setting of parameters in Table \ref{Tablefinance}, Table \ref{Tablereinsurer} and Table \ref{Tableinsurer}, we have $h_L\geq-\frac{1}{\alpha_L}ln(1-r_0-\alpha_L)$ and $h_i\geq-\frac{1}{\alpha_i}ln(1-r_0-\alpha_i),i\in\{1,2\}$.
According to Special case 1, the amount invested in the risky asset with delay is lower than that without delay, which is consistent with Figure \ref{fig:bdelayt}.
That is, the delay factor will urge the investor to shrink the position invested in the risky asset and make the investment strategy more conservative when the delay time considered is greater than a certain value.
Furthermore, the gap between the investment strategy with delay and the investment strategy without delay will decrease with the increase of time $t$. And they completely coincide at terminal time $T$.

\begin{figure}[htp]
\begin{center}
  \includegraphics[width=7in]{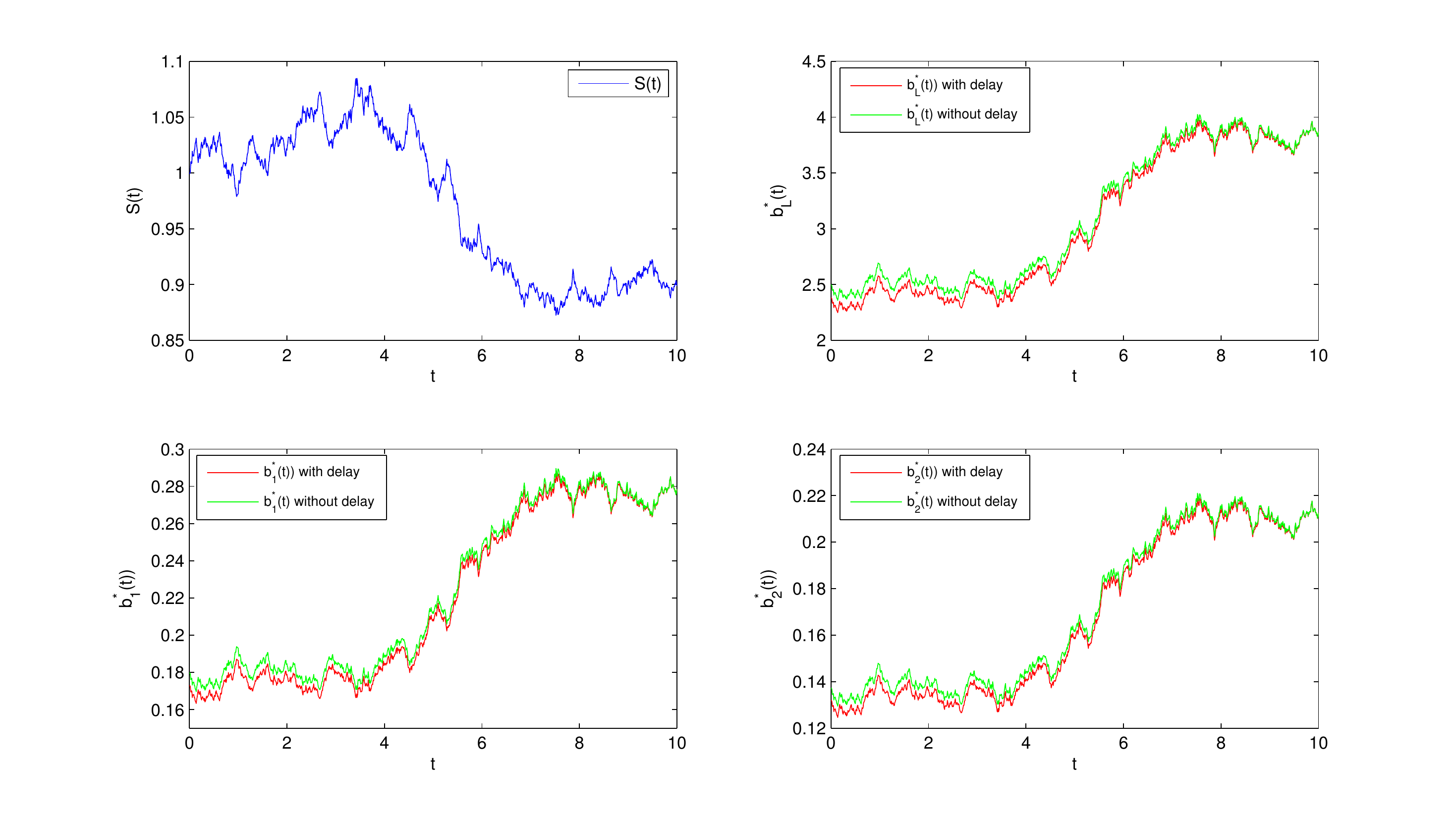}\\
  \caption{$b_L^{\ast}(t)$, $b_1^{\ast}(t)$ and $b_2^{\ast}(t)$ with and without delay.}
  \label{fig:bdelayt}
  \end{center}
\end{figure}

\begin{figure}
  \centering
  \subfigure[$\beta<0$]{
    %\label{fig:investmentbeta0}%% label for first subfigure
    \includegraphics[width=2.8in]{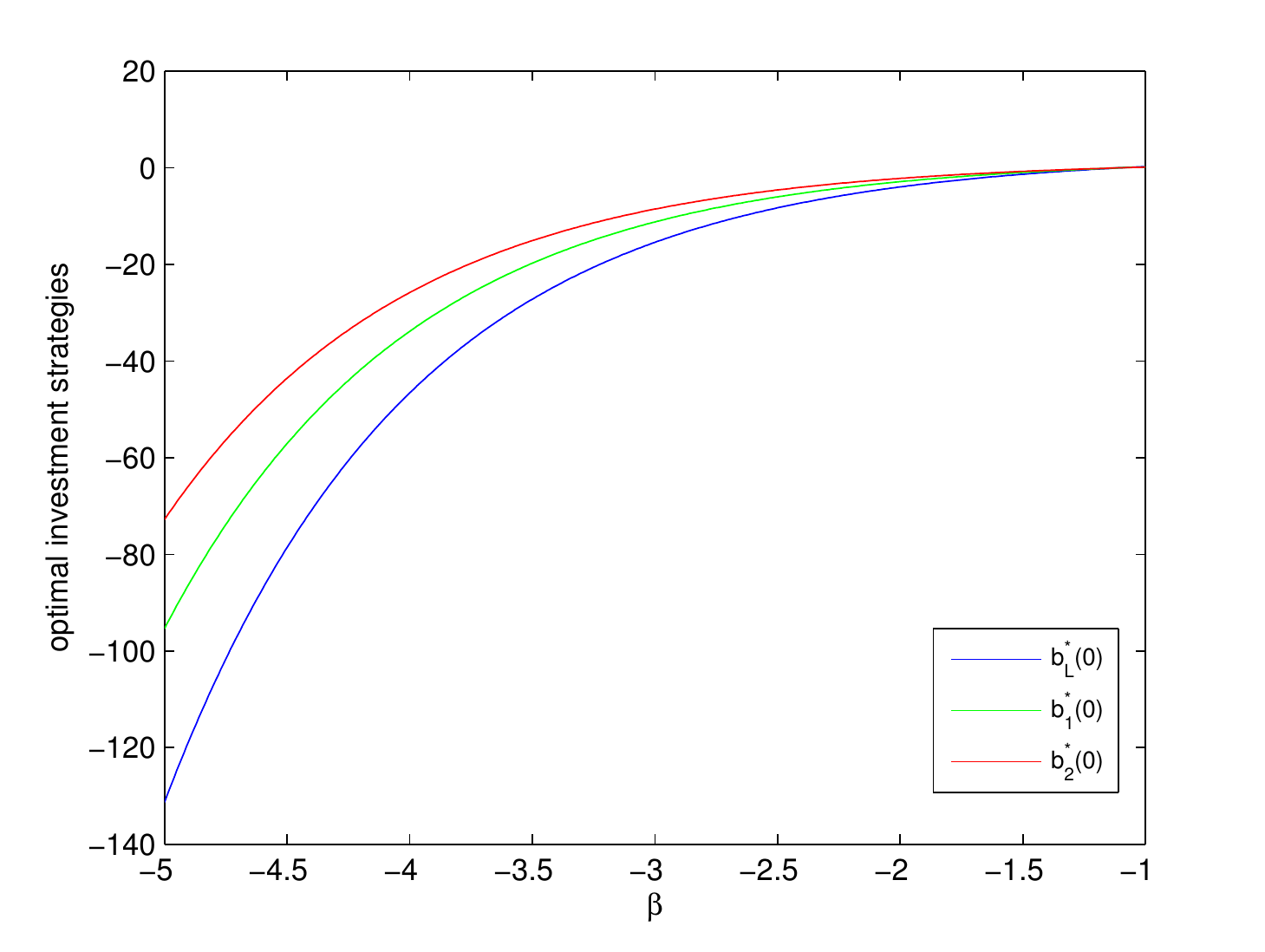}
  }
  \subfigure[$\beta>0$]{
    %\label{fig:investmentbeta} %% label for second subfigure
    \includegraphics[width=2.8in]{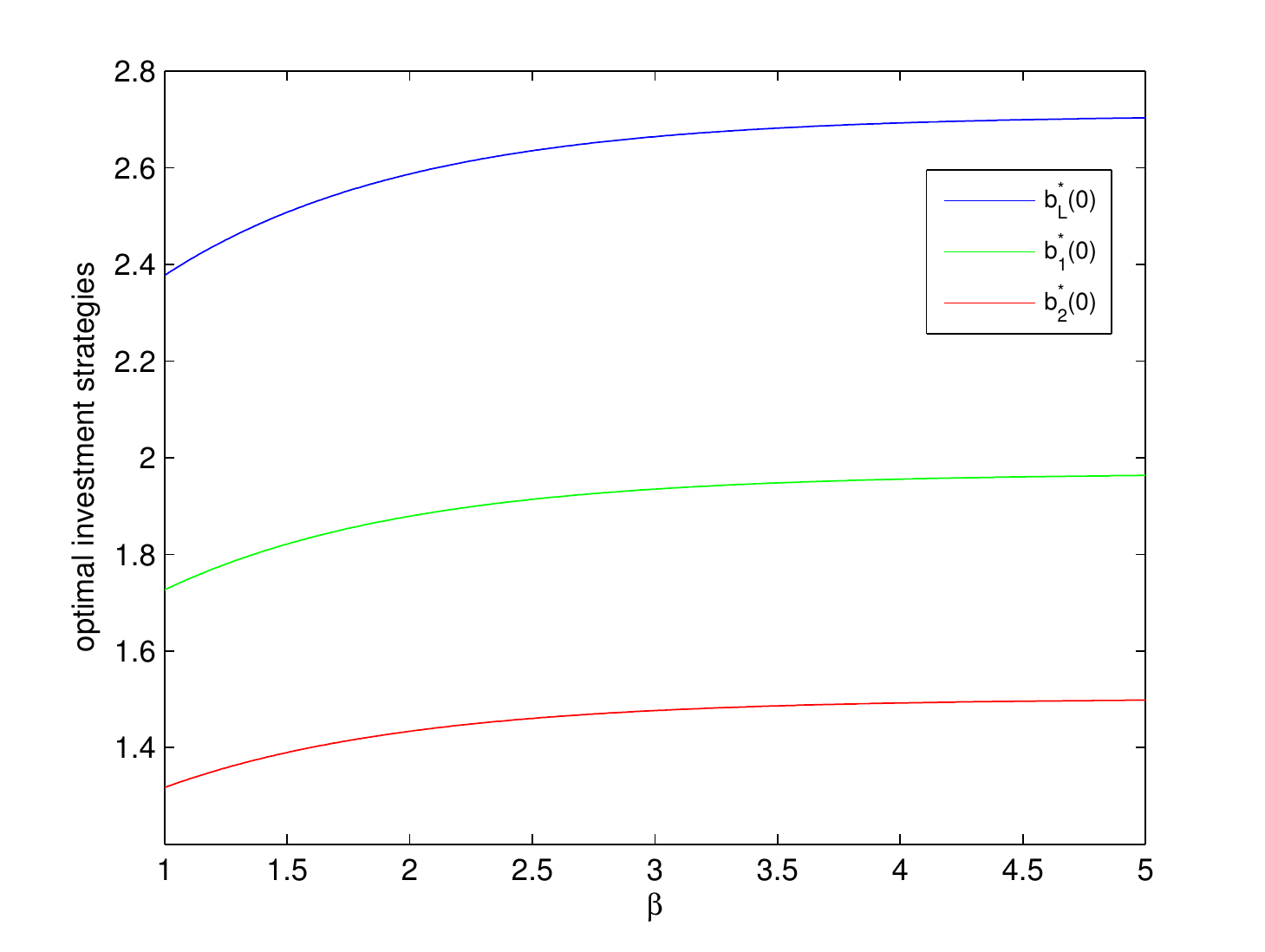}
  }
  \caption{Effects of $\beta$ on $b_L^{\ast}(0)$, $b_1^{\ast}(0)$ and $b_2^{\ast}(0)$.}
  \label{fig:bbeta} %% label for entire figure
\end{figure}

Figure \ref{fig:bbeta} indicates the influence of the constant elasticity parameter $\beta$ on optimal investment strategies at the initial moment, including $\beta<0$ and $\beta>0$.
In Figure \ref{fig:bbeta}, we can note that both the reinsurer's and insurers' investment strategies will increase as $\beta$ increases.
The investment amount is negative when elasticity parameter $\beta<0$, and the investment amount is positive when the elasticity parameter $\beta>0$.
In other words, positive elasticity parameter results in a positive hedging demand; the hedging demand is negative for negative elasticity parameter, which is consistent with the description in Section \ref{subsection3.1}.

Figure \ref{fig:bgammak} indicates the impacts of the risk aversion coefficients (i.e., $\gamma_L$, $\gamma_1$, $\gamma_2$) and sensitivity coefficients (i.e., $k_1$, $k_2$) on optimal investment strategies at the initial moment.
Both \ref{fig:investmentgammaL}, \ref{fig:investmentgamma1} and \ref{fig:investmentgamma2} show that the greater the risk aversion coefficient, the less the amount invested in the risky asset, which is consistent with the actual situation.
Both \ref{fig:investmentgamma1} and \ref{fig:investmentgamma2} show that for insurer $i$ ($i\in\{1,2\}$), the greater the sensitivity coefficient $k_i$, the more the insurer $i$ invests in the risky asset.
Because the sensitivity coefficient $k_i$ reflects the degree to which insurer $i$ cares about the terminal wealth of its competitor (i.e., insurer $j$, $j\neq i\in\{1,2\}$), the larger the $k_i$, the more the insurer $i$ cares about the performance of its opponent.
Therefore, when the sensitivity coefficient $k_i$ is larger, insurer $i$ is more inclined to invest more money into the risky asset for increasing its wealth.

\begin{figure}
  \centering
  \subfigure[Effect of $\gamma_L$ on $b_L^{\ast}(0)$]{
    \label{fig:investmentgammaL}%% label for first subfigure
    \begin{minipage}{5cm}
      \centering
      \includegraphics[width=5cm]{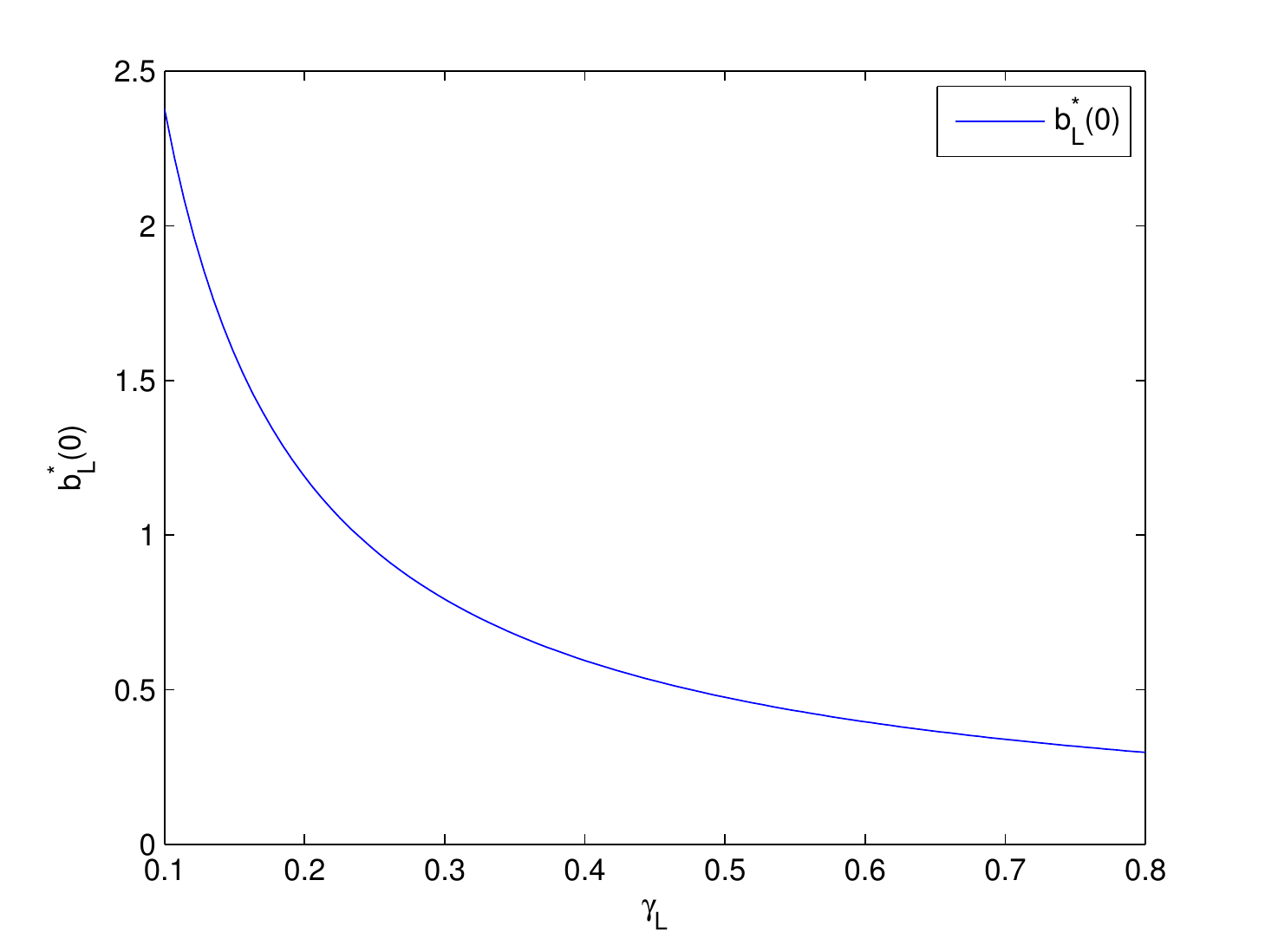}
    \end{minipage}
  }
  \subfigure[Effects of $\gamma_1$ and $k_1$ on $b_1^{\ast}(0)$]{
    \label{fig:investmentgamma1} %% label for second subfigure
    \begin{minipage}{5cm}
      \centering
    \includegraphics[width=5cm]{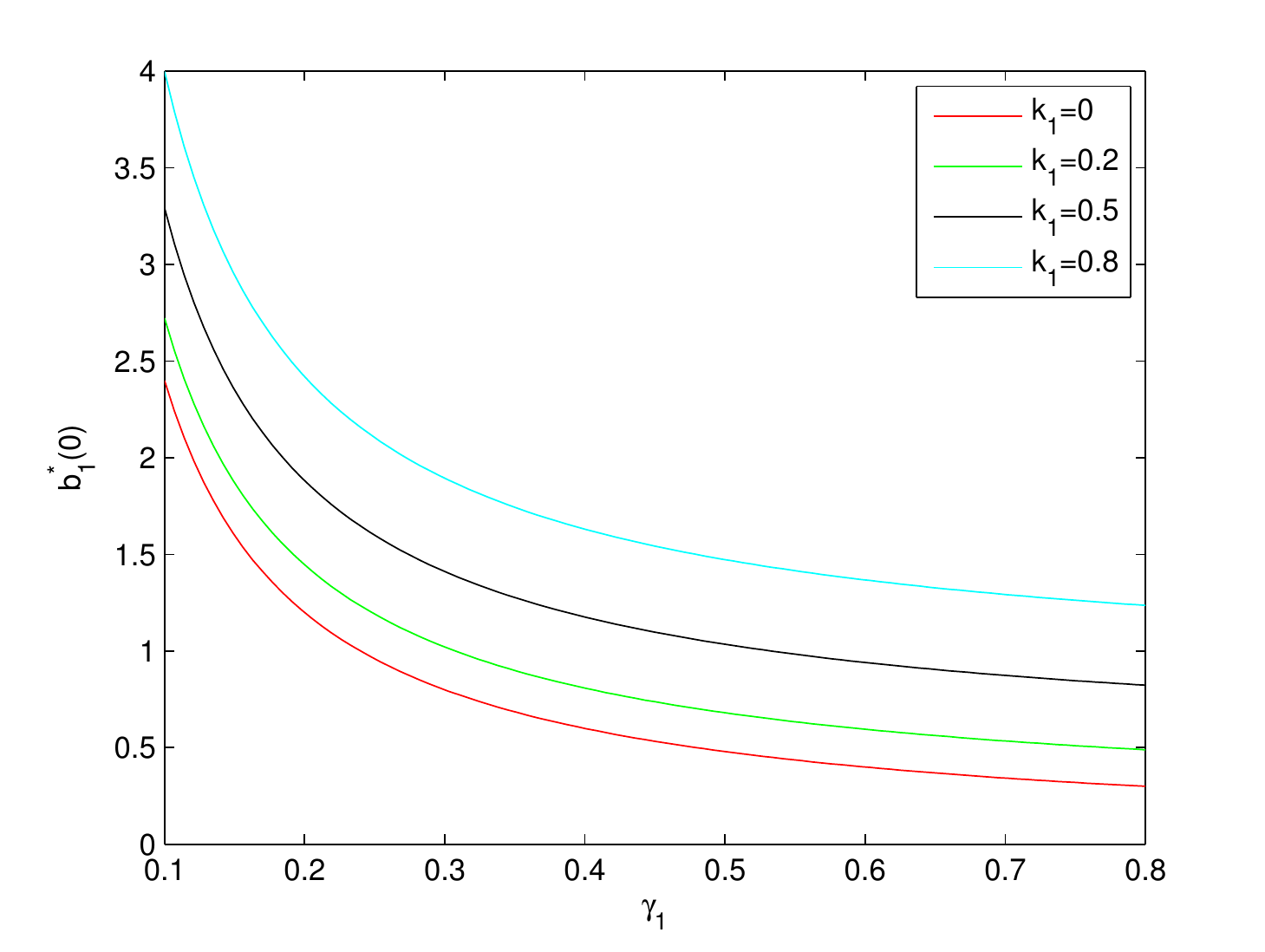}
    \end{minipage}
  }
  \subfigure[Effects of $\gamma_2$ and $k_2$ on $b_2^{\ast}(0)$]{
    \label{fig:investmentgamma2} %% label for third subfigure
     \begin{minipage}{5cm}
      \centering
    \includegraphics[width=5cm]{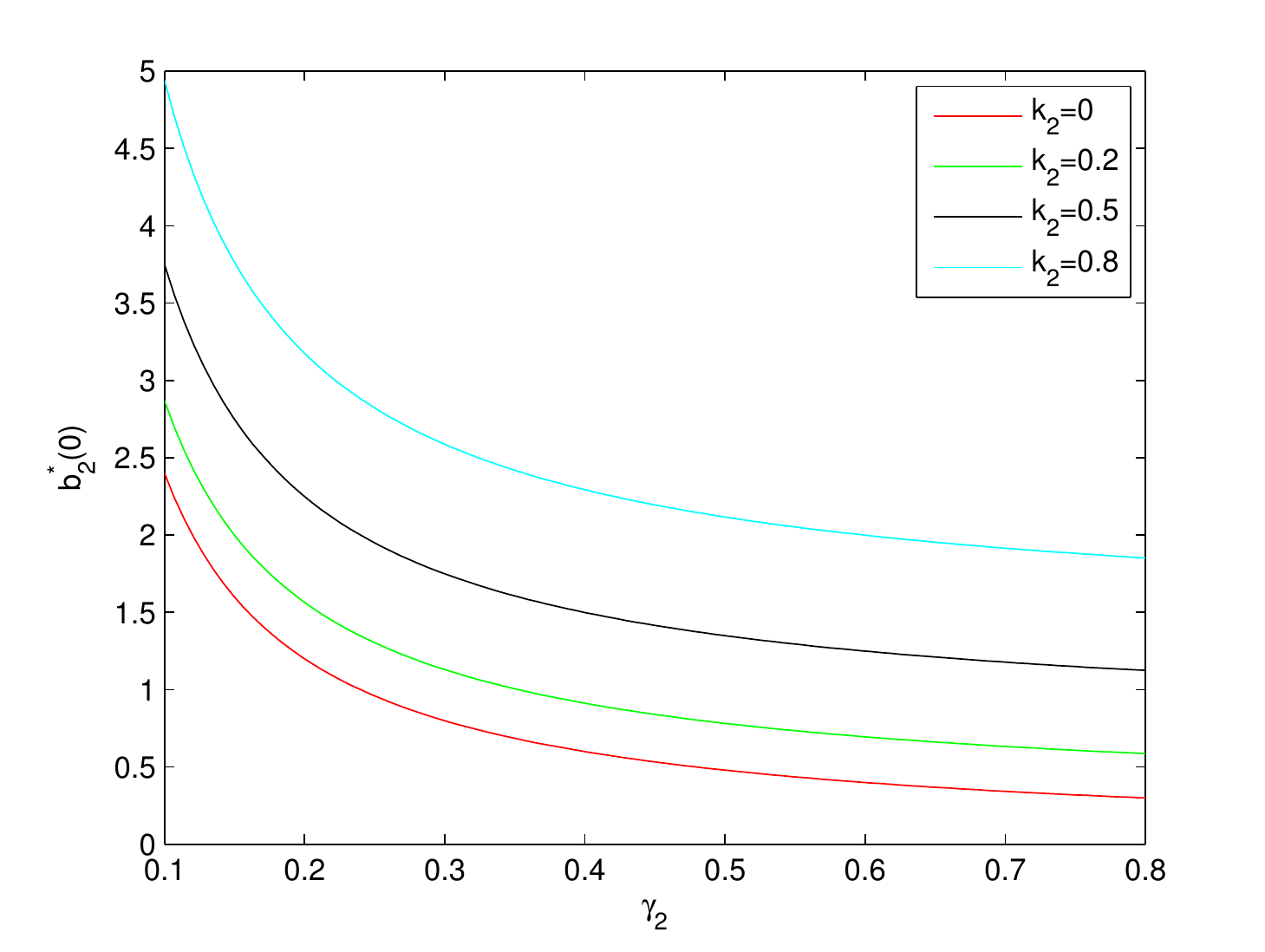}
    \end{minipage}
  }
  \caption{Effects of risk aversion coefficients  and sensitivity coefficients on optimal investment strategies.}
  \label{fig:bgammak} %% label for entire figure
\end{figure}

\begin{figure}
  \centering
  \subfigure[Effects of $\eta_L$ and $h_L$ on $b_L^{\ast}(0)$]{
    \label{fig:betaLhL}%% label for first subfigure
    \begin{minipage}{5cm}
      \centering
    \includegraphics[width=5cm]{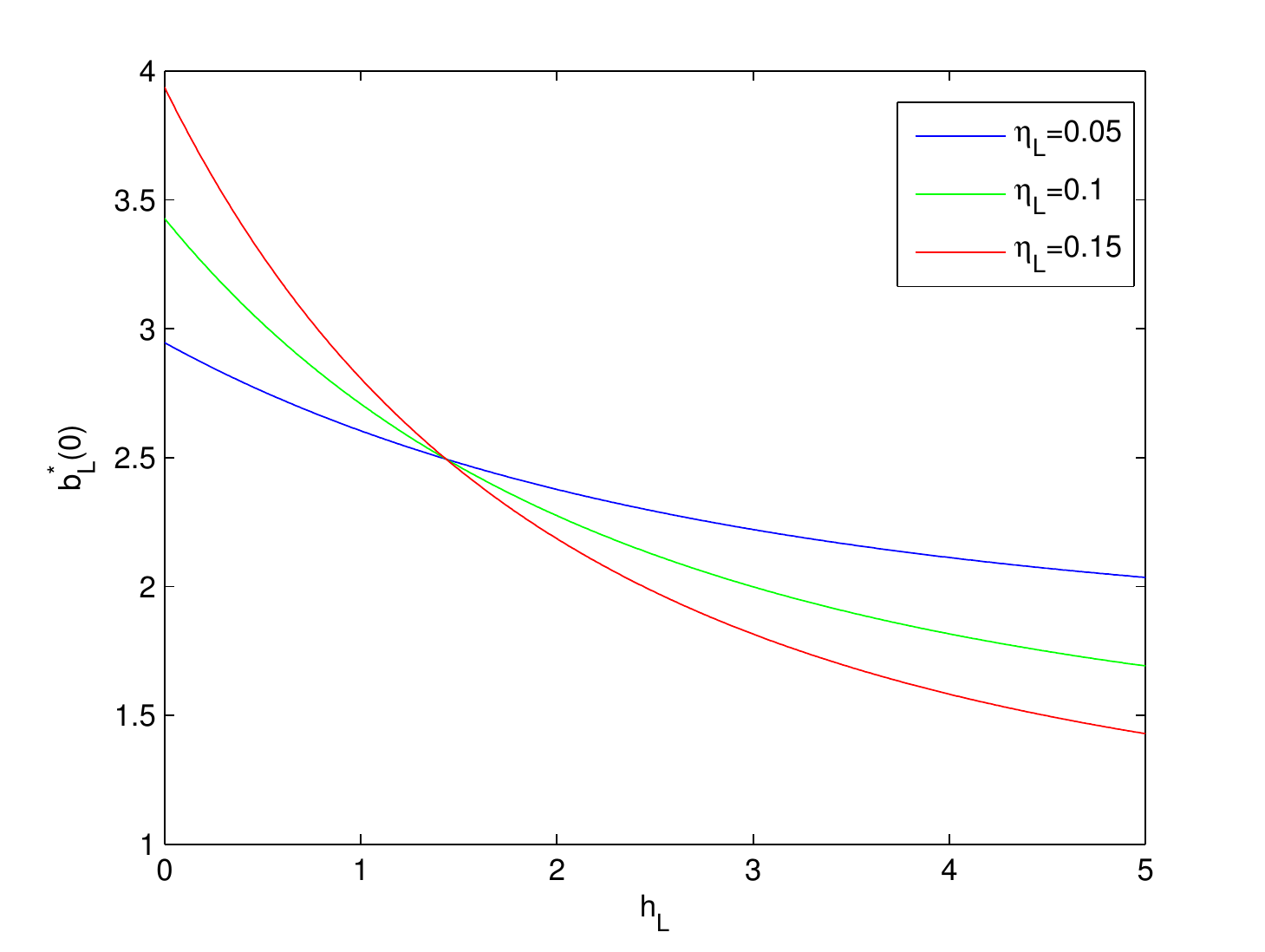}
    \end{minipage}
  }
  \subfigure[Effects of $\eta_1$ and $h_1$ on $b_1^{\ast}(0)$]{
    \label{fig:beta1h1} %% label for second subfigure
    \begin{minipage}{5cm}
      \centering
    \includegraphics[width=5cm]{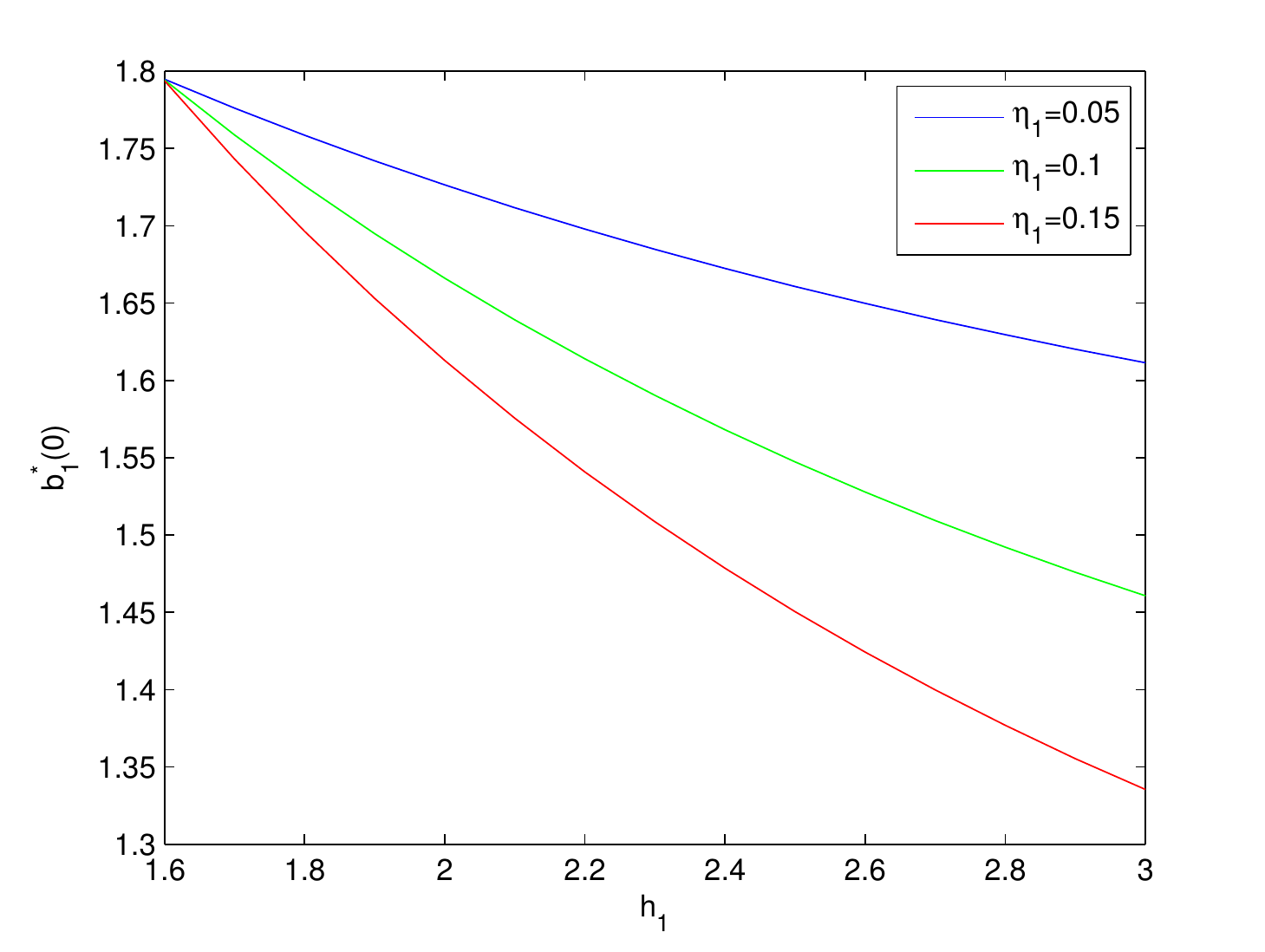}
    \end{minipage}
  }
  \subfigure[Effects of $\eta_2$ and $h_2$ on $b_2^{\ast}(0)$]{
    \label{fig:beta2h2} %% label for third subfigure
    \begin{minipage}{5cm}
      \centering
    \includegraphics[width=5cm]{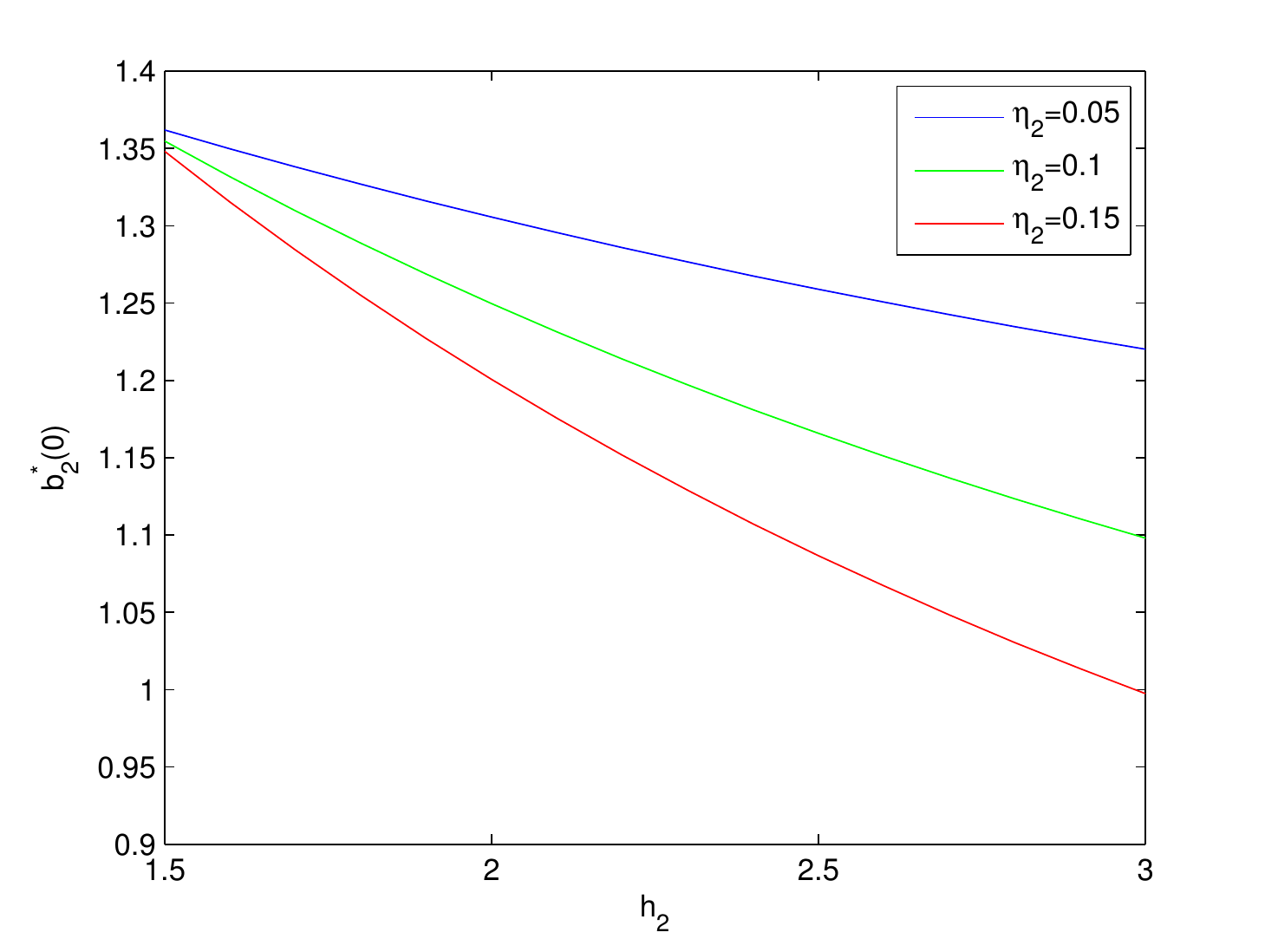}
    \end{minipage}
    }
                \quad           %这个回车键很重要 \quad也可以
    \subfigure[Effect of $\alpha_L$ on $b_L^{\ast}(0)$]{
    \label{fig:bLalphaL}%% label for first subfigure
    \begin{minipage}{5cm}
      \centering
    \includegraphics[width=5cm]{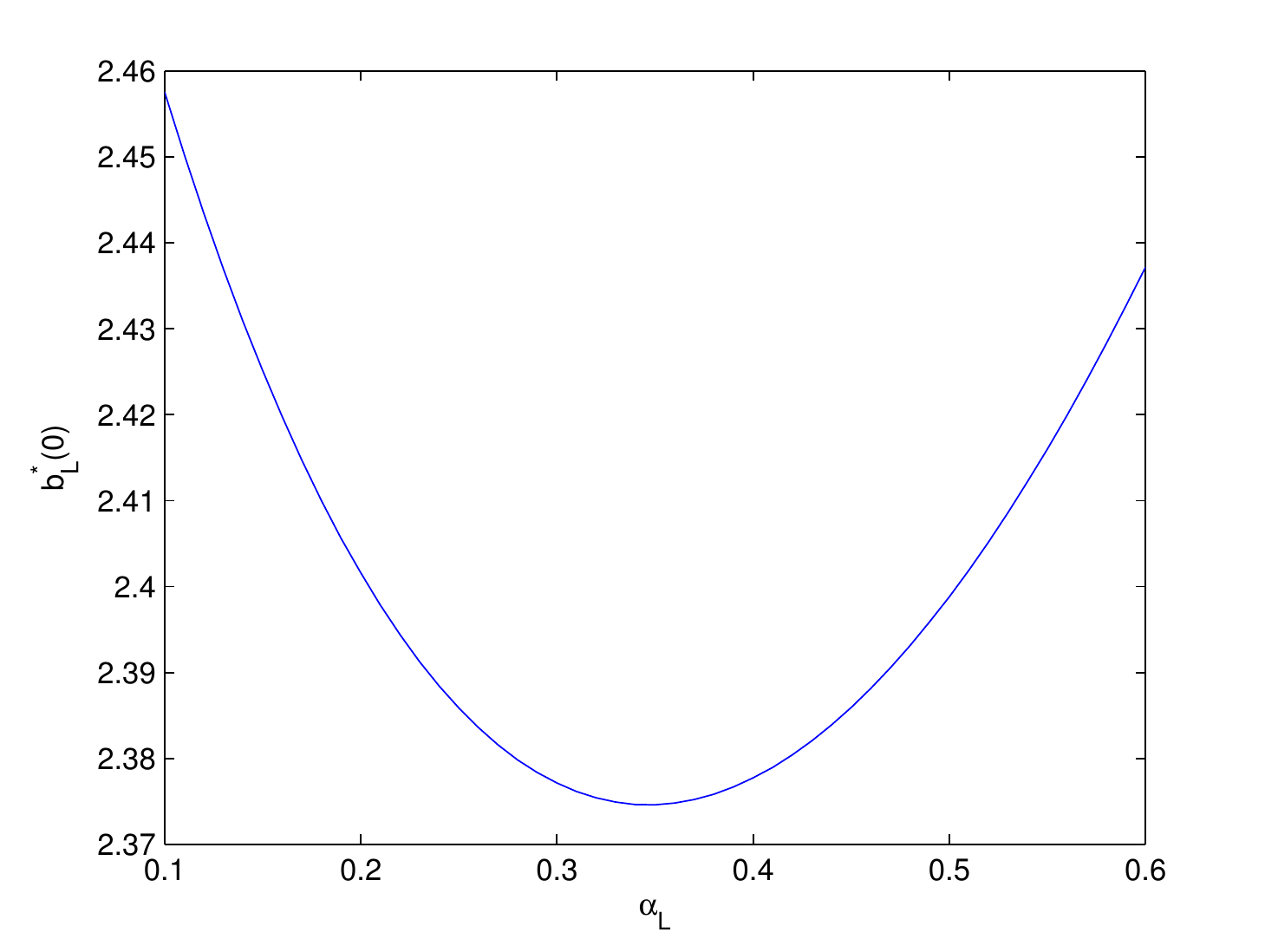}
    \end{minipage}
  }
  \subfigure[Effect of $\alpha_1$ on $b_1^{\ast}(0)$]{
    \label{fig:b1alpha1} %% label for second subfigure
    \begin{minipage}{5cm}
      \centering
    \includegraphics[width=5cm]{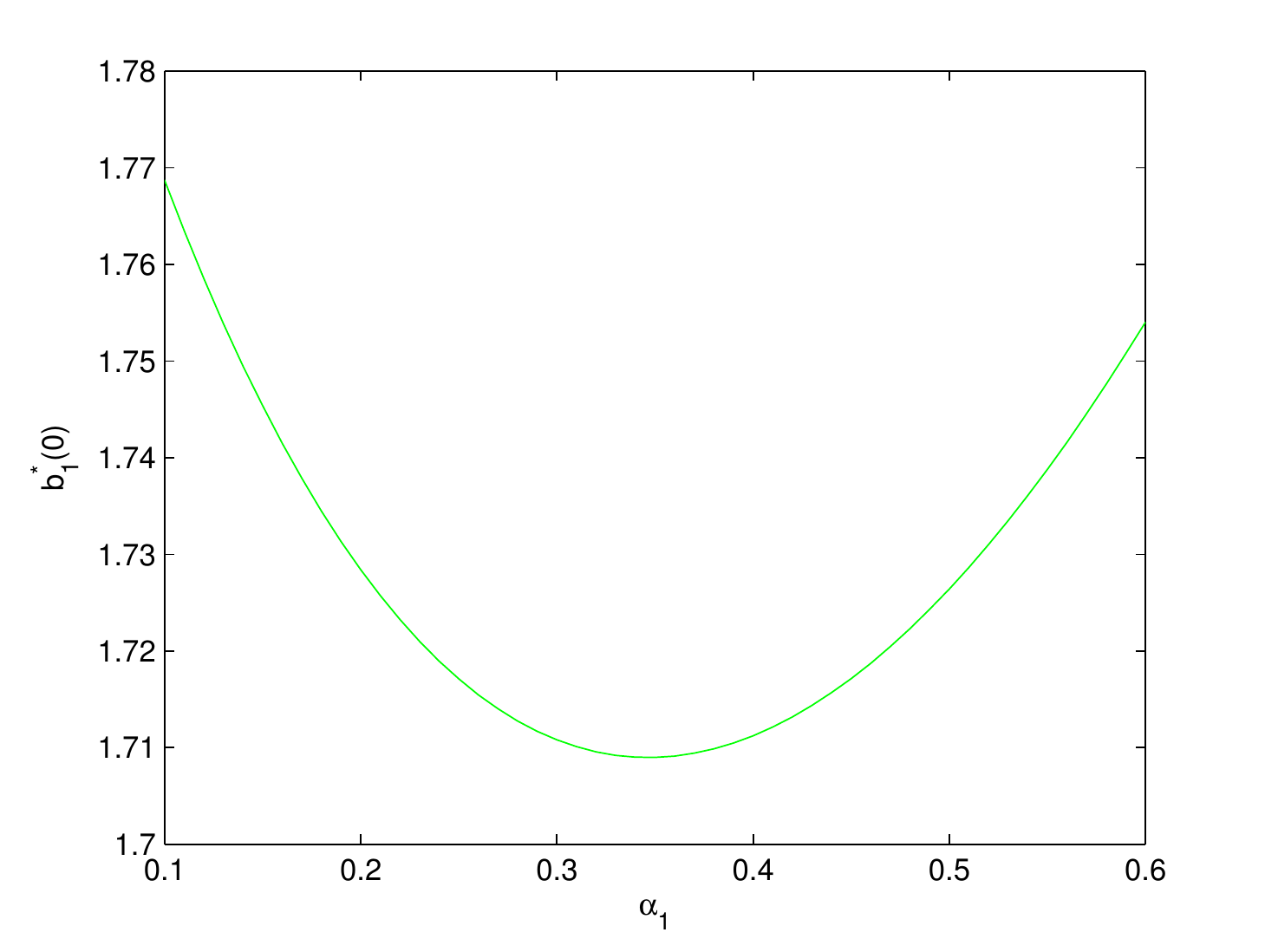}
    \end{minipage}
  }
  \subfigure[Effect of $\alpha_2$ on $b_2^{\ast}(0)$]{
    \label{fig:b2alpha2} %% label for third subfigure
    \begin{minipage}{5cm}
      \centering
    \includegraphics[width=5cm]{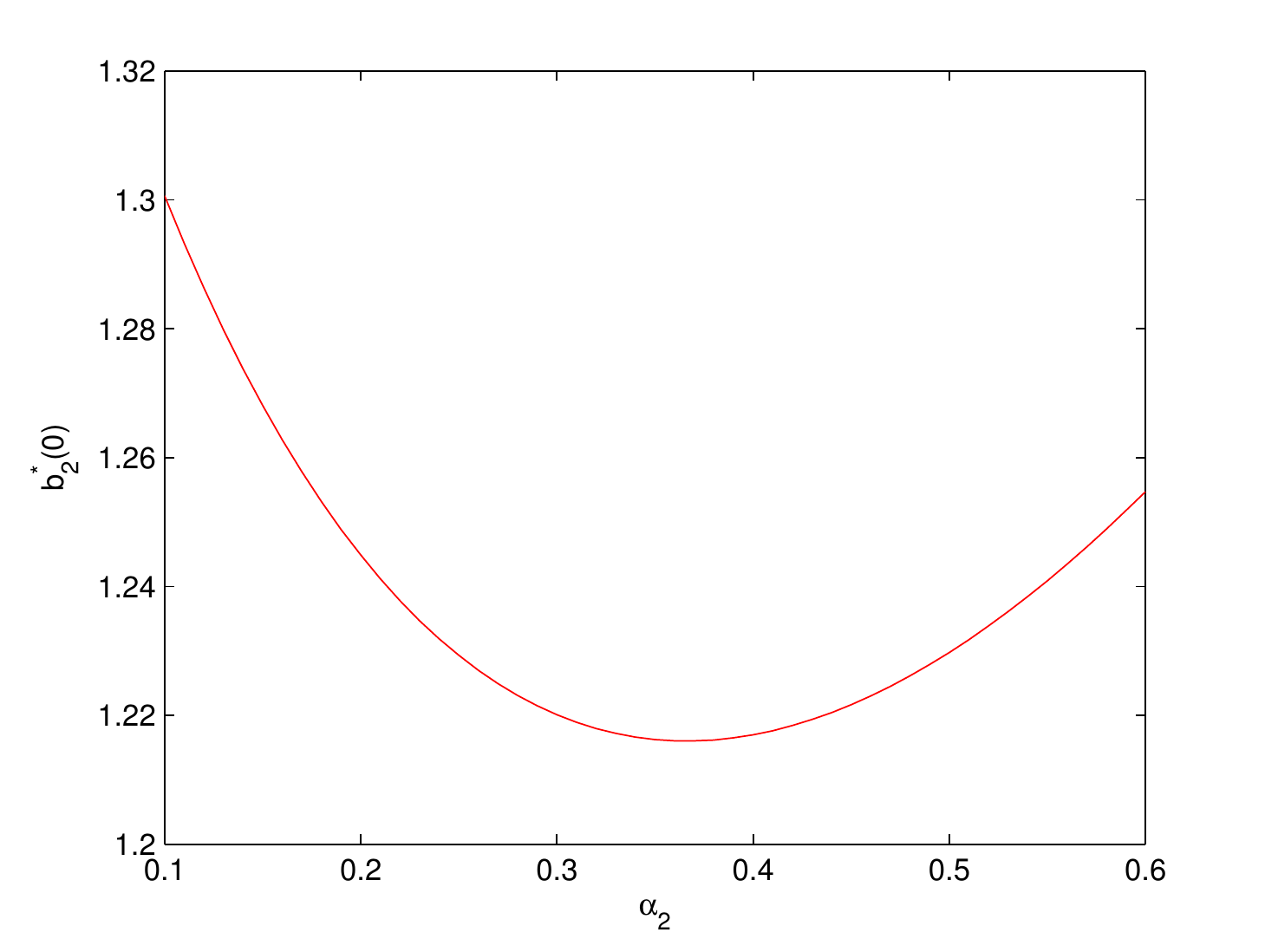}
    \end{minipage}
  }
  \caption{Effects of delay parameters on optimal investment strategies.}
  \label{fig:bdelay} %% label for entire figure
\end{figure}

Figure \ref{fig:bdelay} indicates the effects of delay parameters (i.e., $h_L$, $\eta_L$, $\alpha_L$, $h_1$, $\eta_1$, $\alpha_1$, $h_2$, $\eta_2$ and $\alpha_2$) on optimal investment strategies (i.e., $b_L^{\ast}(0)$, $b_1^{\ast}(0)$ and $b_2^{\ast}(0)$), respectively.
Both \ref{fig:betaLhL}, \ref{fig:beta1h1} and \ref{fig:beta2h2} show that when the delay weight defined, the longer the delay time, the less money is invested in the risky asset.
That is, when the delay time is longer, the reinsurer and the insurers adopt more conservative and robust investment strategies to control their risk.
As can be seen from \ref{fig:betaLhL}, for the reinsurer, when the delay time is within a certain range, the greater the delay weight, the more money is invested in the risky asset. Otherwise, the opposite situation occurs.
From the range of abscissa in \ref{fig:beta1h1} and \ref{fig:beta2h2}, we can know that,
$h_1>-\frac{1}{\alpha_1}ln(1-r_0-\alpha_1)=1.5970$ and
$h_2>-\frac{1}{\alpha_2}ln(1-r_0-\alpha_2)=1.4359$.
\footnote{Note: Because $\eta_j=\frac{(r_0-1+e^{-\alpha_ih_i}+\alpha_i)\eta_i}{(r_0-1+e^{-\alpha_jh_j}+\alpha_j) +(\alpha_j-\alpha_i+e^{-\alpha_jh_j}-e^{-\alpha_ih_i})\eta_i}$ for $i\neq j\in\{1,2\}$, the selection of abscissa, delay time parameters and the delay weight parameters should ensure that $\eta_i\in(0,1),i\in\{1,2\}$.}
For insurer $i$, $i\in\{1,2\}$, when the delay time is greater than the certain value, the greater the delay weight $\eta_i$ is, the less insurer $i$ invests in the risky asset, i.e., $\frac{\partial b_i^{\ast}(t)}{\partial\eta_i}<0$, which is consistent with the case of $h_i>-\frac{1}{\alpha_i}ln(1-r_0-\alpha_i)$ in equation \eqref{equ:beta}.
Subfigures \ref{fig:bLalphaL}, \ref{fig:b1alpha1} and \ref{fig:b2alpha2} show the effects of $\alpha_L$, $\alpha_1$ and $\alpha_2$ on $b_L^{\ast}(0)$, $b_1^{\ast}(0)$ and $b_2^{\ast}(0)$, respectively.
These three subgraphs illustrate that the impact of the average parameter on the optimal investment strategy will change with its size, which is consistent with equation \eqref{equ:balpha}.

\subsection{Sensitivity analysis of the equilibrium reinsurance strategy}

\begin{figure}[htp]
\begin{center}
  \includegraphics[width=5in]{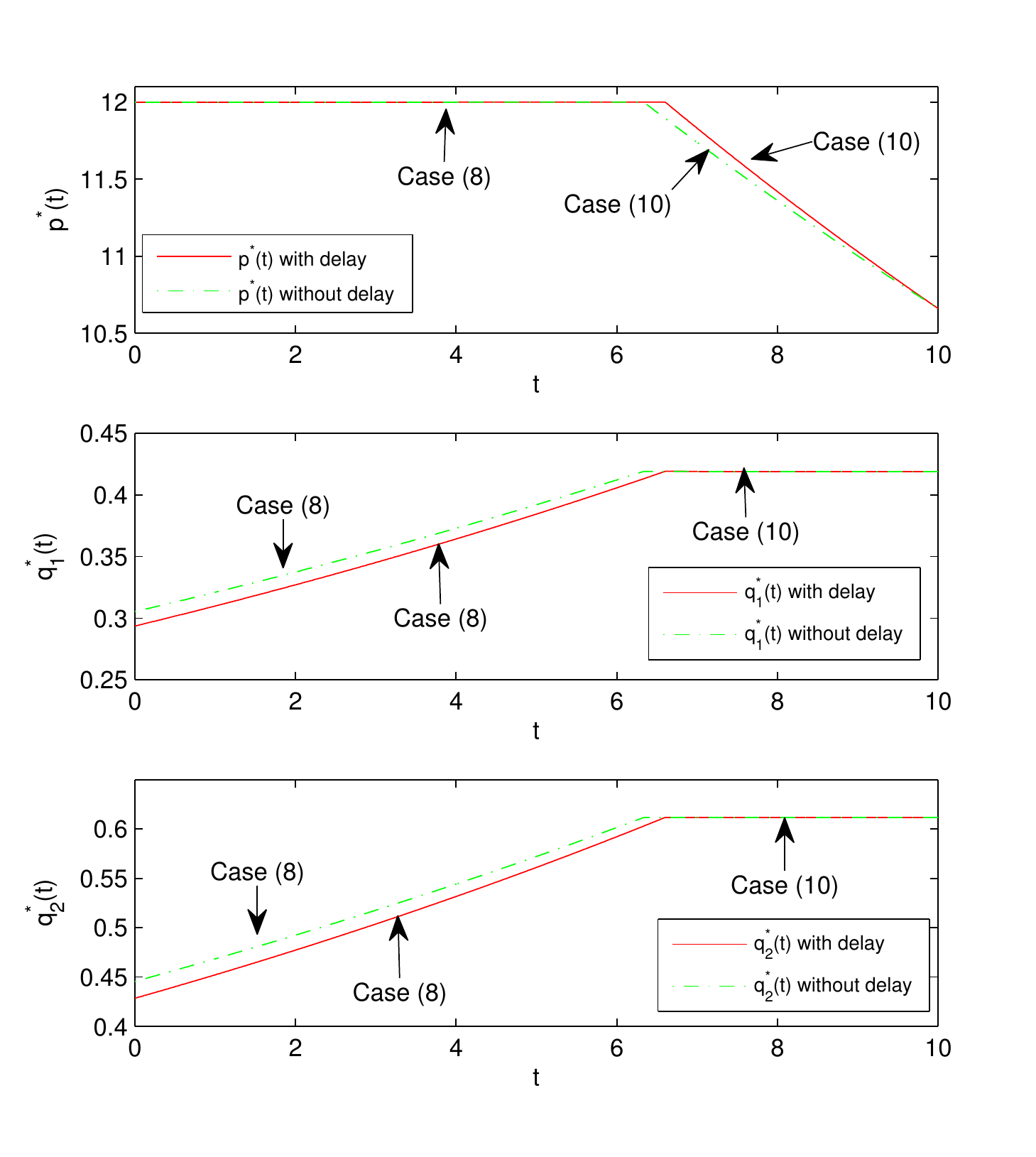}\\
  \caption{$p^{\ast}(t)$, $q_1^{\ast}(t)$ and $q_2^{\ast}(t)$ with and without delay.}
  \label{fig:pqdelayt}
  \end{center}
\end{figure}
\begin{table}[htp]
\setlength{\abovecaptionskip}{0.cm}
\small
\caption{Numerical results corresponding to Figure \ref{fig:pqdelayt}.} % title name of the table
\centering % centering table
\setlength{\tabcolsep}{1.2mm}
\begin{tabular}
    {@{}ccccccccccccccc@{}}
    \toprule
    %\multicolumn{2}{c}{Item} \\ \cmidrule(r){1-2}
$$& $t$ & $0$ & $1$ & $2$ & $3$ & $4$& $5$& $6$& $7$& $8$& $9$& $10$\\
\midrule
\multirow{3}{*}{$with~delay$}
&$p^{\ast}(t)$
& 12 &  12  & 12  & 12  & 12 &  12  & 12 &  11.831 &  11.419  & 11.030 &  10.661\\
&$q_1^{\ast}(t)$
& 0.294 & 0.310 & 0.327 & 0.345 & 0.364 & 0.384 & 0.406 & 0.419 & 0.419 & 0.419 & 0.419 \\
&$q_2^{\ast}(t)$
& 0.429 & 0.452 & 0.477 & 0.504 & 0.532 & 0.561 & 0.592 & 0.612 & 0.612 & 0.612 & 0.612
%其实当matlab保留8位小数的时候，后面的四个的都不一样，只不过matlab默认保留4位，我们这里删掉一位，保留3位
  \\
  \midrule
\multirow{3}{*}{$without~delay$}
&$p^{\ast}(t)$
& 12 &  12  & 12  & 12  & 12 &  12  & 12 & 11.739 & 11.361 & 11.002 & 10.661\\
&$q_1^{\ast}(t)$
& 0.305 & 0.321 & 0.337 & 0.355 & 0.373 & 0.392 & 0.412 & 0.419 & 0.419 & 0.419 & 0.419\\
&$q_2^{\ast}(t)$
& 0.446 & 0.468 & 0.492 & 0.518 & 0.544 & 0.572 & 0.601 & 0.612 & 0.612 & 0.612 & 0.612\\
\bottomrule
\end{tabular}
\label{pqq}
\end{table}
Figure \ref{fig:pqdelayt} shows the changes of the reinsurer's optimal premium strategy and the insurers' optimal reinsurance strategies with and without delay over time.
Table \ref{pqq} shows numerical results corresponding to Figure \ref{fig:pqdelayt}.
When $t\leq 6$, it satisfies conditions in Case (8) of Theorem \ref{Theorem1}.
In this case, the reinsurer's optimal premium strategy takes its upper bound $p^{\ast}(t)=\bar{c}=(1+\bar{\theta})\lambda_F\mu_F=12$; the optimal reserve proportional of insurer $i$ ($i\in\{1,2\}$) is
$q_i^{\ast}(t)=K[N^{\bar{c}F_i}(t)+\frac{k_i\rho\sigma_j}{\sigma_i}N^{\bar{c}F_j}(t)]$, which is an increasing function of $t$.
When $t\geq 7$, it satisfies conditions in Case (10) of Theorem \ref{Theorem1}.
In this case, the reinsurer's optimal premium strategy is $p^{\ast}(t)=\frac{P^N(t)}{P^D(t)}$, which is a decreasing function of $t$;
the optimal reserve proportional of insurer $i$ ($i\in\{1,2\}$) is
$q_i^{\ast}(t)=K [\frac{\frac{P^N(t)}{P^D(t)}-a_i}{\gamma_{i}\sigma_{i}^2\varphi^{F_i}(t)} +\frac{k_i\rho(\frac{P^N(t)}{P^D(t)}-a_j)}{\gamma_{j}\sigma_{1}\sigma_{2}\varphi^{F_j}(t)}]$.
Furthermore, from Figure \ref{fig:pqdelayt}, we can find that the price of reinsurance premium with delay is not lower than that without delay, and the reserve proportional with delay is not higher than that without delay.
This indicates that delay factors can urge the reinsurer and insurers to hedge risk by raising the price of reinsurance premium or reducing the reserve proportional under the setting of parameters in Table \ref{Tablefinance}, Table \ref{Tablereinsurer} and Table \ref{Tableinsurer}.
Furthermore, in Case (10), as time goes on, the gap between the reinsurance premium price with delay and that without delay becomes smaller and smaller, and the two are completely equal until the terminal time $T$.

Since Case (10) in Theorem \ref{Theorem1} is the most general situation in this paper, we will analyze the premium strategy and reinsurance strategies in Case (10) below.
For convenience, we choose the strategies at t=9 for analysis.
%%%%%%%%%%%%%%%%%%%%%%%%%%%%%%%%%%%%%%%%%%%下面的分析都是基于case10的分析

\begin{figure}[htp]
\begin{center}
  \includegraphics[width=4in]{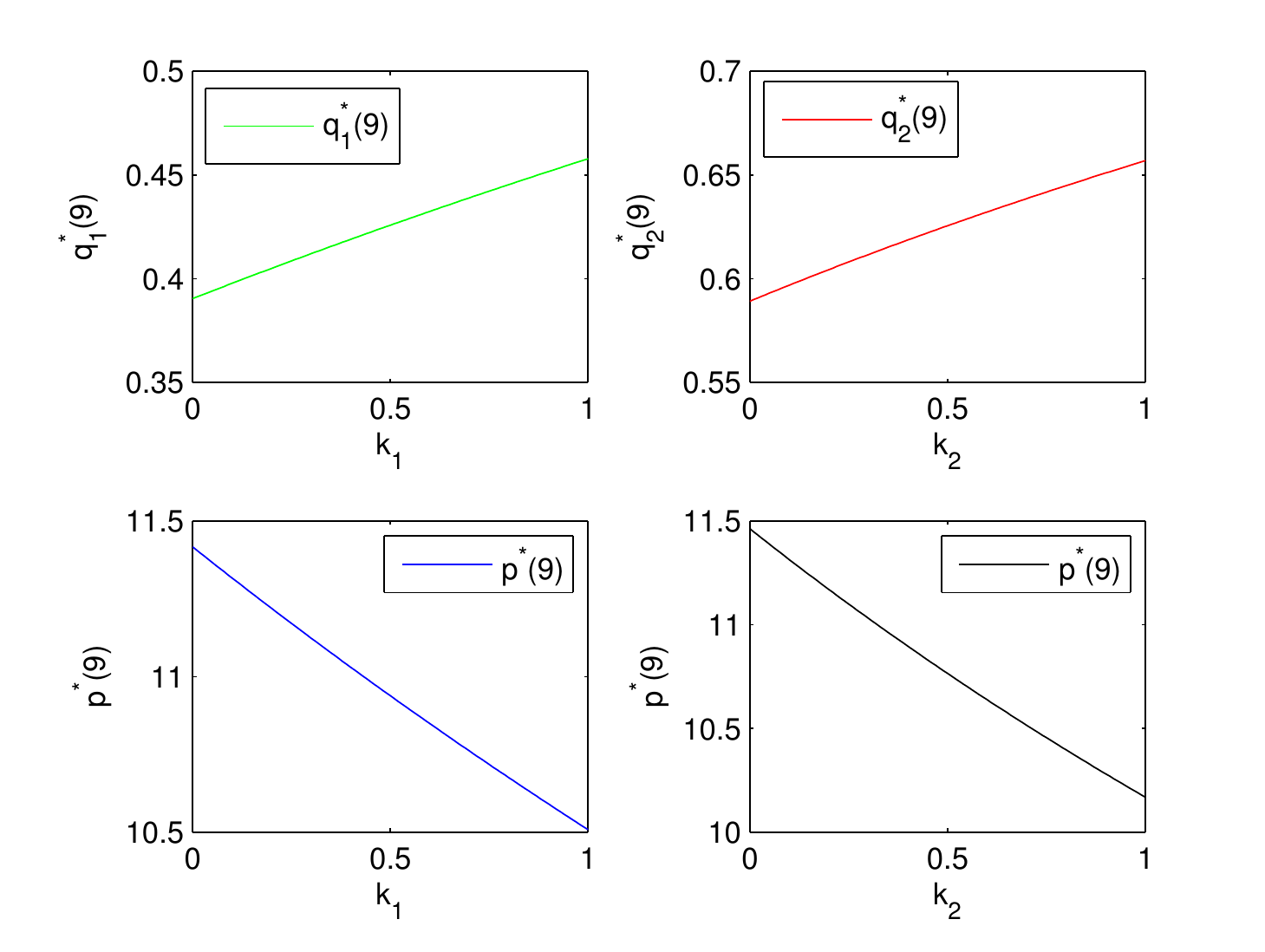}\\
  \caption{Effects of sensitivity parameters on optimal reinsurance strategies.}
  \label{fig:pqk}
  \end{center}
\end{figure}
Figure \ref{fig:pqk} illustrates the effects of sensitive parameters (i.e., $k_1,k_2$) on optimal reinsurance strategies (i.e., $q_1^{\ast}(9)$, $q_2^{\ast}(9)$) and optimal premium strategy (i.e., $p^{\ast}(9)$), respectively.
The more sensitive insurer $i$ ($i\in\{1,2\}$) is to the wealth of insurer $j$ ($j\neq i\in\{1,2\}$), (that is to say, $k_i$ is bigger), the higher the reservation ratio of insurer $i$ is, and the lower the premium price is.
This is because the more sensitive the insurer is to the wealth level of the other party, the more serious the psychology of comparison will be, which makes the insurer more eager to widen the wealth gap with its opponent.
This comparing mentality leads to its willingness to assume more risk of random claims than to spend money on reinsurance contracts.
This phenomenon in turn leads to lower the premium price.
In summary, competitive factors between the two insurers reduce the demand for reinsurance and the price of reinsurance premiums.

%%%关于横坐标范围的选取和时间点的选取均要保证是case10
\begin{figure}
  \centering
  \subfigure[Effect of $\gamma_L$ on $p^{\ast}(9)$]{
    \label{fig:pgammaL} %% label for second subfigure
     \begin{minipage}{5cm}
      \centering
    \includegraphics[width=5cm]{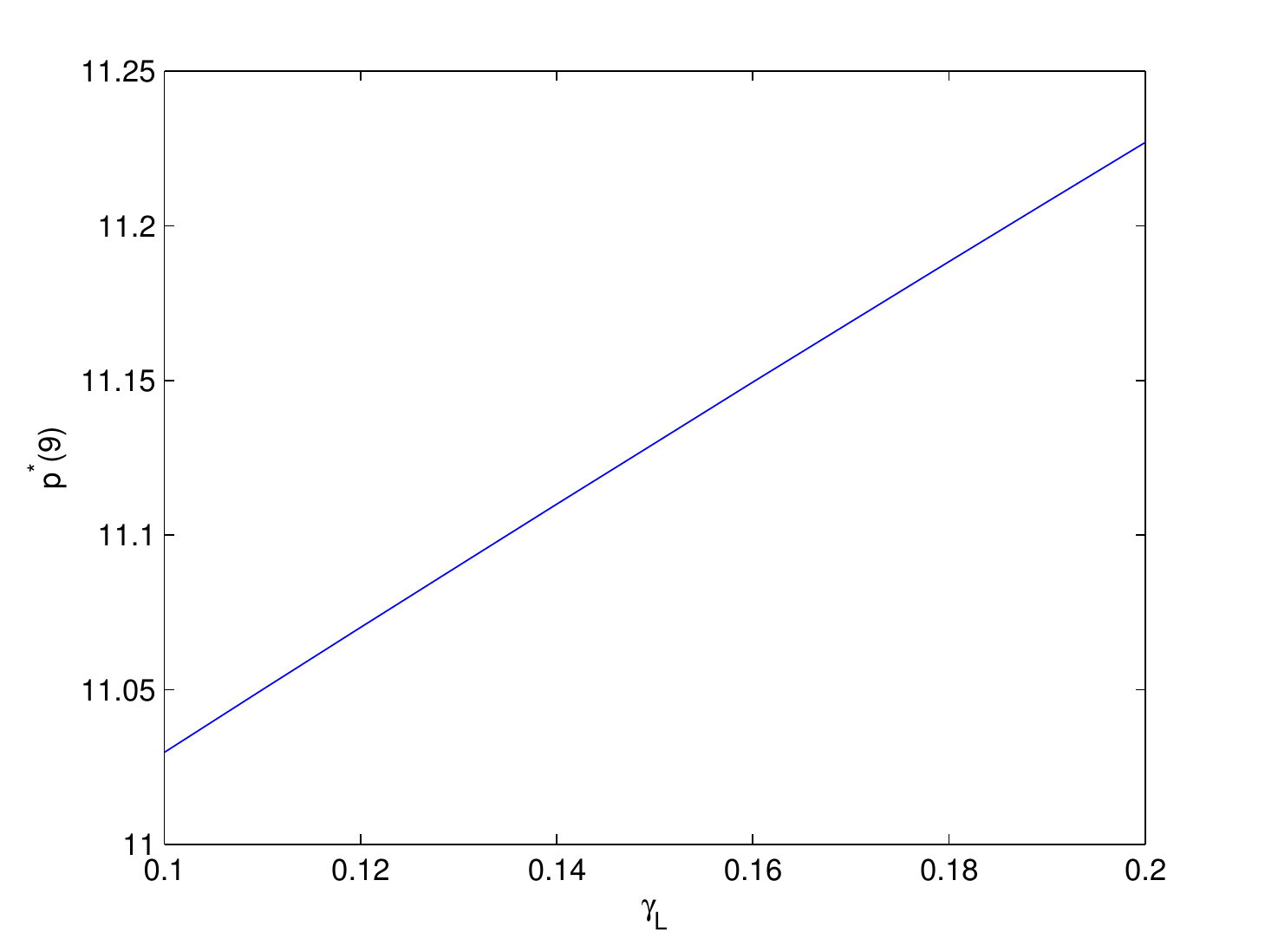}
    \end{minipage}
  }
  \subfigure[Effect of $\gamma_1$ on $p^{\ast}(9)$]{
    \label{fig:pgamma1}
     \begin{minipage}{5cm}
      \centering
    \includegraphics[width=5cm]{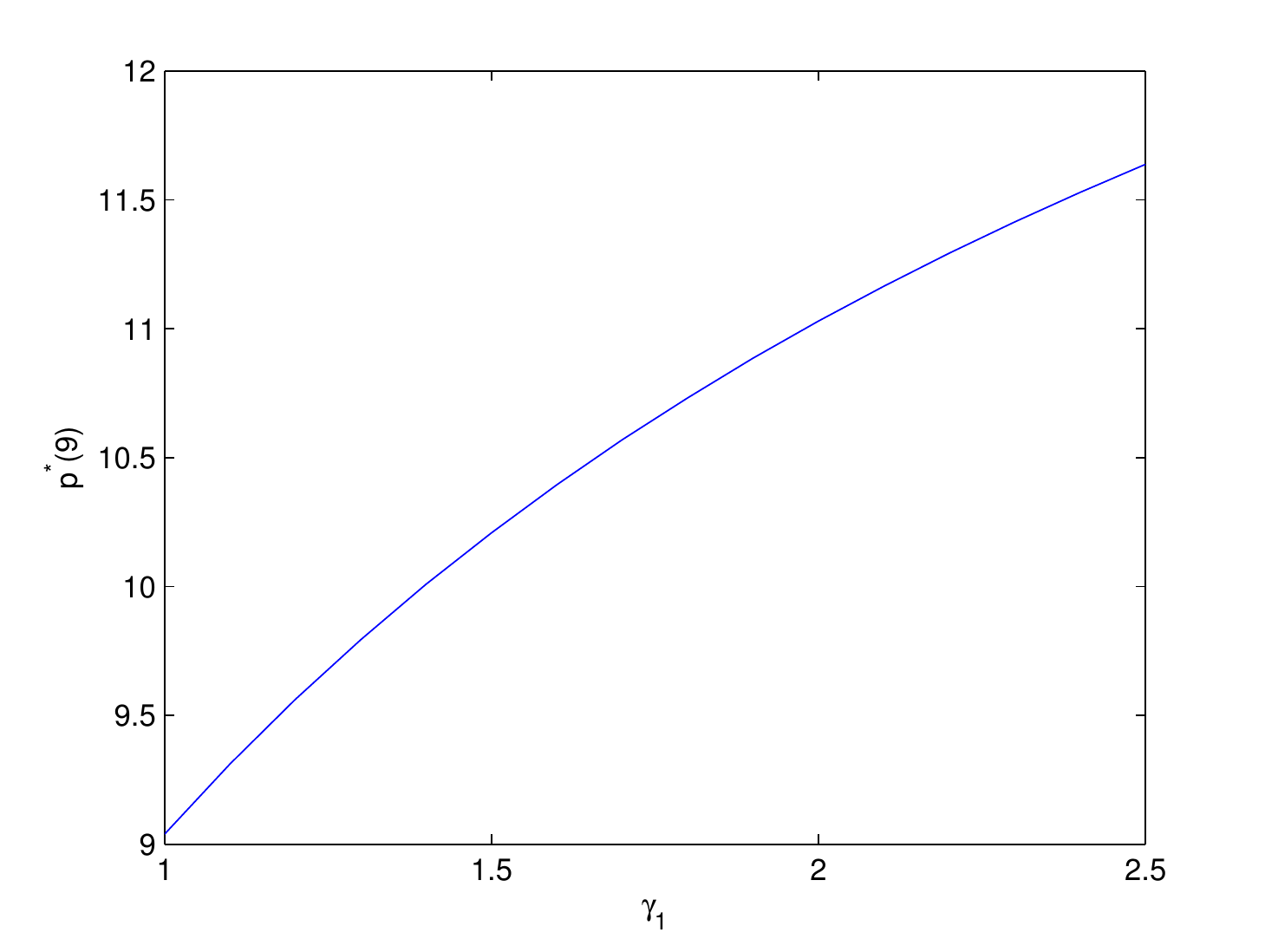}
    \end{minipage}
  }
  \subfigure[Effect of $\gamma_2$ on $p^{\ast}(9)$]{
    \label{fig:pgamma2}
     \begin{minipage}{5cm}
      \centering
    \includegraphics[width=5cm]{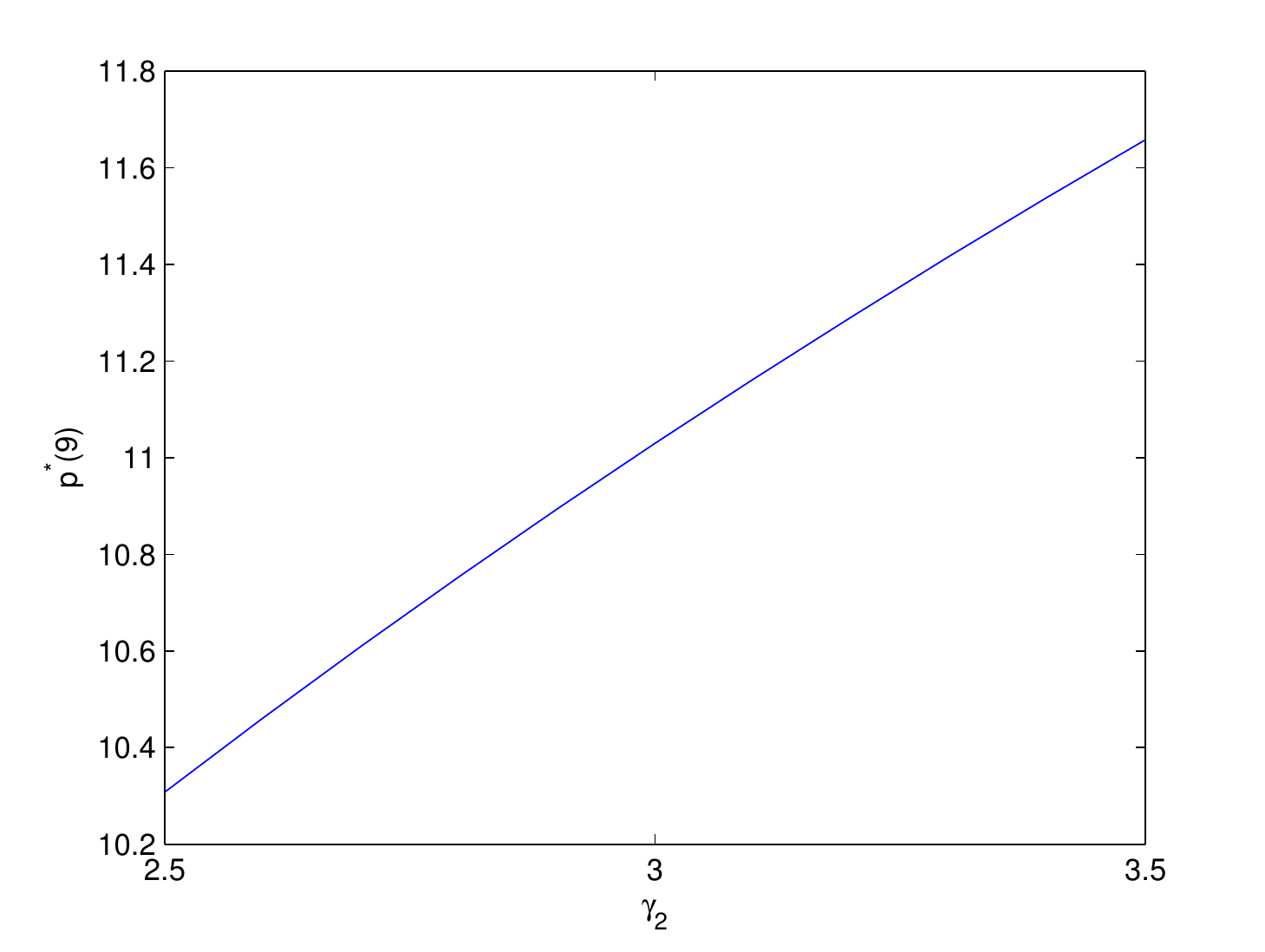}
    \end{minipage}
  }

    \subfigure[Effects of $\gamma_L$ on $q_1^{\ast}(9)$ and $q_2^{\ast}(9)$]{
    \label{fig:q1q2gammaL} %% label for second subfigure
     \begin{minipage}{5cm}
      \centering
    \includegraphics[width=5cm]{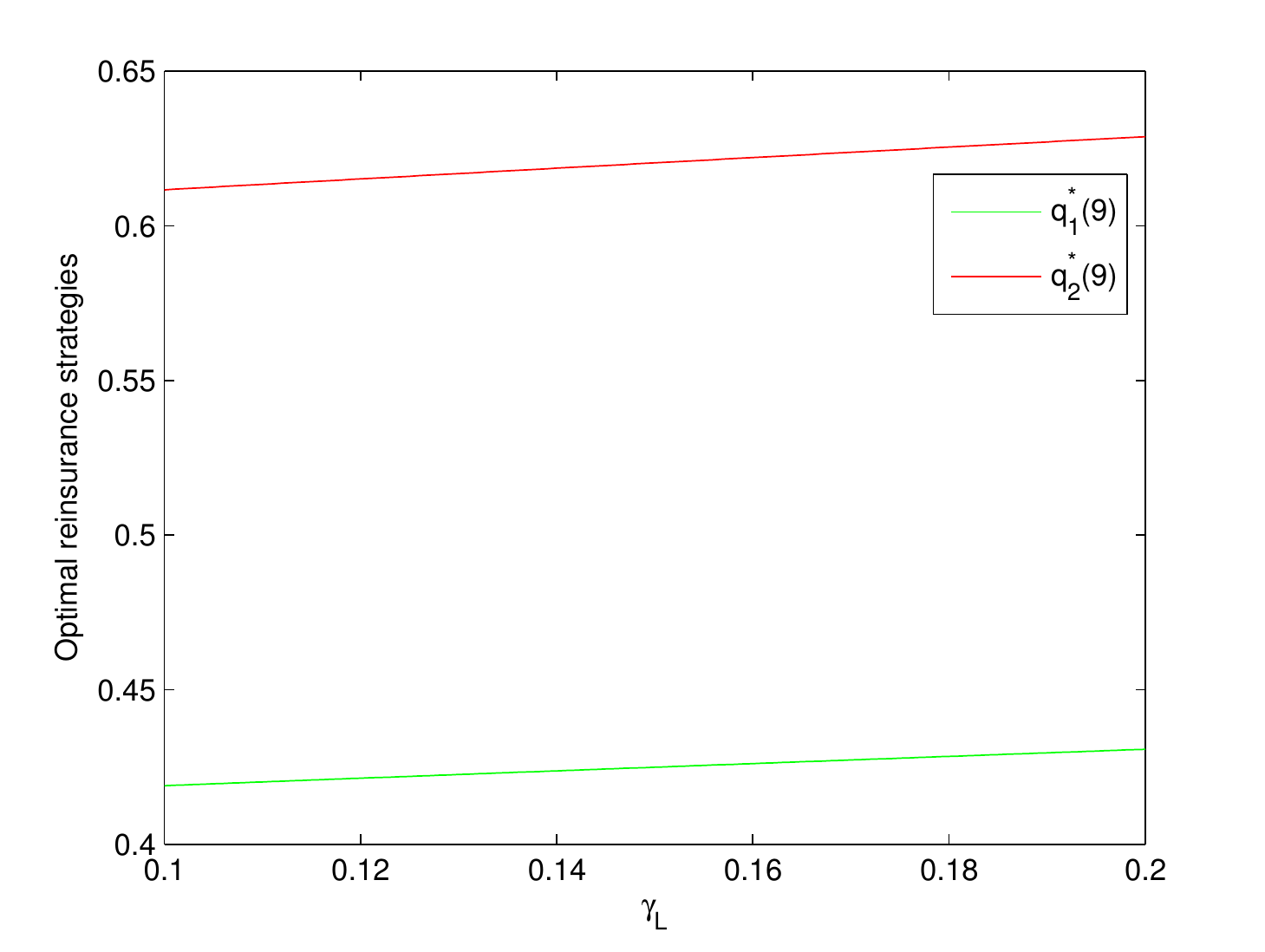}
    \end{minipage}
  }
  \subfigure[Effects of $\gamma_1$ on $q_1^{\ast}(9)$ and $q_2^{\ast}(9)$]{
    \label{fig:q1q2gamma1}
     \begin{minipage}{5cm}
      \centering
    \includegraphics[width=5cm]{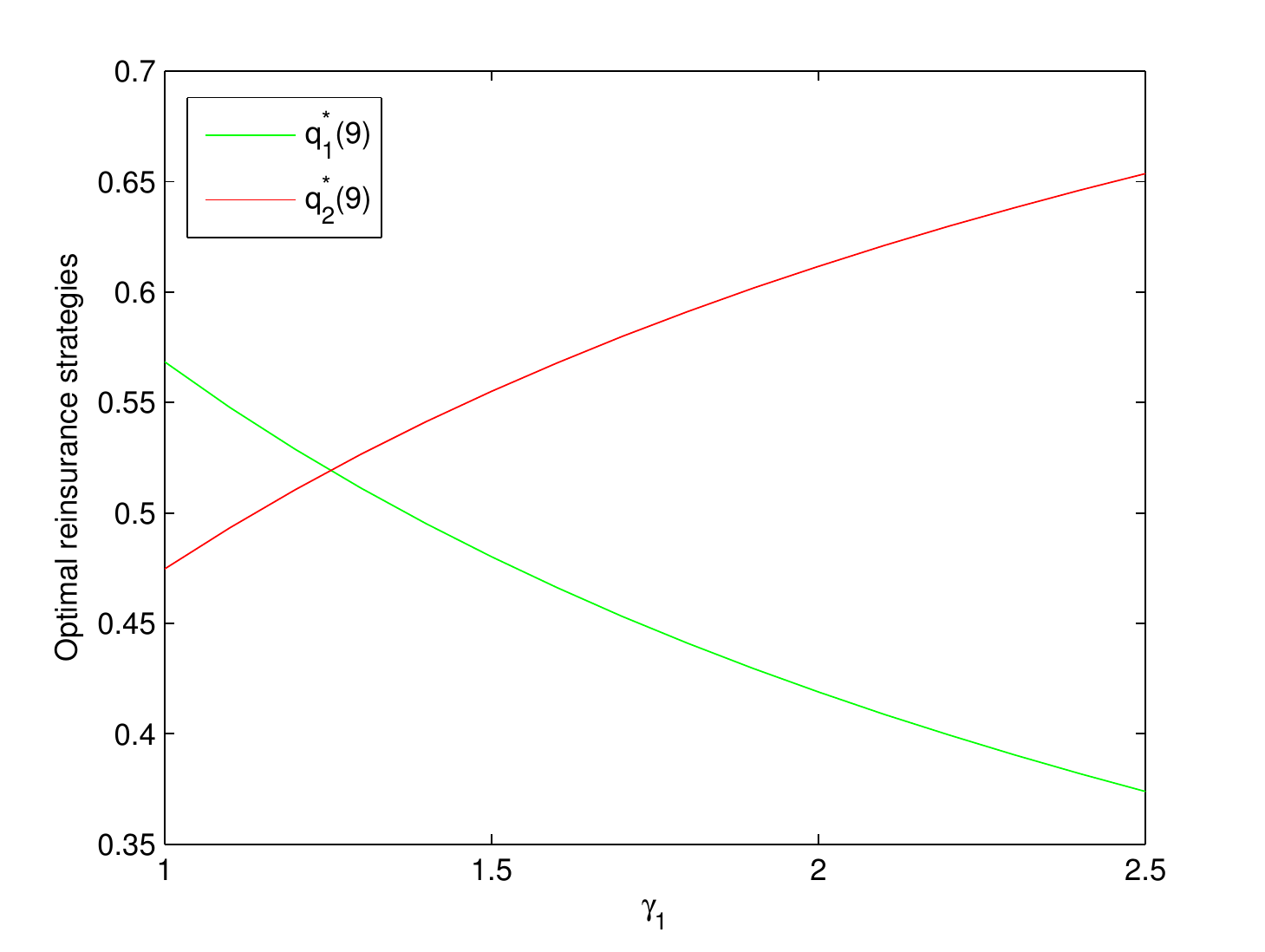}
    \end{minipage}
  }
  \subfigure[Effects of $\gamma_2$ on $q_1^{\ast}(9)$ and $q_2^{\ast}(9)$]{
    \label{fig:q1q2gamma2}
     \begin{minipage}{5cm}
      \centering
    \includegraphics[width=5cm]{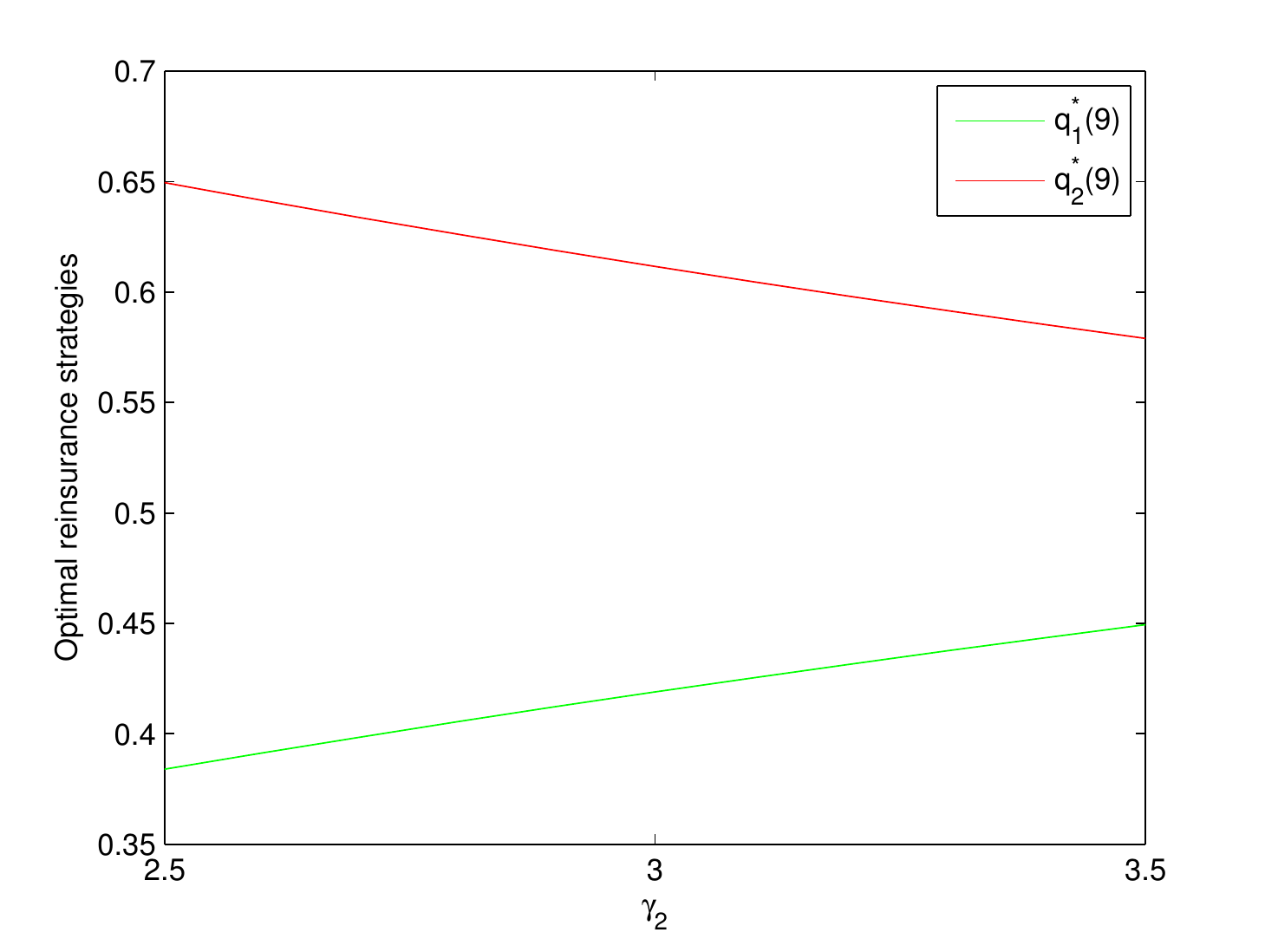}
    \end{minipage}
  }
  \caption{Effects of risk aversion coefficients on optimal strategies.}
  \label{fig:gamma} %% label for entire figure
\end{figure}
Three subgraphs in Figure \ref{fig:gamma} illustrate the sensitivity of $p^{\ast}(9)$, $q_1^{\ast}(9)$ and $q_2^{\ast}(9)$ to the $\gamma_L$, $\gamma_1$ and $\gamma_2$, respectively.
Subfigures \ref{fig:pgammaL}, \ref{fig:pgamma1} and \ref{fig:pgamma2} indicate that the reinsurance premium price will increase with the increase of risk aversion coefficients.
This is understandable, since the more risk-averse the reinsurer is, the more inclined the reinsurer is to reduce its risk by raising the premium price; the more risk-averse insurers are, the more demand for reinsurance will increase, leading to an increase in the price of reinsurance premium.
From \ref{fig:q1q2gammaL}, we can find that the reserve proportions of two insurers will increase with the increase of the reinsurer's risk aversion coefficient.
This phenomenon is because that the more risk-averse the reinsurer is, the higher the reinsurance premium price will be, which will reduce the reinsurance demand and lead to an increase in the reserve level of insurers.
\ref{fig:q1q2gamma1} and \ref{fig:q1q2gamma2}
show that the reserve proportion of insurer $i$ ($i\in\{1,2\}$) will decrease with the increase of its risk aversion coefficient $\gamma_i$ and increase with the increase of its competitor's (i.e., insurer $j$'s, $j\neq i\in\{1,2\}$) risk aversion coefficient $\gamma_j$.
Because when the insurer is more risk-averse, it tends to sign a reinsurance contract to transfer part of its random claim risk; when its competitor is more risk-averse, it prefers to increase its claim reserve ratio, but is not willing to spend money to buy reinsurance contracts due to the psychology of comparison.
To sum up, risk aversion factors increase the price of reinsurance premiums;
the reinsurer's risk aversion factor reduces the insurers' reinsurance demand;
insurer's risk aversion factor increases its own reinsurance demand and reduces the reinsurance demand of its competitor.

%%%%%%坐标范围的选取还要保证eta落在0和1之间
\begin{figure}
  \centering
  \subfigure[Effects of $\eta_L$ and $h_L$ on $p^{\ast}(9)$]{
    \label{fig:reinsuranceetaLhL} %% label for second subfigure
     \begin{minipage}{5cm}
      \centering
    \includegraphics[width=5cm]{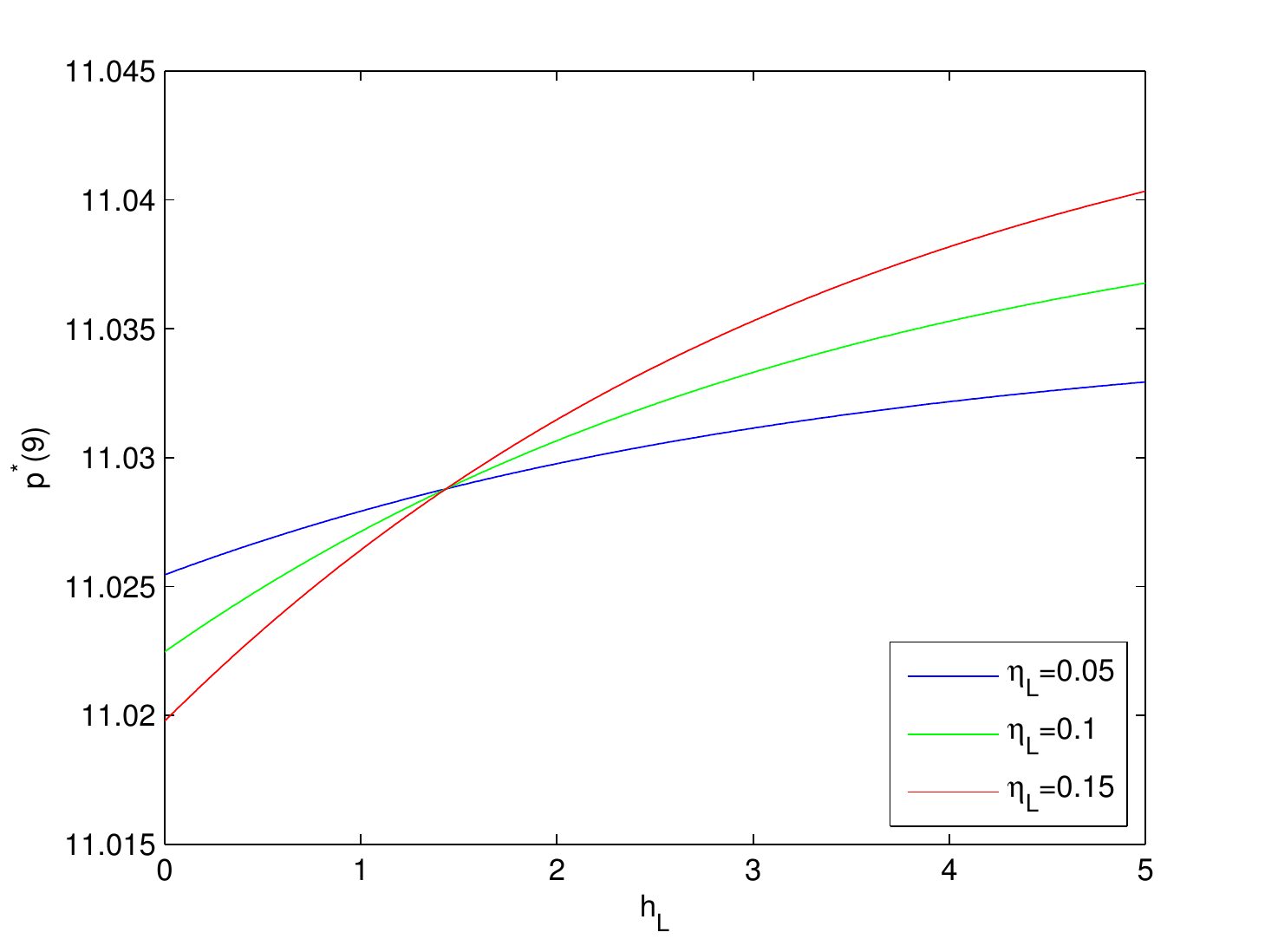}
    \end{minipage}
  }
  \subfigure[Effects of $\eta_1$ and $h_1$ on $q_1^{\ast}(9)$]{
    \label{fig:reinsuranceeta1h1}
     \begin{minipage}{5cm}
      \centering
    \includegraphics[width=5cm]{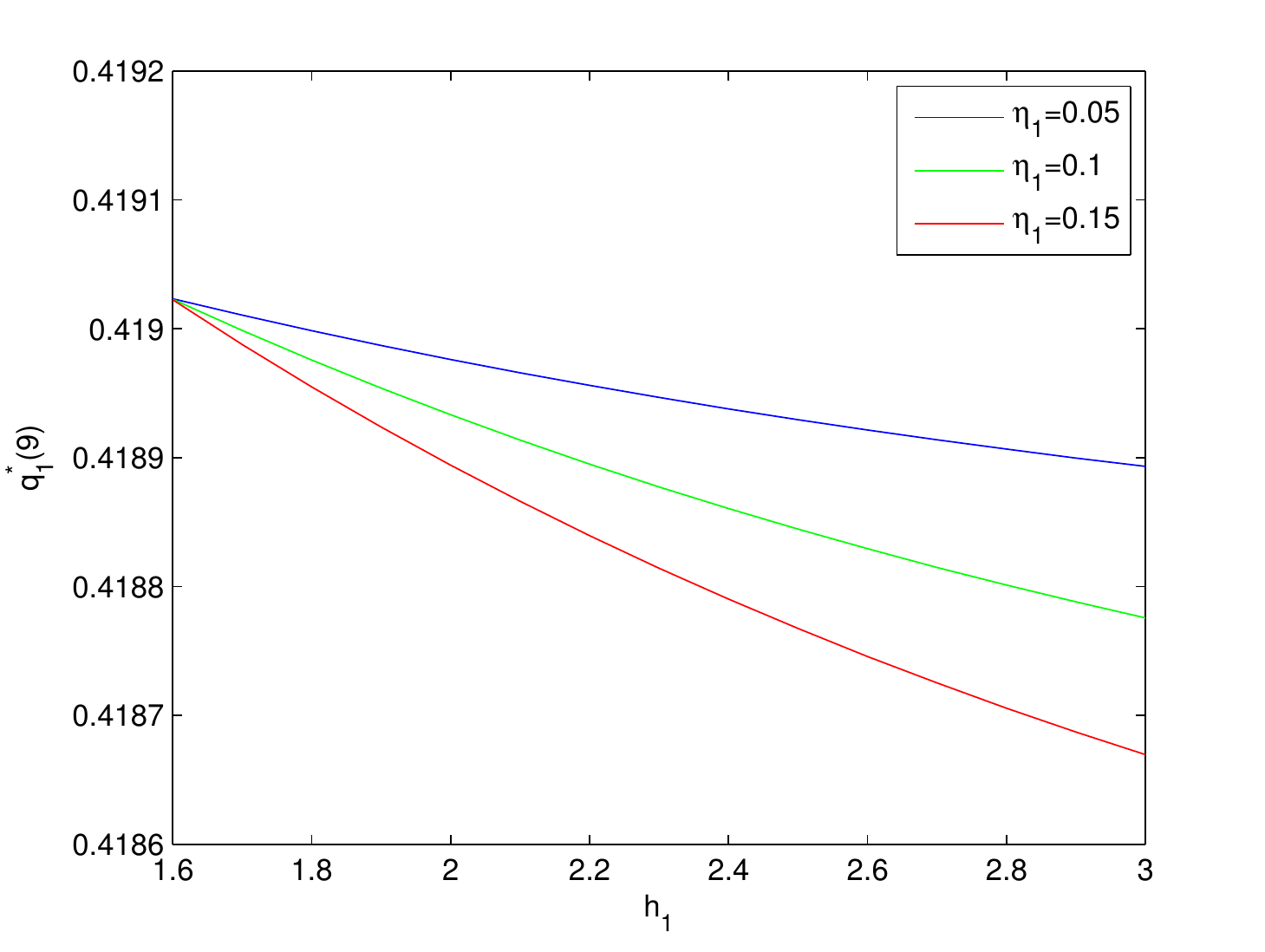}
    \end{minipage}
  }
  \subfigure[Effects of $\eta_2$ and $h_2$ on $q_2^{\ast}(9)$]{
    \label{fig:reinsuranceeta2h2}
     \begin{minipage}{5cm}
      \centering
    \includegraphics[width=5cm]{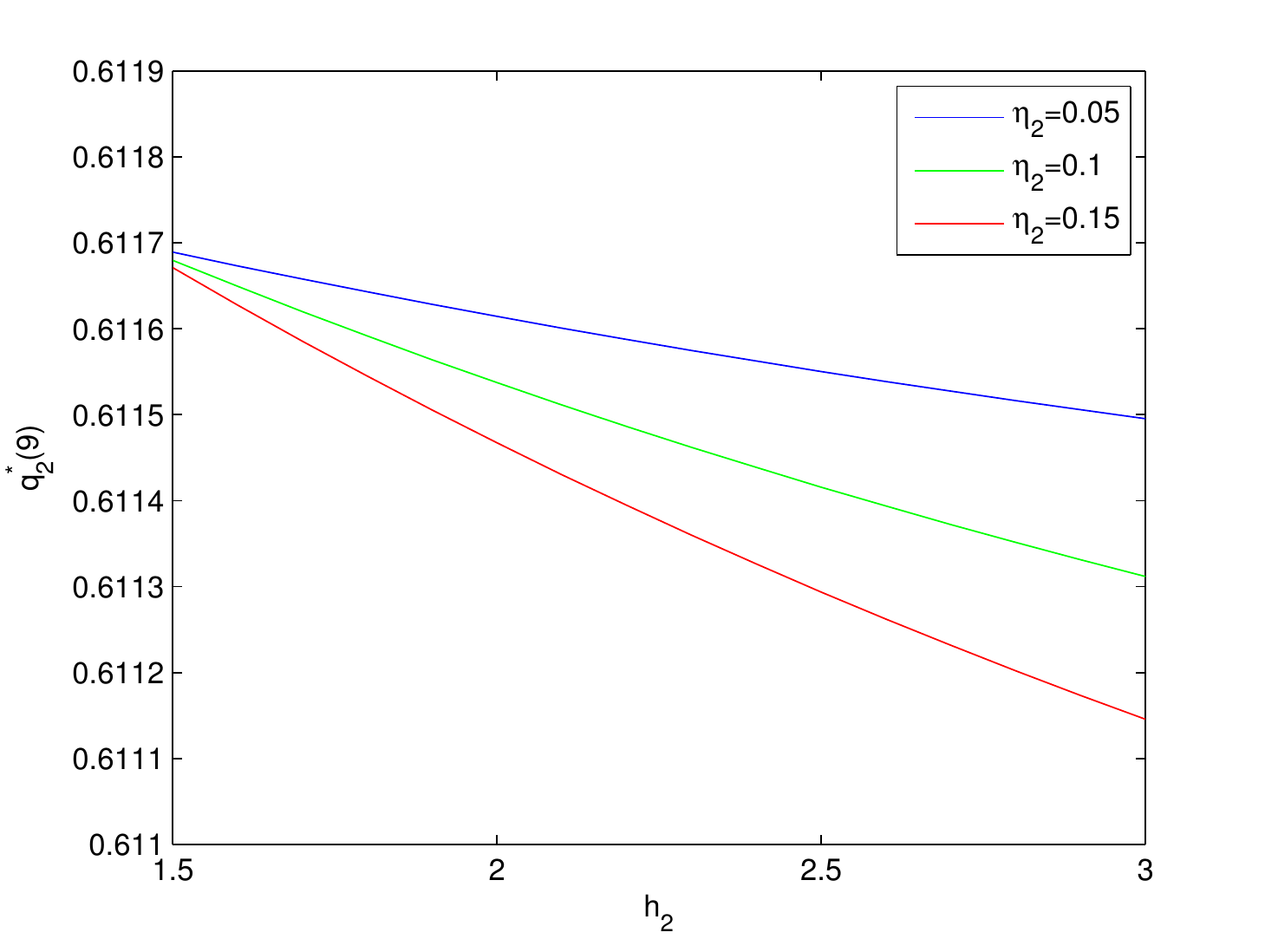}
    \end{minipage}
  }

    \subfigure[Effect of $\alpha_L$ on $p^{\ast}(9)$]{
    \label{fig:reinsurancealphaL} %% label for second subfigure
     \begin{minipage}{5cm}
      \centering
    \includegraphics[width=5cm]{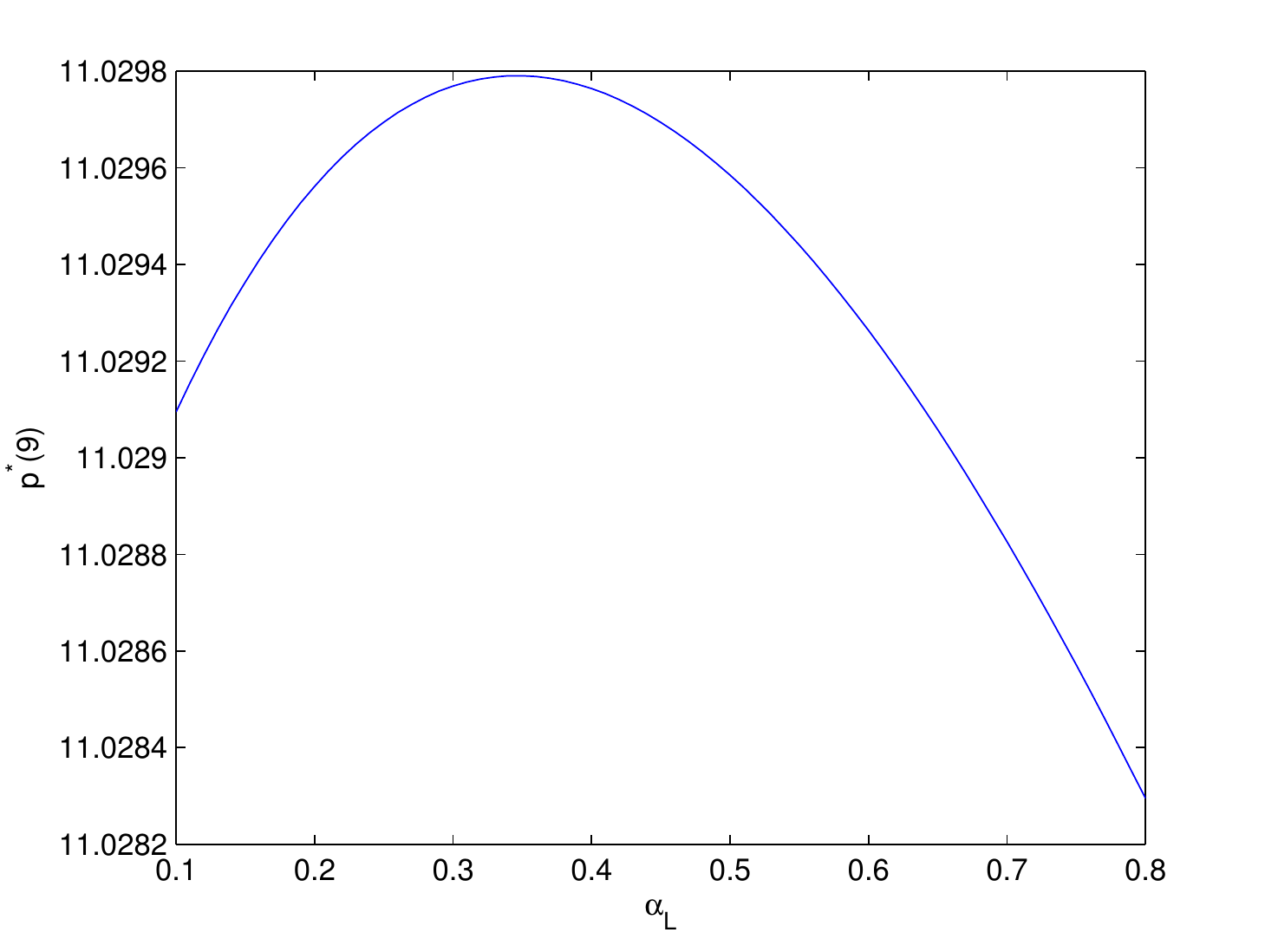}
    \end{minipage}
  }
  \subfigure[Effect of $\alpha_1$ on $q_1^{\ast}(9)$]{
    \label{fig:reinsurancealpha1}
     \begin{minipage}{5cm}
      \centering
    \includegraphics[width=5cm]{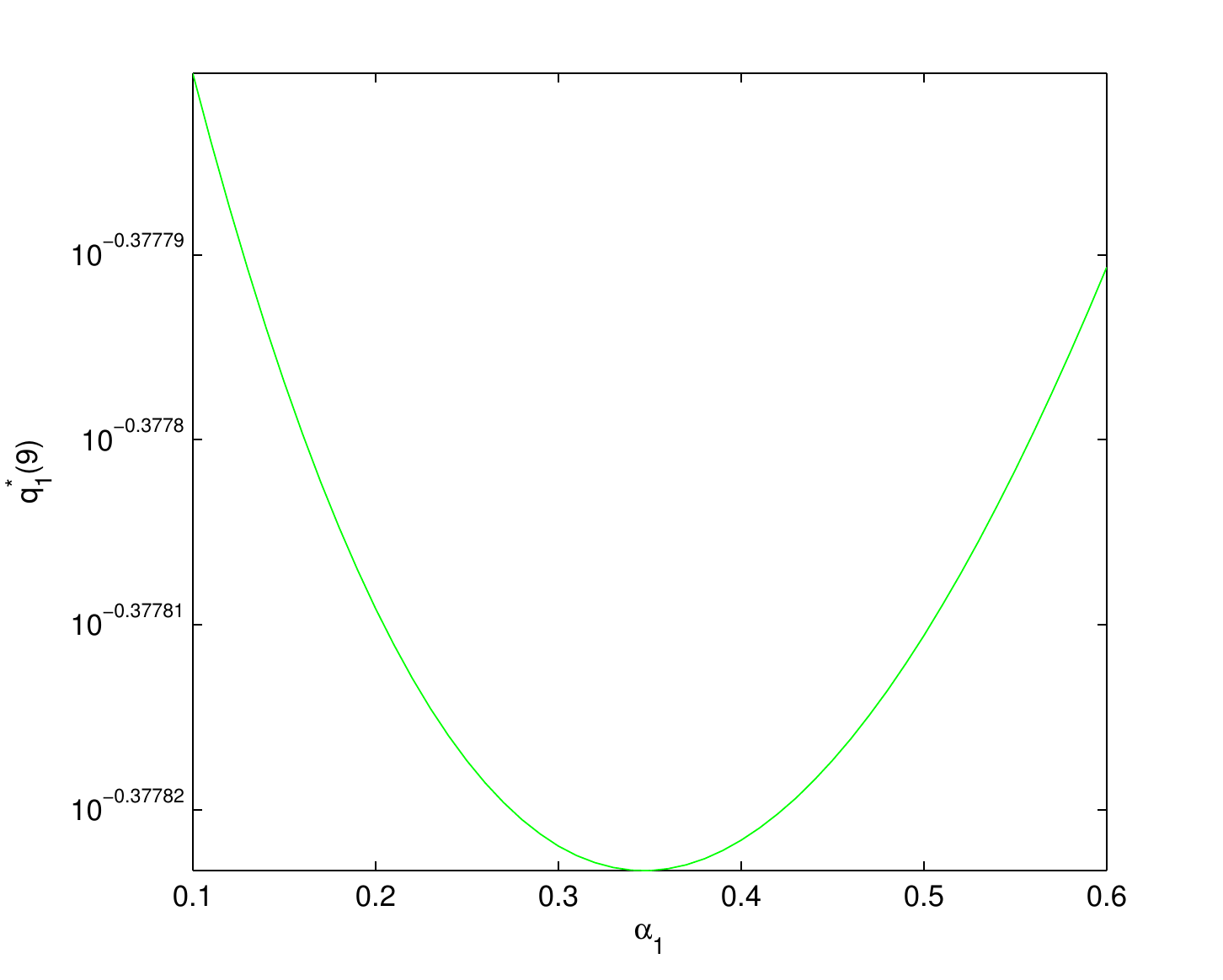}
    \end{minipage}
  }
  \subfigure[Effect of $\alpha_2$ on $q_2^{\ast}(9)$]{
    \label{fig:reinsurancealpha2}
     \begin{minipage}{5cm}
      \centering
    \includegraphics[width=5cm]{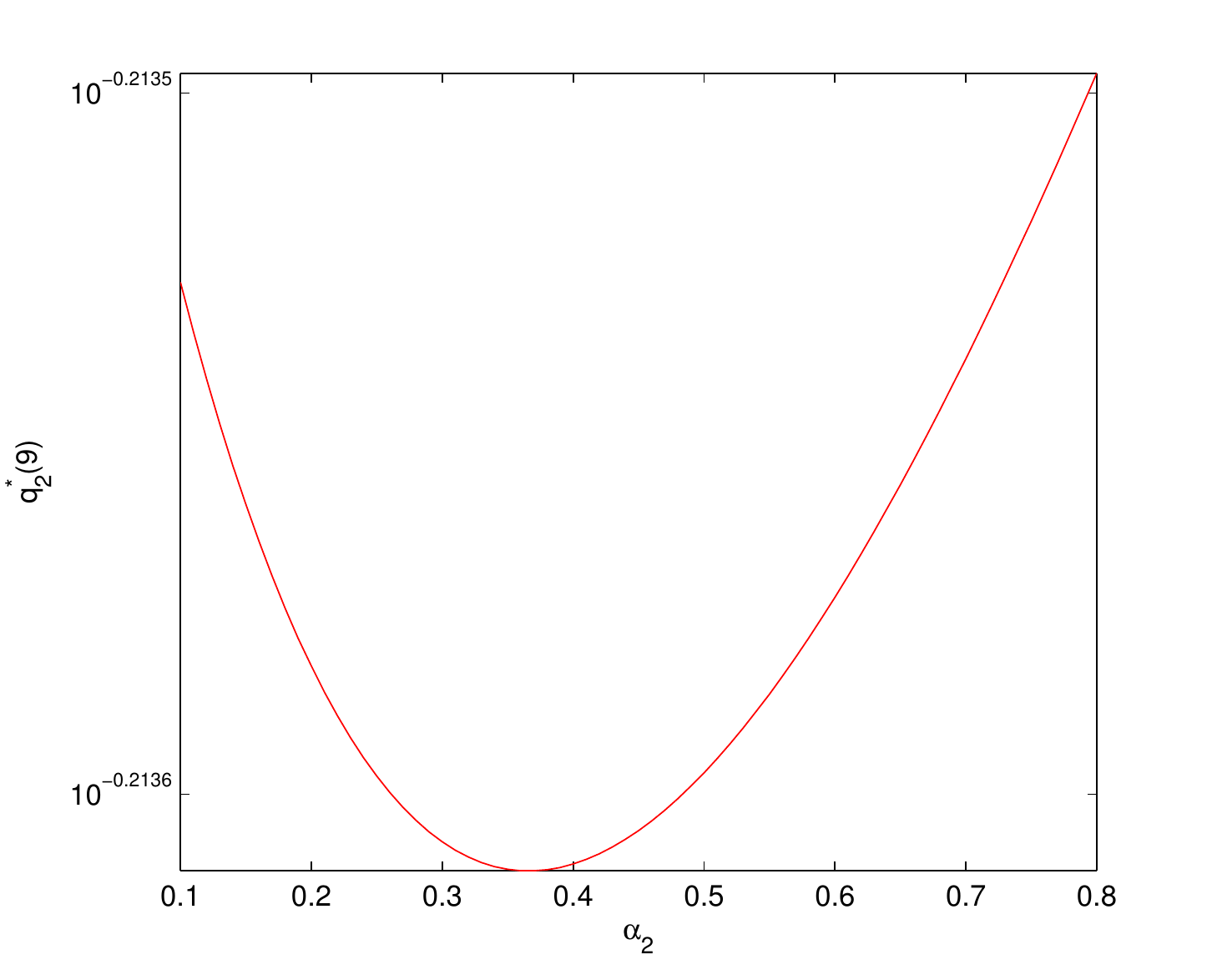}
    \end{minipage}
  }
  \caption{Effects of delay parameters on optimal strategies.}
  \label{fig:reinsurancedelay} %% label for entire figure
\end{figure}
Figure \ref{fig:reinsurancedelay} indicates the effects of delay parameters (i.e., $h_L$, $\eta_L$, $\alpha_L$, $h_1$, $\eta_1$, $\alpha_1$, $h_2$, $\eta_2$ and $\alpha_2$) on optimal premium strategy (i.e., $p^{\ast}(9)$) and reinsurance strategies (i.e., $q_1^{\ast}(9)$, $q_2^{\ast}(9)$), respectively.
%\footnote{Note: The selection of abscissa, delay time parameters and the delay weight parameters should not only ensure that the conditions of Case (10) in Theorem \ref{Theorem1} are satisfied, but also ensure that $\eta_i\in(0,1),i\in\{1,2\}$.}
As can be seen from \ref{fig:reinsuranceetaLhL}, when the delay weight is defined, the premium price increases with the increase of delay time.
Furthermore, when the delay time is less than a certain value, the larger the delay weight is, the lower the premium price is.
When the delay time is greater than this certain value, the larger the delay weight is, the higher the premium price is.
Therefore, for the reinsurer, the effect of the delay weight on the optimal premium strategy is related to the length of the delay time in Case (10) of Theorem \ref{Theorem1}.
From \ref{fig:reinsuranceeta1h1} and \ref{fig:reinsuranceeta2h2},  we can see that for insurers, when the delay time takes a special interval, the longer the delay time, the lower the reserve level; the larger the delay weight, the lower the reserve level.
Subfigures \ref{fig:reinsurancealphaL}, \ref{fig:reinsurancealpha1} and \ref{fig:reinsurancealpha2} show the effects of $\alpha_L$, $\alpha_1$ and $\alpha_2$ on $p^{\ast}(9)$, $q_1^{\ast}(9)$ and $q_2^{\ast}(9)$, respectively.
These three subgraphs illustrate that the impact of the average parameter on the optimal premium strategy (or optimal reinsurance strategy) will change with its size.

\section{Conclusion}\label{section 5}

In this study, we investigate a hybrid stochastic differential reinsurance and investment game, including a stochastic Stackelberg differential subgame and a non-zero-sum stochastic differential subgame.
One reinsurer and two insurers are three players in the hybrid game.
In view of the monopoly position of the reinsurer and the competitive relationship between insurers in the market, we consider the reinsurer and two insurers as the leader and the followers of the Stackelberg game respectively, and model the insurers' competition relationship as the non-zero-sum game.
We investigate the reinsurer's premium pricing and investment optimization problem as well as
insurers' reinsurance and investment optimization problem.
Under the consideration of the performance-related capital inflow/outflow,
the wealth processes of the reinsurer and insurers are described by SDDEs.
We derive the equilibrium strategy and value functions explicitly by using the idea of backward induction and the dynamic programming approach.
Then, we establish a verification theorem for the optimality of the given strategy.
Furthermore, we study several special cases of our model.
Moreover, some numerical examples and sensitivity analysis are presented to demonstrate the effects of the model parameters on the equilibrium strategy.

The main findings are as follows:
(1) the optimal reinsurance-investment strategies of two insurers interact with each other and reflect the herd effect.
(2) the optimal reinsurance strategies of insurers depend on the optimal premium strategy.
(3) competitive factors between two insurers reduce the demand for reinsurance and the price of reinsurance premium.
(4) the delay factor discourages or stimulates investment depending on the length of the delay.
When the delay time is greater than a certain value, the delay factor will make the investment become more conservative.
On the contrary, when the delay time is less than this certain value, the delay factor will stimulate investment.
(5) the effect of the delay weight on the equilibrium strategy is related to the length of the delay.
When the delay time is greater than a certain value, the optimal investment strategies and the optimal reinsurance strategies are negatively correlated with the corresponding delay weight parameters; the optimal premium strategy is positively correlated with the reinsurer's delay weight.
Conversely, when the delay time is less than this value, the opposite case occurs.

Stochastic differential games between reinsurers and insurers are a very common social phenomenon and an important current research issue in the economic and financial fields.
In the future work, this study can be extended in the following directions:
one is introducing multi-asset investment, which is closer to reality;
the other is considering regime switching to better describe the stochastic market.

\bibliographystyle{apalike}
\bibliography{main}%参考文献数据文件名（注意这里无后缀“.bib”）

\appendix
\renewcommand{\thesection}{Appendix~\Alph{section}}

\renewcommand{\theequation}{A.\arabic{equation}}

\section{Proof of Theorem \ref{Theorem1}}\label{Appendix A}

\begin{proof}
We solve the Stackelberg game problem by using the idea of backward induction mentioned in Section \ref{section 2.4} and standard dynamic programming techniques.

{\bfseries {Step 1}}
In the stochastic Stackelberg differential game, the reinsurer takes action first by announcing its any admissible strategy $(p(\cdot),b_L(\cdot))\in\Pi_L$.

{\bfseries {Step 2}}
Based on the reinsurer's strategy $(p(\cdot),b_L(\cdot))\in\Pi_L$,
we solve two insurers' optimization problems (i.e., \eqref{equ:218} for $i=1,2$) under the CARA preference simultaneously. $i\neq j\in\{1,2\}$ in this step.

For the value function of insurer $i$, we conjecture that
\begin{align}
V^{F_i}(t,\hat{x}_i,y_i,y_j,s)
=-\frac{1}{\gamma_i}\exp\{-\gamma_i\varphi^{F_i}(t)(\hat{x}_i
+\eta_iy_i-k_i\eta_jy_j)+g^{F_i}_1(t)s^{-2\beta}+g^{F_i}_2(t)\},\label{equ:VFi}
\end{align}
where $\varphi^{F_i}(t)$, $g^{F_i}_1(t)$ and $g^{F_i}_2(t)$ are deterministic, continuously differentiable functions with boundary conditions $\varphi^{F_i}(T)=1$, $g^{F_i}_1(T)=0$ and $g^{F_i}_2(T)=0$.
For insurer $i$, the HJB equation is
\begin{align}\label{equ:FiHJB}
0=&\sup_{(q_i(\cdot),b_i(\cdot))\in\Pi_i}\big\{V^{F_i}_t+V^{F_i}_{\hat{x}_i}
\big[\theta_ia_i-k_i\theta_ja_j-(p(t)-a_i)(1-q_i(t))+k_i(p(t)-a_j)(1-q_j^{\ast}(t))\nonumber\\
&+A_ix_i-k_iA_jx_j+B_iy_i-k_iB_jy_j
+C_iz_i-k_iC_jz_j+(r-r_0)(b_i(t)-k_ib_j^{\ast}(t))\big]
+\frac{1}{2}\big[(q_i(t)\sigma_i)^2\nonumber\\
&+(k_iq_j^{\ast}(t)\sigma_j)^2
-2q_i(t)\sigma_ik_iq_j^{\ast}(t)\sigma_j\rho
+(b_i(t)-k_ib_j^{\ast}(t))^2\sigma^2s^{2\beta}\big]V^{F_i}_{\hat{x}_i\hat{x}_i}
+(x_i-\alpha_iy_i-e^{-\alpha_ih_i}z_i)V^{F_i}_{y_i}\nonumber\\
&
+(x_j-\alpha_jy_j-e^{-\alpha_jh_j}z_j)V^{F_i}_{y_j}+rsV^{F_i}_{s}
+\frac{1}{2}\sigma^2s^{2\beta+2}V^{F_i}_{ss}
+(b_i(t)-k_ib_j^{\ast}(t))\sigma^2s^{2\beta+1} V^{F_i}_{\hat{x}_is}\big\}.
\end{align}
The first-order condition for maximizing the value in \eqref{equ:FiHJB} gives that
\begin{align}
0=&V^{F_i}_{\hat{x}_i}[p(t)-a_i]
+[q_i^{\ast}(t)\sigma_i^2-\sigma_ik_iq_j^{\ast}(t)\sigma_j\rho]V^{F_i}_{\hat{x}_i\hat{x}_i},
\label{equ:dqi}\\
0=&(r-r_0)V^{F_i}_{\hat{x}_i}+[b_i^{\ast}(t)
-k_ib_j^{\ast}(t)]\sigma^2s^{2\beta}V^{F_i}_{\hat{x}_i\hat{x}_i}
+\sigma^2s^{2\beta+1}V^{F_i}_{\hat{x}_is}\label{equ:dbi}.
\end{align}
From \eqref{equ:VFi}, \eqref{equ:dqi} and \eqref{equ:dbi}, we know that reinsurance strategy and investment strategy are independent.
The investment strategy of insurer $i$ is independent of the investment strategy of the reinsurer. Then, $\alpha_i^{\ast}(\cdot,p(\cdot),b_L(\cdot))$ and $\beta_i^{\ast}(\cdot,p(\cdot),b_L(\cdot))$ could be written as $\alpha_i^{\ast}(\cdot,p(\cdot))$ and $\beta_i^{\ast}(\cdot)$, respectively. Due to $q_i(t)\in[0,1]$, we can get that
\begin{align}
q_i^{\ast}(t,p(t))=\alpha_i^{\ast}(t,p(t))
&=\Big[\frac{-(p(t)-a_i)V^{F_i}_{\hat{x}_i}}{\sigma_i^2V^{F_i}_{\hat{x}_i\hat{x}_i}}
+\frac{k_i\sigma_j\rho q_j^{\ast}(t)}{\sigma_i}\Big]\vee 0 \wedge 1
=\Big[\frac{p(t)-a_i}{\gamma_i\sigma_i^2\varphi^{F_i}(t)}+\frac{k_i\rho \sigma_jq_j^{\ast}(t)}{\sigma_i}\Big]\vee 0 \wedge 1,\label{equ:fq}\\
b^{\ast}_i(t)=\beta^{\ast}(t)
&=k_ib_j^{\ast}(t)-\frac{sV^{F_i}_{\hat{x}_is}}{V^{F_i}_{\hat{x}_i\hat{x}_i}}
-\frac{(r-r_0)V^{F_i}_{\hat{x}_i}}{\sigma^2s^{2\beta}V^{F_i}_{\hat{x}_i\hat{x}_i}}
=k_ib_j^{\ast}(t)+\frac{1}{\gamma_i\varphi^{F_i}(t)s^{2\beta}}
\big[\frac{r-r_0}{\sigma^2}-2\beta g_1^{F_i}(t)\big].\label{equ:fb}
\end{align}
Substitute the investment strategy \eqref{equ:fb} into the HJB equation \eqref{equ:FiHJB}. Then, we have
\begin{align}
0=&V^{F_i}\big\{
\gamma_ix_i\big[-\varphi^{F_i}_t-\varphi^{F_i}(t)A_i-\varphi^{F_i}(t)\eta_i\big]
+k_i\gamma_ix_j\big[\varphi^{F_i}_t+\varphi^{F_i}(t)A_j+\varphi^{F_i}(t)\eta_j\big]
+\gamma_iy_i\big[-\varphi^{F_i}_t\eta_i
-\varphi^{F_i}(t)B_i\nonumber\\
&+\varphi^{F_i}(t)\eta_i\alpha_i\big]
+k_i\gamma_iy_j\big[\varphi^{F_i}_t\eta_j+\varphi^{F_i}(t)B_j
-\varphi^{F_i}(t)\eta_j\alpha_j\big]+\gamma_i\varphi^{F_i}(t)z_i\big[-C_i+\eta_ie^{-\alpha_ih_i}\big]
+k_i\gamma_i\varphi^{F_i}(t)z_j\big[C_j\nonumber\\
&-\eta_je^{-\alpha_jh_j}\big]+s^{-2\beta}\big(\frac{dg^{F_i}_1(t)}{dt}-2\beta r_0g^{F_i}_1(t)-\frac{1}{2}\frac{(r-r_0)^2}{\sigma^2}\big)\big\}
+V^{F_i}\big\{\frac{dg^{F_i}_2(t)}{dt}+\beta(2\beta+1)\sigma^2g^{F_i}_1(t)
\nonumber\\
&-\gamma_i\varphi^{F_i}(t)
\big[\theta_ia_i-k_i\theta_ja_j-(p(t)-a_i)(1-q_i^{\ast}(t))+k_i(p(t)-a_j)(1-q_j^{\ast}(t))\big]\nonumber\\
&
+\frac{1}{2}(\gamma_i\varphi^{F_i}(t))^2\big[(q_i^{\ast}(t)\sigma_i)^2+(k_iq_j^{\ast}(t)\sigma_j)^2
-2q_i^{\ast}(t)\sigma_ik_iq_j^{\ast}(t)\sigma_j\rho\big]\big\}.
\end{align}
Obviously, $q_i^{\ast}(t,p(t))$ does not depend on the state variables $x_i$, $x_j$, $y_i$, $y_j$ and $s^{-2\beta}$.
Due to $C_i=\eta_ie^{-\alpha_ih_i}$, $B_ie^{-\alpha_ih_i}=(\alpha_i+A_i+\eta_i)C_i$, $C_j=\eta_je^{-\alpha_jh_j}$, $B_je^{-\alpha_jh_j}=(\alpha_j+A_j+\eta_j)C_j$, $\varphi^{F_i}(T)=1$ and $g^{F_i}_1(T)=0$, we have
\begin{align}
\varphi^{F_i}(t)&=\exp\{(A_i+\eta_i)(T-t)\}=\exp\{(A_j+\eta_j)(T-t)\}=\varphi^{F_j}(t),\label{equ:varphiF}\\
g^{F_i}_1(t)&=g_1(t),\label{equ:g1Fi}
\end{align}
and
\begin{align}\label{equ:FiHJB1}
 0=&V^{F_i}\big\{-\gamma_i\varphi^{F_i}(t)
\big[\theta_ia_i-k_i\theta_ja_j-(p(t)-a_i)(1-q_i^{\ast}(t))+k_i(p(t)-a_j)(1-q_j^{\ast}(t))\big]\nonumber\\
&
+\frac{1}{2}(\gamma_i\varphi^{F_i}(t))^2\big[(q_i^{\ast}(t)\sigma_i)^2+(k_iq_j^{\ast}(t)\sigma_j)^2
-2q_i^{\ast}(t)\sigma_ik_iq_j^{\ast}(t)\sigma_j\rho\big]
+\frac{dg^{F_i}_2(t)}{dt}+\beta(2\beta+1)\sigma^2g^{F_i}_1(t)\big\},
\end{align}
where
\begin{align}
g_1(t)=-\frac{1}{4\beta r_0}(\frac{r-r_0}{\sigma})^2[1-\exp\{-2\beta r_0(T-t)\}].\label{equ:g1t}
\end{align}

Then, we have $\varphi^{F_1}(t)=\varphi^{F_2}(t)$ and $g^{F_1}_1(t)=g^{F_2}_1(t)=g_1(t)$.
Assume that $k_1k_2<1$, by \eqref{equ:fb}, we can get that
\begin{align}
b^{\ast}_i(t)=\frac{s^{-2\beta}}{(1-k_1k_2)\varphi^{F_i}(t)}
(\frac{1}{\gamma_i}+\frac{k_i}{\gamma_j})
\big[\frac{r-r_0}{\sigma^2}-2\beta g_1(t)\big].\label{equ:fb*}
\end{align}
Due to $p(t)\in[c_F,\bar{c}]$ and \eqref{equ:fq}, we know that $\frac{p(t)-a_i}{\gamma_i\sigma_i^2\varphi^{F_i}(t)}+\frac{k_i\rho \sigma_jq_j^{\ast}(t)}{\sigma_i}>0$. Then,
\begin{align}
q_i^{\ast}(t,p(t))=\alpha_i^{\ast}(t,p(t))
&=\Big[\frac{p(t)-a_i}{\gamma_i\sigma_i^2\varphi^{F_i}(t)}+\frac{k_i\rho \sigma_jq_j^{\ast}(t)}{\sigma_i}\Big]\wedge 1.
\end{align}

Assume that $k_1k_2\rho^2<1$. According to \cite{Bensoussan-2014-50} and \cite{Deng-2018-264},
let $\tilde{q}_i(t,p(t))$ be the solution of the following system of equations:
\begin{equation}
\left\{
\begin{array}{rcl}
\tilde{q}_1(t,p(t))=\frac{p(t)-a_1}{\gamma_{1}\sigma_{1}^2\varphi^{F_1}(t)}+\frac{k_1\rho \sigma_{2}\tilde{q}_2(t)}{\sigma_{1}},\\
\tilde{q}_2(t,p(t))=\frac{p(t)-a_2}{\gamma_{2}\sigma_{2}^2\varphi^{F_2}(t)}+\frac{k_2\rho \sigma_{1}\tilde{q}_1(t)}{\sigma_{2}}.
\end{array} \right.
\end{equation}
Denote
\begin{align}\label{equ:g(t)}
g(t)=-\frac{(2\beta+1)(r-r_0)^2}{4r_0}\big[(T-t)+\frac{1}{2\beta r_0}(\exp\{-2\beta r_0(T-t)\}-1)\big].
\end{align}

Then, we will discuss the following situations:
\begin{itemize}
\item {\bfseries {Case (Fa)}}
If $\tilde{q}_i(t,p(t))\geq 1$, $i\in\{1,2\}$, we have $q_i^{\ast}(t,p(t))=1$.
Substituting $q_i^{\ast}(t,p(t))$ into \eqref{equ:FiHJB1}
and integrating from $T$ to $t$ gives
\begin{align}\label{equ:g2Fia}
g^{F_i}_2(t)=g^{F_ia}_2(t)
\doteq &
g(t)-\frac{\gamma_i(\theta_ia_i-k_i\theta_ja_j)}{A_i+\eta_i}[\varphi^{F_i}(t)-1]
+\frac{\gamma_i^2\big[\sigma_i^2+k_i^2\sigma_j^2
-2\sigma_ik_i\sigma_j\rho\big]}{4(A_i+\eta_i)}[(\varphi^{F_i}(t))^2-1].
\end{align}

\item {\bfseries {Case (Fb)}}
If $\tilde{q}_i(t,p(t))\geq 1$ and $\tilde{q}_j(t,p(t))<1$, $i\neq j\in\{1,2\}$, we have $q_i^{\ast}(t,p(t))=1$ and $q_j^{\ast}(t,p(t))=\frac{p(t)-a_j}{\gamma_j\sigma_j^2\varphi^{F_j}(t)}+\frac{k_j\rho \sigma_i}{\sigma_j}$.
Substituting $q_i^{\ast}(t,p(t))$ and $q_j^{\ast}(t,p(t))$ into \eqref{equ:FiHJB1}
and integrating from $T$ to $t$ gives
\begin{align}\label{equ:g2Fib}
&g^{F_i}_2(t)=
g(t)-\frac{\gamma_i(\theta_ia_i-k_i\theta_ja_j)}{A_i+\eta_i}[\varphi^{F_i}(t)-1]
+\frac{\gamma_i^2\sigma_i^2[1+(k_1k_2\rho)^2
-2k_1k_2\rho^2]}{4(A_i+\eta_i)}[(\varphi^{F_i}(t))^2-1]\nonumber\\
&
-\frac{k_i\gamma_i}{\gamma_j\sigma_j^2}
[1+\frac{k_i\gamma_i}{2\gamma_j}]\int_T^t(p(s)-a_j)^2ds
-k_i\gamma_i[-1+\frac{k_j\rho \sigma_i}{\sigma_j}
+\frac{k_1k_2\rho\gamma_i\sigma_i}{\gamma_j\sigma_j}
-\frac{\rho\sigma_i\gamma_i}{\gamma_j\sigma_j}]\int_T^t(p(s)-a_j)\varphi^{F_i}(s)ds,
\end{align}
and
\begin{align}\label{equ:g2Fjb}
g^{F_j}_2(t)=&
g(t)-\frac{\gamma_j
(\theta_ja_j-k_j\theta_ia_i)}{A_j+\eta_j}[\varphi^{F_j}(t)-1]
+\frac{(k_j\gamma_j\sigma_i)^2(1-\rho^2)}{4(A_j+\eta_j)}[(\varphi^{F_j}(t))^2-1]
+\frac{1}{2(\sigma_j)^2}\int_T^t(p(s)-a_j)^2ds\nonumber\\
&
-\gamma_j(1-\frac{k_j\rho\sigma_i}{\sigma_j})\int_T^t(p(s)-a_j)\varphi^{F_j}(s)ds.
\end{align}

\item {\bfseries {Case (Fc)}}
If $\tilde{q}_i(t,p(t))<1$, $i\in\{1,2\}$, we have $q_i^{\ast}(t,p(t))=\tilde{q}_i(t,p(t))$.
Substituting $q_i^{\ast}(t,p(t))$ into \eqref{equ:FiHJB1} and integrating from $T$ to $t$ gives
\begin{align}\label{equ:g2Fic}
g^{F_i}_2(t)&=
g(t)-\frac{\gamma_i
(\theta_ia_i-k_i\theta_ja_j)}{A_i+\eta_i}[\varphi^{F_i}(t)-1]
-\gamma_i\int_T^t\varphi^{F_i}(s)(p(s)-a_i)ds
+k_i\gamma_i\int_T^t\varphi^{F_i}(s)(p(s)-a_j)ds\nonumber\\
&
+\frac{1-(k_1k_2\rho)^2}{2(1-k_1k_2\rho^2)^2(\sigma_i)^2}\int_T^t(p(s)-a_i)^2ds
-\frac{k_i\gamma_i[2\gamma_j(1-k_1k_2\rho^2)+k_i\gamma_i(1-\rho^2)]}
{2(1-k_1k_2\rho^2)^2(\gamma_j\sigma_j)^2}
\int_T^t(p(s)-a_j)^2ds\nonumber\\
&-\frac{k_i\rho[-\gamma_i(1-k_1k_2)+k_j\gamma_j(1-k_1k_2\rho^2)]}
{(1-k_1k_2\rho^2)^2\sigma_1\sigma_2\gamma_j}\int_T^t(p(s)-a_i)(p(s)-a_j)ds.
\end{align}

\end{itemize}

{\bfseries {Step 3}}
Knowing that insurer $i$ would execute its reinsurance strategy and investment strategy according to \eqref{equ:fq} and \eqref{equ:fb*}, the reinsurer then decides on its optimal strategy $(p^{\ast}(\cdot),b^{\ast}_L(\cdot))\in\Pi_L$.

For the value function of the reinsurer, we conjecture that
\begin{align}
V^L(t,x_L,y_L,s)
=-\frac{1}{\gamma_L}\exp\{-\gamma_L\varphi^L(t)(x_L+\eta_Ly_L)+g^L_1(t)s^{-2\beta}+g^L_2(t)\},\label{equ:VL}
\end{align}
where $\varphi^L(t)$, $g^L_1(t)$ and $g^L_2(t)$ are deterministic, continuously differentiable functions with boundary conditions $\varphi^L(T)=1$, $g^L_1(T)=0$ and $g^L_2(T)=0$.
The HJB equation of the reinsurer is
\begin{align}\label{equ:LHJB}
0=&\sup_{(p(\cdot),b_L(\cdot))\in\Pi_L}\big\{V^L_t
+V^L_{x_L}\big[(p(t)-a_1)(1-q_1^{\ast}(t))+(p(t)-a_2)(1-q_2^{\ast}(t))+(r-r_0)b_L(t)
+A_Lx_L+B_Ly_L+C_Lz_L\big]\nonumber\\
&
+\frac{1}{2}\big[(1-q_1^{\ast}(t))^2(\sigma_1)^2+(1-q_2^{\ast}(t))^2(\sigma_2)^2
+(b_L(t))^2\sigma^2s^{2\beta}
+2(1-q_1^{\ast}(t))(1-q_2^{\ast}(t))\sigma_1\sigma_2\rho\big]V^L_{x_Lx_L}
\nonumber\\
&
+(x_L-\alpha_Ly_L-e^{-\alpha_Lh_L}z_L)V^L_{y_L}
+rsV^L_{s}+\frac{1}{2}\sigma^2s^{2\beta+2}V^L_{ss}
+b_L(t)\sigma^2s^{2\beta+1} V^L_{x_Ls}
\big\}.
\end{align}
The first-order condition about $b_L(t)$ for maximizing the value in \eqref{equ:LHJB} gives that
\begin{align}\label{equ:lb}
b^{\ast}_L(t)=-\frac{(r-r_0)}{\sigma^2s^{2\beta}}\frac{V^L_{x_L}}{V^L_{x_Lx_L}}-s\frac{V^L_{x_Ls}}{V^L_{x_Lx_L}}
=\frac{s^{-2\beta}}{\gamma_L\varphi^L(t)}\left[\frac{(r-r_0)}{\sigma^2}-2\beta g_1^L(t)\right].
\end{align}
Substitute the investment strategy \eqref{equ:lb} into the HJB equation \eqref{equ:LHJB}. Then, we have
\begin{align}\label{equ:LHJB1}
0
=&\sup_{p(\cdot)\in[c_F,\bar{c}]}
V^L\big\{\gamma_Lx_L[-\varphi^L_t-\varphi^L(t)A_L-\varphi^L(t)\eta_L]
+\gamma_Ly_L[-\varphi^L_t\eta_L-\varphi^L(t)B_L+\varphi^L(t)\eta_L\alpha_L]\nonumber\\
&-\gamma_L\varphi^L(t)[C_Lz_L-\eta_Le^{-\alpha_Lh_L}z_L]
+s^{-2\beta}\big[\frac{dg^L_1(t)}{dt}-2\beta r_0g^L_1(t)-\frac{1}{2}(\frac{r-r_0}{\sigma})^2\big]+\frac{dg^L_2(t)}{dt}
\nonumber\\
&+\beta(2\beta+1)\sigma^2g^L_1(t)
-\gamma_L\varphi^L(t)[(p(t)-a_1)(1-q_1^{\ast}(t))+(p(t)-a_2)(1-q_2^{\ast}(t))]\nonumber\\
&+\frac{1}{2}\gamma_L^2(\varphi^L(t))^2
[(1-q_1^{\ast}(t))^2\sigma_1^2+(1-q_2^{\ast}(t))^2\sigma_2^2
+2(1-q_1^{\ast}(t))(1-q_2^{\ast}(t))\sigma_1\sigma_2\rho]\big\}.
\end{align}
Obviously, $p(t)$ does not depend on the state variables $x_L$, $y_L$ and $s^{-2\beta}$. Then, we have
\begin{align}
\varphi^L(t)=&\exp\{(A_L+\eta_L)(T-t)\},\label{equ:varphiL}\\
g^L_1(t)=&g_1(t),\label{equ:g1L}\\
0=&\sup_{p(\cdot)\in[c_F,\bar{c}]}\big\{\frac{dg^L_2(t)}{dt}+\beta(2\beta+1)\sigma^2g^L_1(t)
-\gamma_L\varphi^L(t)[(p(t)-a_1)(1-q_1^{\ast}(t))+(p(t)-a_2)(1-q_2^{\ast}(t))]\nonumber\\
&+\frac{1}{2}\gamma_L^2(\varphi^L(t))^2
[(1-q_1^{\ast}(t))^2\sigma_1^2+(1-q_2^{\ast}(t))^2\sigma_2^2
+2(1-q_1^{\ast}(t))(1-q_2^{\ast}(t))\sigma_1\sigma_2\rho]\big\}.\label{equ:LHJB11}
\end{align}

To simplify our presentation, for $i\neq j\in\{1,2\}$, we denote
\begin{align}\label{equ:KKK}
&K^{\tilde{F}_i}=(\frac{\gamma_j\sigma_j^2}{\gamma_i\sigma_i^2}
+\frac{k_i\rho\sigma_j}{\sigma_i})
(1-\frac{k_j\rho\sigma_i}{\sigma_j}),  \quad K^{F_i}=(1+\frac{k_i\rho\gamma_i\sigma_i}{\gamma_j\sigma_j})
(1-\frac{k_i\rho\sigma_j}{\sigma_i}),\quad
N^{cF_i}(t)=\frac{c_F-a_i}{\gamma_i\sigma_i^2\varphi^{F_i}(t)},\nonumber\\
&
N^{\bar{c}F_i}(t)=\frac{\bar{c}-a_i}{\gamma_i\sigma_i^2\varphi^{F_i}(t)},\quad
N^{aF_i}(t)=\frac{a_j-a_i}{\gamma_i\sigma_i^2\varphi^{F_i}(t)}, \quad
M^{F_i}(t)=\frac{\gamma_i\varphi^{F_i}(t)+\gamma_L\varphi^L(t)}
{2\gamma_i\varphi^{F_i}(t)+\gamma_L\varphi^L(t)}, \quad
 K=\frac{1}{1-k_1k_2\rho^2}.
\end{align}

For the premium strategy $p(t)$, we discuss it in the following situations:
\begin{itemize}
\item {\bfseries {Case (La)}} If $\tilde{q}_i(t)\geq 1$, $i=1,2$, we have $q_1^{\ast}(t,p(t))=1$ and $q_2^{\ast}(t,p(t))=1$.
Substituting $q_1^{\ast}(t,p(t))$ and $q_2^{\ast}(t,p(t))$ into \eqref{equ:LHJB11}, we can get that
\begin{align}\label{equ:dg2La}
0=&\sup_{p(\cdot)\in[c_F,\bar{c}]}\big\{\frac{dg^L_2(t)}{dt}+\beta(2\beta+1)\sigma^2g^L_1(t)
\big\}.
\end{align}
Then,
\begin{align}
p^{\ast}(t)=p, \quad  q_1^{\ast}(t,p^{\ast}(t))=1, \quad q_2^{\ast}(t,p^{\ast}(t))=1,
\end{align}
where $p$ is an arbitrary value in the interval $[c_F,\bar{c}]$.
The precondition $\tilde{q}_i(t)\geq 1$ ($i=1,2$) becomes
$N^{cF_i}(t)+\frac{k_i\rho\sigma_j}{\sigma_i}\geq 1$, $i\neq j\in\{1,2\}$.
By equation \eqref{equ:dg2La} and $g_2^L(T)=0$, we can get that
\begin{align}
g^L_2(t)=g^{La}_2(t)\doteq g(t).\label{equ:gLa}
\end{align}

\item {\bfseries {Case (Lb)}}
If $\tilde{q}_i(t)\geq 1$, $\tilde{q}_j(t)< 1$, $i\neq j\in\{1,2\}$, we have $q_i^{\ast}(t,p(t))=1$ and $q_j^{\ast}(t,p(t))=\frac{p(t)-a_j}{\gamma_j\sigma_j^2\varphi^{F_j}(t)}+\frac{k_j\rho \sigma_i}{\sigma_j}$.
Substituting $q_i^{\ast}(t,p(t))$ and $q_j^{\ast}(t,p(t))$ into \eqref{equ:LHJB11} and simplifying gives
\begin{align}\label{equ:LHJBb}
0
=&\sup_{p(\cdot)\in[c_F,\bar{c}]}\big\{\frac{dg^L_2(t)}{dt}
+\beta(2\beta+1)\sigma^2g^L_1(t)-(p(t)-a_j)\gamma_L\varphi^L(t)(1-\frac{k_j\rho\sigma_i}{\sigma_j})
[1+\frac{\gamma_L\varphi^L(t)}{\gamma_j\varphi^{F_j}(t)}]\nonumber\\
&
+\frac{1}{2}(\gamma_L\varphi^L(t))^2(\sigma_j-k_j\rho\sigma_i)^2
+(p(t)-a_j)^2\frac{\gamma_L\varphi^L(t)}{\gamma_j\sigma_j^2\varphi^{F_j}(t)}
[1+\frac{1}{2}\frac{\gamma_L\varphi^L(t)}{\gamma_j\varphi^{F_j}(t)}]\big\}.
\end{align}

The first-order condition about $p(t)$ for maximizing the value in equation \eqref{equ:LHJBb} gives
\begin{align}
p^{\ast}(t)=[a_j+\gamma_j\sigma_j^2\varphi^{F_j}(t)(1-\frac{k_j\rho\sigma_i}{\sigma_j})
M^{F_j}(t)]\vee c_F \wedge \bar{c}.
\end{align}

\begin{itemize}
\item {\bfseries {Subcase (Lb1)}} If $a_j+\gamma_j\sigma_j^2\varphi^{F_j}(t)(1-\frac{k_j\rho\sigma_i}{\sigma_j})
M^{F_j}(t)\geq \bar{c}$,
we have
\begin{align}
p^{\ast}(t)=\bar{c},\quad
q_i^{\ast}(t,p^{\ast}(t))=1,\quad
q_j^{\ast}(t,p^{\ast}(t))=N^{\bar{c}F_j}(t)+\frac{k_j\rho \sigma_i}{\sigma_j}.
\end{align}
The preconditions become
$K[N^{\bar{c}F_i}(t)+\frac{k_i\rho\sigma_j}{\sigma_i}N^{\bar{c}F_j}(t)]\geq 1$ and
$N^{\bar{c}F_j}(t)+\frac{k_j\rho\sigma_i}{\sigma_j}< 1$.
Then, substituting $p^{\ast}(t)$ into \eqref{equ:LHJBb} and integrating from $T$ to $t$ gives
\begin{align}\label{equ:gLjb1}
g^L_2(t)=g^{Ljb1}_2(t)
\doteq &g(t)+\frac{\gamma_L^2(\sigma_j-k_j\rho\sigma_i)^2}{4(A_L+\eta_L)}
[(\varphi^L(t))^2-1]+(\bar{c}-a_j)\gamma_L(1-\frac{k_j\rho\sigma_i}{\sigma_j})
\big[\int_T^t\varphi^L(s)ds\nonumber\\
&
+\frac{\gamma_L}{\gamma_j}\int_T^t\frac{(\varphi^L(s))^2}{\varphi^{F_j}(s)}ds\big]
-(\bar{c}-a_j)^2\frac{\gamma_L}{\gamma_j\sigma_j^2}
\big[\int_T^t\frac{\varphi^L(s)}{\varphi^{F_j}(s)}ds
+\frac{\gamma_L}{2\gamma_j}\int_T^t(\frac{\varphi^L(s)}{\varphi^{F_j}(s)})^2ds\big].
\end{align}
Equations \eqref{equ:g2Fib} and \eqref{equ:g2Fjb} become
\begin{align}\label{equ:g2Fib1}
g^{F_i}_2(t)=g^{F_ib1}_2(t)
\doteq &
g(t)+\frac{\gamma_i(\varphi^{F_i}(t)-1)}{A_i+\eta_i}
\big[k_i(\bar{c}-a_j)(-1+\frac{k_j\rho \sigma_i}{\sigma_j}
+\frac{k_1k_2\rho\gamma_i\sigma_i}{\gamma_j\sigma_j}
-\frac{\rho\sigma_i\gamma_i}{\gamma_j\sigma_j})-(\theta_ia_i-k_i\theta_ja_j)\big]\nonumber\\
&
+\frac{\gamma_i^2\sigma_i^2[1+(k_1k_2\rho)^2
-2k_1k_2\rho^2]}{4(A_i+\eta_i)}[(\varphi^{F_i}(t))^2-1]
+\frac{k_i\gamma_i}{\gamma_j\sigma_j^2}
[1+\frac{k_i\gamma_i}{2\gamma_j}](\bar{c}-a_j)^2(T-t),
\end{align}
and
\begin{align}\label{equ:g2Fjb1}
g^{F_j}_2(t)=g^{\tilde{F}_jb1}_2(t)
\doteq &
g(t)+\frac{\gamma_j(\varphi^{F_j}(t)-1)}{A_j+\eta_j}
\big[(1-\frac{k_j\rho\sigma_i}{\sigma_j})(\bar{c}-a_j)-(\theta_ja_j-k_j\theta_ia_i)\big]\nonumber\\
&
+\frac{(k_j\gamma_j\sigma_i)^2(1-\rho^2)}{4(A_j+\eta_j)}[(\varphi^{F_j}(t))^2-1]
-\frac{1}{2(\sigma_j)^2}(\bar{c}-a_j)^2(T-t).
\end{align}

\item {\bfseries {Subcase (Lb2)}}
If $a_j+\gamma_j\sigma_j^2\varphi^{F_j}(t)(1-\frac{k_j\rho\sigma_i}{\sigma_j})
M^{F_j}(t)\leq c_F$,
we have
\begin{align}
p^{\ast}(t)=c_F,\quad
q_i^{\ast}(t,p^{\ast}(t))=1,\quad
q_j^{\ast}(t,p^{\ast}(t))=N^{cF_j}(t)+\frac{k_j\rho \sigma_i}{\sigma_j}.
\end{align}

The preconditions become
$K[N^{cF_i}(t)+\frac{k_i\rho\sigma_j}{\sigma_i}N^{cF_j}(t)]\geq 1$ and
$N^{cF_j}(t)+\frac{k_j\rho\sigma_i}{\sigma_j}<1$.
Then, substituting $p^{\ast}(t)$ into \eqref{equ:LHJBb} and integrating from $T$ to $t$ gives
\begin{align}\label{equ:gLjb2}
g^L_2(t)=g^{Ljb2}_2(t)
\doteq&g(t)
+\frac{\gamma_L^2(\sigma_j-k_j\rho\sigma_i)^2}{4(A_L+\eta_L)}[(\varphi^L(t))^2-1]
+(c_F-a_j)\gamma_L(1-\frac{k_j\rho\sigma_i}{\sigma_j})
\big[\int_T^t\varphi^L(s)ds\nonumber\\
&
+\frac{\gamma_L}{\gamma_j}\int_T^t\frac{(\varphi^L(s))^2}{\varphi^{F_j}(s)}ds\big]
-(c_F-a_j)^2\frac{\gamma_L}{\gamma_j\sigma_j^2}
\big[\int_T^t\frac{\varphi^L(s)}{\varphi^{F_j}(s)}ds
+\frac{\gamma_L}{2\gamma_j}\int_T^t(\frac{\varphi^L(s)}{\varphi^{F_j}(s)})^2ds\big].
\end{align}
Equations \eqref{equ:g2Fib} and \eqref{equ:g2Fjb} become
\begin{align}\label{equ:g2Fib2}
g^{F_i}_2(t)=g^{F_ib2}_2(t)
\doteq &
g(t)+\frac{\gamma_i(\varphi^{F_i}(t)-1)}{A_i+\eta_i}
\big[k_i(c_F-a_j)(-1+\frac{k_j\rho \sigma_i}{\sigma_j}
+\frac{k_1k_2\rho\gamma_i\sigma_i}{\gamma_j\sigma_j}
-\frac{\rho\sigma_i\gamma_i}{\gamma_j\sigma_j})-(\theta_ia_i-k_i\theta_ja_j)\big]\nonumber\\
&
+\frac{\gamma_i^2\sigma_i^2[1+(k_1k_2\rho)^2
-2k_1k_2\rho^2]}{4(A_i+\eta_i)}[(\varphi^{F_i}(t))^2-1]
+\frac{k_i\gamma_i}{\gamma_j\sigma_j^2}
[1+\frac{k_i\gamma_i}{2\gamma_j}](c_F-a_j)^2(T-t),
\end{align}
and
\begin{align}\label{equ:g2Fjb2}
g^{F_j}_2(t)=g^{\tilde{F}_jb2}_2(t)
\doteq &
g(t)+\frac{\gamma_j(\varphi^{F_j}(t)-1)}{A_j+\eta_j}
\big[(1-\frac{k_j\rho\sigma_i}{\sigma_j})(c_F-a_j)-(\theta_ja_j-k_j\theta_ia_i)\big]\nonumber\\
&+\frac{(k_j\gamma_j\sigma_i)^2(1-\rho^2)}{4(A_j+\eta_j)}[(\varphi^{F_j}(t))^2-1]
-\frac{1}{2(\sigma_j)^2}(c_F-a_j)^2(T-t).
\end{align}

\item {\bfseries {Subcase (Lb3)}} If $c_F<a_j+\gamma_j\sigma_j^2\varphi^{F_j}(t)(1-\frac{k_j\rho\sigma_i}{\sigma_j})
M^{F_j}(t)<\bar{c}$, we have
\begin{align}
&p^{\ast}(t)=a_j+\gamma_j\sigma_j^2\varphi^{F_j}(t)(1-\frac{k_j\rho\sigma_i}
{\sigma_j})M^{F_j}(t),\nonumber\\
&q_i^{\ast}(t,p^{\ast}(t))=1,\quad
q_j^{\ast}(t,p^{\ast}(t))=(1-\frac{k_j\rho\sigma_i}{\sigma_j})M^{F_j}(t)
+\frac{k_j\rho \sigma_i}{\sigma_j}.
\end{align}
The preconditions become
$K[N^{aF_i}(t)+K^{\tilde{F}_i}M^{F_j}(t)] \geq 1$ and
$(1-\frac{k_j\rho\sigma_i}{\sigma_j})M^{F_j}(t)
+\frac{k_j\rho \sigma_i}{\sigma_j}<1$.
Substituting $p^{\ast}(t)$ into \eqref{equ:LHJBb} and integrating from $T$ to $t$ gives
\begin{align}\label{equ:gLjb3}
g^L_2(t)=g^{Ljb3}_2(t)
\doteq  & g(t)+\frac{\gamma_L^2(\sigma_j-k_j\rho\sigma_i)^2}{4(A_L+\eta_L)}[(\varphi^L(t))^2-1]\nonumber\\
&+\frac{1}{2}\gamma_L(\sigma_j-k_j\rho\sigma_i)^2
\int_T^t\varphi^L(s)(\gamma_j\varphi^{F_j}(s)+\gamma_L\varphi^L(s))M^{F_j}(s)ds.
\end{align}
Equations \eqref{equ:g2Fib} and \eqref{equ:g2Fjb} become
\begin{align}\label{equ:g2Fib3}
g^{F_i}_2(t)=g^{F_ib3}_2(t)
&\doteq
g(t)-\frac{\gamma_i(\theta_ia_i-k_i\theta_ja_j)}{A_i+\eta_i}[\varphi^{F_i}(t)-1]
+\frac{\gamma_i^2\sigma_i^2[1+(k_1k_2\rho)^2
-2k_1k_2\rho^2]}{4(A_i+\eta_i)}[(\varphi^{F_i}(t))^2-1]
\nonumber\\
&
-k_i\gamma_i(\sigma_j-k_j\rho\sigma_i)[-\gamma_j\sigma_j+k_j\rho \gamma_j\sigma_i
+k_1k_2\rho\gamma_i\sigma_i
-\rho\sigma_i\gamma_i]\int_T^t\varphi^{F_i}(s)\varphi^{F_j}(s)M^{F_j}(s)ds\nonumber\\
&
-\frac{1}{2}k_i\gamma_i(2\gamma_j+k_i\gamma_i)(\sigma_j-k_j\rho\sigma_i)^2
\int_T^t(\varphi^{F_j}(s)M^{F_j}(s))^2ds,
\end{align}
and
\begin{align}\label{equ:g2Fjb3}
g^{F_j}_2(t)=g^{\tilde{F}_jb3}_2(t)
\doteq &
g(t)-\frac{\gamma_j(\theta_ja_j-k_j\theta_ia_i)}{A_j+\eta_j}[\varphi^{F_j}(t)-1]+\frac{(k_j\gamma_j\sigma_i)^2(1-\rho^2)}{4(A_j+\eta_j)}[(\varphi^{F_j}(t))^2-1]
\nonumber\\
&-[\gamma_j(\sigma_j-k_j\rho\sigma_i)]^2
\int_T^t(\varphi^{F_j}(s))^2M^{F_j}(s)\frac{[3\gamma_j\varphi^{F_j}(s)+\gamma_L\varphi^L(s)]}
{2[2\gamma_j\varphi^{F_j}(s)+\gamma_L\varphi^L(s)]}ds.
\end{align}
\end{itemize}

\item {\bfseries {Case (Lc)}} If $\tilde{q}_i(t,p(t))<1$, $i=1,2$, we have $q_i^{\ast}(t,p(t))=\tilde{q}_i(t,p(t))$.
Substituting $q_i^{\ast}(t,p(t))$ into \eqref{equ:LHJB11} and simplifying gives
\begin{align}\label{equ:LHJBc}
0=&\sup_{p(\cdot)\in[c_F,\bar{c}]}\big\{\frac{dg^L_2(t)}{dt}+\beta(2\beta+1)\sigma^2g^L_1(t)
+\frac{1}{2}(\gamma_L\varphi^{L}(t))^2[(\sigma_1)^2+(\sigma_2)^2+2\sigma_1\sigma_2\rho]\nonumber\\
&+K\gamma_L
\big[-\frac{\varphi^{L}(t)D^{F_1}(t)}{\gamma_1\varphi^{F_1}(t)}(p(t)-a_1)
-\frac{\varphi^{L}(t)D^{F_2}(t)}{\gamma_2\varphi^{F_2}(t)}(p(t)-a_2)
+\frac{\varphi^{L}(t)D^{\tilde{F}_1}(t)}{\gamma_1\sigma_1^2\varphi^{F_1}(t)}(p(t)-a_1)^2\nonumber\\
&
+\frac{\varphi^{L}(t)D^{\tilde{F}_2}(t)}{\gamma_2\sigma_2^2\varphi^{F_2}(t)}(p(t)-a_2)^2
+\frac{\rho\varphi^{L}(t)D^{F_{12}}(t)}{\gamma_1\gamma_2
\sigma_1\sigma_2\varphi^{F_1}(t)\varphi^{F_2}(t)}(p(t)-a_1)(p(t)-a_2)\big]\big\},
\end{align}
where
\begin{align}
D^{F_i}(t)=&\frac{1}{K}\gamma_i\varphi^{F_i}(t)
+\gamma_L\varphi^{L}(t)[1+k_j\rho^2+\frac{\sigma_j\rho}{\sigma_i}(1+k_j)], \quad i\neq j\in\{1,2\},\nonumber\\
D^{\tilde{F}_i}(t)=&1
+\frac{K(1+(k_j\rho)^2+2k_j\rho^2)\gamma_L\varphi^{L}(t)}
{2\gamma_i\varphi^{F_i}(t)}, \quad i\neq j\in\{1,2\},\nonumber\\
D^{F_{12}}(t)=&k_1\gamma_1\varphi^{F_1}(t)+k_2\gamma_2\varphi^{F_2}(t)
+K(1+k_1+k_2+k_1k_2\rho^2)\gamma_L\varphi^{L}(t).
\end{align}

The first-order condition about $p(t)$ for maximizing the value in equation \eqref{equ:LHJBc} gives that
\begin{align}
p^{\ast}(t)=&\frac{P^N(t)}{P^D(t)}\vee c_F \wedge \bar{c},
\end{align}
where
\begin{align}
P^N(t)=&(\sigma_1\sigma_2)^2
[\gamma_2\varphi^{F_2}(t)D^{F_1}(t)+\gamma_1\varphi^{F_1}(t)D^{F_2}(t)]
+2a_1\sigma_2^2\gamma_2\varphi^{F_2}(t)D^{\tilde{F}_1}(t)\nonumber\\
&+2a_2\sigma_1^2\gamma_1\varphi^{F_1}(t)D^{\tilde{F}_2}(t)
+(a_1+a_2)\rho\sigma_1\sigma_2D^{F_{12}}(t),\label{equ:PNt}\\
P^D(t)=&2\sigma_2^2\gamma_2\varphi^{F_2}(t)D^{\tilde{F}_1}(t)
+2\sigma_1^2\gamma_1\varphi^{F_1}(t)D^{\tilde{F}_2}(t)+2\rho\sigma_1\sigma_2D^{F_{12}}(t)\label{equ:PDt}.
\end{align}

\begin{itemize}
\item {\bfseries {Subcase (Lc1)}} If $\frac{P^N(t)}{P^D(t)}\geq \bar{c}$, then
\begin{align}
p^{\ast}(t)=\bar{c},\quad
q_i^{\ast}(t,p^{\ast}(t))=K[N^{\bar{c}F_i}(t)+\frac{k_i\rho\sigma_j}{\sigma_i}N^{\bar{c}F_j}(t)].
\end{align}

The precondition becomes
$K[N^{\bar{c}F_i}(t)+\frac{k_i\rho\sigma_j}{\sigma_i}N^{\bar{c}F_j}(t)]<1$, $i=1,2$.
Substituting $p^{\ast}(t)$ into the equation \eqref{equ:LHJBc} and integrating from $T$ to $t$ gives
\begin{align}\label{equ:gLc1}
g^L_2(t)=g^{Lc1}_2(t)
\doteq &
g(t)+\frac{\gamma_L^2[(\sigma_1)^2+(\sigma_2)^2+2\sigma_1\sigma_2\rho]}
{4(A_L+\eta_L)}[(\varphi^{L}(t))^2-1]+K\gamma_L\Big\{
\frac{(\bar{c}-a_1)}{\gamma_1}
\int_T^t\frac{\varphi^{L}(s)D^{F_1}(s)}{\varphi^{F_1}(s)}ds\nonumber\\
&
+\frac{(\bar{c}-a_2)}{\gamma_2}
\int_T^t\frac{\varphi^{L}(s)D^{F_2}(s)}{\varphi^{F_2}(s)}ds
-\frac{(\bar{c}-a_1)^2}{\gamma_1\sigma_1^2}
\int_T^t\frac{\varphi^{L}(s)D^{\tilde{F}_1}(s)}{\varphi^{F_1}(s)}ds
\nonumber\\
&-\frac{(\bar{c}-a_2)^2}{\gamma_2\sigma_2^2}
\int_T^t\frac{\varphi^{L}(s)D^{\tilde{F}_2}(s)}{\varphi^{F_2}(s)}ds
-\frac{\rho(\bar{c}-a_1)(\bar{c}-a_2)}{\gamma_1\gamma_2\sigma_1\sigma_2}
\int_T^t\frac{\varphi^{L}(s)D^{F_{12}}(s)}{\varphi^{F_1}(s)\varphi^{F_2}(s)}ds\Big\}.
\end{align}
Equation \eqref{equ:g2Fic} becomes
\begin{align}\label{equ:g2Fic1}
g^{F_i}_2(t)=g^{F_ic1}_2(t)
\doteq &
g(t)+\frac{\gamma_i}{A_i+\eta_i}
[\bar{c}-a_i-k_i(\bar{c}-a_j)-\theta_ia_i+k_i\theta_ja_j]
[\varphi^{F_i}(t)-1]\nonumber\\
&+\frac{K^2(T-t)}{2}\big\{
-\frac{[1-(k_1k_2\rho)^2]}{(\sigma_i)^2}(\bar{c}-a_i)^2
+\frac{k_i\gamma_i[2\gamma_j-2k_1k_2\gamma_j\rho^2+k_i\gamma_i(1-\rho^2)]}
{(\gamma_j\sigma_j)^2}(\bar{c}-a_j)^2\nonumber\\
&+\frac{2k_i\rho[-\gamma_i(1-k_1k_2)+k_j\gamma_j(1-k_1k_2\rho^2)]}
{\sigma_i\sigma_j\gamma_j}(\bar{c}-a_i)(\bar{c}-a_j)\big\}.
\end{align}

\item {\bfseries {Subcase (Lc2)}}  If $\frac{P^N(t)}{P^D(t)}\leq c_F$, then
\begin{align}
p^{\ast}(t)=c_F, \quad
q_i^{\ast}(t,p^{\ast}(t))=K[N^{cF_i}(t)+\frac{k_i\rho\sigma_j}{\sigma_i}N^{cF_j}(t)].
\end{align}
The precondition becomes
$K[N^{cF_i}(t)+\frac{k_i\rho\sigma_j}{\sigma_i}N^{cF_j}(t)]<1$, $i=1,2$.
Substituting $p^{\ast}(t)$ into the equation\eqref{equ:LHJBc} and integrating from $T$ to $t$ gives
\begin{align}\label{equ:gLc2}
g^L_2(t)=g^{Lc2}_2(t)
\doteq & g(t)+\frac{\gamma_L^2[(\sigma_1)^2+(\sigma_2)^2+2\sigma_1\sigma_2\rho]}
{4(A_L+\eta_L)}[(\varphi^{L}(t))^2-1]+K\gamma_L\Big\{\frac{(c_F-a_1)}{\gamma_1}
\int_T^t\frac{\varphi^{L}(s)D^{F_1}(s)}{\varphi^{F_1}(s)}ds\nonumber\\
&
+\frac{(c_F-a_2)}{\gamma_2}
\int_T^t\frac{\varphi^{L}(s)D^{F_2}(s)}{\varphi^{F_2}(s)}ds-
\frac{(c_F-a_1)^2}{\gamma_1\sigma_1^2}
\int_T^t\frac{\varphi^{L}(s)D^{\tilde{F}_1}(s)}{\varphi^{F_1}(s)}ds\nonumber\\
&
-\frac{(c_F-a_2)^2}{\gamma_2\sigma_2^2}
\int_T^t\frac{\varphi^{L}(s)D^{\tilde{F}_2}(s)}{\varphi^{F_2}(s)}ds
-\frac{\rho(c_F-a_1)(c_F-a_2)}{\gamma_1\gamma_2\sigma_1\sigma_2}
\int_T^t
\frac{\varphi^{L}(s)D^{F_{12}}(s)}{\varphi^{F_1}(s)\varphi^{F_2}(s)}ds\Big\}.
\end{align}
Equation \eqref{equ:g2Fic} becomes
\begin{align}\label{equ:g2Fic2}
g^{F_i}_2(t)=g^{F_ic2}_2(t)
&\doteq
g(t)+\frac{\gamma_i}{A_i+\eta_i}
[c_F-a_i-k_i(c_F-a_j)-\theta_ia_i+k_i\theta_ja_j]
[\varphi^{F_i}(t)-1]\nonumber\\
&+\frac{K^2(T-t)}{2}\big\{
-\frac{[1-(k_1k_2\rho)^2]}{(\sigma_i)^2}(c_F-a_i)^2
+\frac{k_i\gamma_i[2\gamma_j-2k_1k_2\gamma_j\rho^2+k_i\gamma_i(1-\rho^2)]}
{(\gamma_j\sigma_j)^2}(c_F-a_j)^2\nonumber\\
&+\frac{2k_i\rho[-\gamma_i(1-k_1k_2)+k_j\gamma_j(1-k_1k_2\rho^2)]}
{\sigma_i\sigma_j\gamma_j}(c_F-a_i)(c_F-a_j)\big\}.
\end{align}

\item {\bfseries {Subcase (Lc3)}}
 If $c_F<\frac{P^N(t)}{P^D(t)}<\bar{c}$, then
\begin{align}
p^{\ast}(t)=\frac{P^N(t)}{P^D(t)},\quad
q_i^{\ast}(t,p^{\ast}(t))=K
\Big[\frac{\frac{P^N(t)}{P^D(t)}-a_i}{\gamma_i\sigma_i^2\varphi^{F_i}(t)}
+\frac{k_i\rho(\frac{P^N(t)}{P^D(t)}-a_j)}{\gamma_j\sigma_j\sigma_i\varphi^{F_j}(t)}\Big].\label{equ:pq*}
\end{align}
The precondition becomes
$K[\frac{\frac{P^N(t)}{P^D(t)}-a_i}{\gamma_i\sigma_i^2\varphi^{F_i}(t)}
+\frac{k_i\rho(\frac{P^N(t)}{P^D(t)}-a_j)}{\gamma_j\sigma_j\sigma_i\varphi^{F_j}(t)}]<1$, $i=1,2$.
Substituting $p^{\ast}(t)$ into the equation\eqref{equ:LHJBc} and integrating from $T$ to $t$ gives
\begin{align}\label{equ:gLc3}
&g^L_2(t)=g^{Lc3}_2(t)\nonumber\\
&\doteq g(t)+\frac{\gamma_L^2[(\sigma_1)^2+(\sigma_2)^2+2\sigma_1\sigma_2\rho]}{4(A_L+\eta_L)}
[(\varphi^{L}(t))^2-1]+K\gamma_L\big\{
\int_T^t\frac{\varphi^{L}(s)D^{F_1}(s)}{\gamma_1\varphi^{F_1}(s)}(\frac{P^N(s)}{P^D(s)}-a_1)ds\nonumber\\
&
+\int_T^t\frac{\varphi^{L}(s)D^{F_2}(s)}{\gamma_2\varphi^{F_2}(s)}(\frac{P^N(s)}{P^D(s)}-a_2)ds
-\int_T^t\frac{\varphi^{L}(s)D^{\tilde{F}_1}(s)}{\gamma_1\sigma_1^2\varphi^{F_1}(s)}
(\frac{P^N(s)}{P^D(s)}-a_1)^2ds
\nonumber\\
&-\int_T^t\frac{\varphi^{L}(s)D^{\tilde{F}_2}(s)}{\gamma_2\sigma_2^2\varphi^{F_2}(s)}
(\frac{P^N(s)}{P^D(s)}-a_2)^2ds
-\int_T^t
\frac{\rho\varphi^{L}(s)D^{F_{12}}(s)}{\gamma_1\gamma_2\sigma_1\sigma_2\varphi^{F_1}(s)\varphi^{F_2}(s)}
(\frac{P^N(s)}{P^D(s)}-a_1)(\frac{P^N(s)}{P^D(s)}-a_2)ds\big\}.
\end{align}
Equation \eqref{equ:g2Fic} becomes
\begin{align}\label{equ:g2Fic3}
&g^{F_i}_2(t)=g^{F_ic3}_2(t)\nonumber\\
&\doteq
g(t)-\frac{\gamma_i
(\theta_ia_i-k_i\theta_ja_j)}{A_i+\eta_i}[\varphi^{F_i}(t)-1]
-\gamma_i\int_T^t\varphi^{F_i}(s)(\frac{P^N(s)}{P^D(s)}-a_i)ds
+k_i\gamma_i\int_T^t\varphi^{F_i}(s)(\frac{P^N(s)}{P^D(s)}-a_j)ds\nonumber\\
&+\frac{K^2}{2}\big\{\frac{1-(k_1k_2\rho)^2}{(\sigma_i)^2}
\int_T^t(\frac{P^N(s)}{P^D(s)}-a_i)^2ds
-\frac{k_i\gamma_i[2\gamma_j-2k_1k_2\gamma_j\rho^2+k_i\gamma_i(1-\rho^2)]}
{(\gamma_j\sigma_j)^2}
\int_T^t(\frac{P^N(s)}{P^D(s)}-a_j)^2ds\nonumber\\
&-\frac{2k_i\rho[-\gamma_i(1-k_1k_2)+k_j\gamma_j(1-k_1k_2\rho^2)]}
{\sigma_i\sigma_j\gamma_j}]\int_T^t(\frac{P^N(s)}{P^D(s)}-a_i)
(\frac{P^N(s)}{P^D(s)}-a_j)ds\big\}.
\end{align}
\end{itemize}
\end{itemize}

Summing up the above processes, we can get Theorem \ref{Theorem1}.
\end{proof}

\renewcommand{\theequation}{B.\arabic{equation}}
\section{Proof of Corollary \ref{corollary2}}\label{Appendix B}

\begin{proof}
The conclusions in Table \ref{TablebLbF} can be obtained by taking partial derivatives of $b_L^{\ast}(t)$ and $b_i^{\ast}(t)$ with corresponding variables, respectively.
From $A_L=r_0-B_L-C_L$, $C_L=\eta_Le^{-\alpha_Lh_L}$, $B_Le^{-\alpha_Lh_L}=(\alpha_L+A_L+\eta_L)C_L$, we can get that
\begin{align*}
A_L=\frac{1}{1+\eta_L}[r_0-(\alpha_L+\eta_L)\eta_L-\eta_Le^{-\alpha_Lh_L}].
\end{align*}
Put $A_L$ into the equation \eqref{equ:bL*}, and take the derivative with respect to $\eta_L$ and $\alpha_L$ respectively. We can get that
\begin{align*}
\frac{\partial b_L^{\ast}(t)}{\partial\eta_L}=b_L^{\ast}(t)(T-t)\frac{1}{(1+\eta_L)^2}[r_0+\alpha_L+e^{-\alpha_Lh_L}-1],\quad
\frac{\partial b_L^{\ast}(t)}{\partial\alpha_L}=b_L^{\ast}(t)(T-t)\frac{\eta_L}{1+\eta_L}[1-h_Le^{-\alpha_Lh_L}].
\end{align*}
$b_L^{\ast}(t)>0$ because of $\beta\geq 0$.
Thus, the left side of equation \eqref{equ:balpha} and equation \eqref{equ:beta} is established. Following similar derivations, we can get that the right side of equation \eqref{equ:balpha} and equation \eqref{equ:beta} is true.
\end{proof}

\renewcommand{\theequation}{C.\arabic{equation}}
\section{Proof of Lemma \ref{lemma1}}\label{Appendix C}

\begin{proof}
According to \eqref{equ:St} and \eqref{equ:X^F}, we apply It\^{o}'s formula
\begin{align*}
d(V^{F_i}(t))^2
&=2V^{F_i}(t)\big\{\mathcal{A}^{{F_i}}V^{F_i}(t,\hat{x}_i,y_i,y_j,s)|_{(q_1^{\ast}(\cdot),
b_1^{\ast}(\cdot),q_2^{\ast}(\cdot),b_2^{\ast}(\cdot))}dt
+V^{F_i}_{\hat{x}_i}[q_i^{\ast}(t)\sigma_idW_i(t)-k_iq_j^{\ast}(t)\sigma_jdW_j(t)]\\
&
+[V^{F_i}_{\hat{x}_i}(b_i^{\ast}(t)-k_ib_j^{\ast}(t))+sV^{F_i}_{s}]\sigma s^{\beta} dW(t)\big\}+\big\{(V^{F_i}_{\hat{x}_i})^2\big[(q_i^{\ast}(t)\sigma_i)^2+(k_iq_j^{\ast}(t)\sigma_j)^2
+(b_i^{\ast}(t)-k_ib_j^{\ast}(t))^2\sigma^2s^{2\beta}\\
&
-2q_i^{\ast}(t)\sigma_ik_iq_j^{\ast}(t)\sigma_j\rho\big]+\sigma^2s^{2\beta+2}(V^{F_i}_{s})^2
+2V^{F_i}_{\hat{x}_i}V^{F_i}_{s}(b_i^{\ast}(t)-k_ib_j^{\ast}(t))\sigma^2 s^{2\beta+1}\big\}dt.
\end{align*}
From Section \ref{subsection3.1}, we know that $\mathcal{A}^{{F_i}}V^{F_i}(t,\hat{x}_i,y_i,y_j,s)|_{(q_1^{\ast}(\cdot),
b_1^{\ast}(\cdot),q_2^{\ast}(\cdot),b_2^{\ast}(\cdot))}=0$. By plugging the expressions of $V^{F_i}_{\hat{x}_i}$, $V^{F_i}_{s}$,
$b_1^{\ast}(\cdot)$ and $b_2^{\ast}(\cdot)$ into the above equation, we obtain
\begin{align}\label{equ:dVFT}
\frac{d(V^{F_i}(t))^2}{(V^{F_i}(t))^2}
&=-2\{\gamma_i\varphi^{F_i}(t)[q_i^{\ast}(t)\sigma_idW_i(t)-k_iq_j^{\ast}(t)\sigma_jdW_j(t)]
+\frac{r-r_0}{\sigma s^{\beta}}dW(t)\}\nonumber\\
&~~+\{(\gamma_i\varphi^{F_i}(t))^2[(q_i^{\ast}(t)\sigma_i)^2+(k_iq_j^{\ast}(t)\sigma_j)^2
-2q_i^{\ast}(t)\sigma_ik_iq_j^{\ast}(t)\sigma_j\rho]
+(\frac{r-r_0}{\sigma s^{\beta}})^2\}dt\nonumber\\
&=\Theta_1^{F_i}(t)dt
+\Theta_2^{F_i}(t)dW_i(t)+\Theta_3^{F_i}(t)dW_j(t)-2\frac{r-r_0}{\sigma s^{\beta}}dW(t),
\end{align}
where $\Theta_1^{F_i}(t)=(\gamma_i\varphi^{F_i}(t))^2[(q_i^{\ast}(t)\sigma_i)^2+(k_iq_j^{\ast}(t)\sigma_j)^2
-2q_i^{\ast}(t)\sigma_ik_iq_j^{\ast}(t)\sigma_j\rho]
+(\frac{r-r_0}{\sigma s^{\beta}})^2$,
$\Theta_2^{F_i}(t)=-2\gamma_i\varphi^{F_i}(t)q_i^{\ast}(t)\sigma_i$, $\Theta_3^{F_i}(t)=2\gamma_i\varphi^{F_i}(t)k_iq_j^{\ast}(t)\sigma_j$.
The forms of $q_i^{\ast}(t)$ and $q_j^{\ast}(t)$ in different cases are given by Theorem \ref{Theorem1}.
Thus, the solution to the equation \eqref{equ:dVFT} is
\begin{align}\label{equ:VFT}
\frac{(V^{F_i}(t))^2}{(V^{F_i}(0))^2}
=&\exp\big\{\int_0^t\Theta_1^{F_i}(\iota)d\iota\big\}
+\exp\big\{\int_0^t[
-\frac{1}{2}\Theta_2^{F_i}(\iota)^2
-\frac{1}{2}\Theta_3^{F_i}(\iota)^2]d\iota
+\int_0^t\Theta_2^{F_i}(\iota)dW_i(\iota)
+\int_0^t\Theta_3^{F_i}(\iota)dW_j(\iota)\big\}\nonumber\\
&+\exp\big\{-\frac{1}{2}\int_0^t\frac{4(r-r_0)^2}{\sigma^2}(S(\iota))^{-2\beta}d\iota
-\int_0^t\frac{2(r-r_0)}{\sigma}(S(\iota))^{-\beta}dW(\iota)\big\}.
\end{align}
By virtue of Novikov's condition, we know that $\exp\big\{-\frac{1}{2}\int_0^t\Theta_2^{F_i}(\iota)^2d\iota
+\int_0^t\Theta_2^{F_i}(\iota)dW_i(\iota)\big\}$ and
$\exp\big\{-\frac{1}{2}\int_0^t\Theta_3^{F_i}(\iota)^2d\iota
+\int_0^t\Theta_3^{F_i}(\iota)dW_j(\iota)\big\}$ are two martingales.
In addition, according to \eqref{equ:St} and It\^{o}'s formula, we can drive
\begin{align}
d(S(t))^{-2\beta}
=[\beta(2\beta+1)\sigma^2-2\beta r(S(t))^{-2\beta}]dt-2\beta\sigma(S(t))^{-\beta}dW(t).
\end{align}
By using Lemma 4.3 and Theorem 5.1 of \cite{Zeng-Taksar-2013}, we can verify that
$\exp\big\{-\frac{1}{2}\int_0^t\frac{4(r-r_0)^2}{\sigma^2}(S(\iota))^{-2\beta}d\iota
-\int_0^t\frac{2(r-r_0)}{\sigma}(S(\iota))^{-\beta}dW(\iota)$ is a martingale.
From the range of parameters and the form of $\varphi^{F_i}(t)$, we know that $\Theta_1^{F_i}(t)<\infty$.
Thus taking expectation from both sides of equation \eqref{equ:VFT}, we obtain
\begin{align*}
E[(V^{F_i}(t))^2]=(V^{F_i}(0))^2\exp\big\{\int_0^t\Theta_1^{F_i}(\iota)d\iota\big\}<+\infty.
\end{align*}
Therefore, for $n=1,2,\cdots$, we have
\begin{align*}
E_{t,\hat{x}_i,y_i,y_j,s}\big\{[V^{F_i}(\tau_n\wedge T,\hat{X}_i(\tau_n\wedge T),Y_i(\tau_n\wedge T),Y_j(\tau_n\wedge T),S(\tau_n\wedge T))]^2\big\}<+\infty.
\end{align*}
Then, the proof is complete.
\end{proof}

\renewcommand{\theequation}{D.\arabic{equation}}
\section{Proof of Theorem \ref{Theorem2}}\label{Appendix D}

\begin{proof}
Obviously, the pair $(p^{\ast}(t),b_L^{\ast}(t),q_1^{\ast}(t),b_1^{\ast}(t),q_2^{\ast}(t),b_2^{\ast}(t))$ obtained in Theorem \ref{Theorem1} is an admissible strategy, i.e., $(p^{\ast}(t),b_L^{\ast}(t),q_1^{\ast}(t),b_1^{\ast}(t),q_2^{\ast}(t),b_2^{\ast}(t))
\in\Pi_L\times\Pi_{1}\times\Pi_{2}$. Next, we show the optimality of $(p^{\ast}(t),b_L^{\ast}(t),q_1^{\ast}(t),b_1^{\ast}(t),\newline q_2^{\ast}(t),b_2^{\ast}(t))$  in $\Pi_L\times\Pi_{1}\times\Pi_{2}$.
From the construction of $\tau_n$, we know that $\tau_n\wedge T\rightarrow T$ when $n\rightarrow+\infty$.
For $\forall(p^{\ast}(t),b_L^{\ast}(t),q_1^{\ast}(t),b_1^{\ast}(t),q_2^{\ast}(t),b_2^{\ast}(t))
\in\Pi_L\times\Pi_{1}\times\Pi_{2}$ and
$\forall\iota\in[t,T]$, we apply It\^{o}'s formula to $V^{F_i}(t,\hat{x}_i,y_i,y_j,s)$ and deduce
\begin{align*}
&V^{F_i}(\tau_n\wedge T,\hat{X}_i(\tau_n\wedge T),Y_i(\tau_n\wedge T),Y_j(\tau_n\wedge T),S(\tau_n\wedge T))\\
&=V^{F_i}(t,\hat{x}_i,y_i,y_j,s)
+\int_t^{\tau_n\wedge T}\mathcal{A}^{F_i}V^{F_i}(\iota,\hat{X}_i(\iota),Y_i(\iota),Y_j(\iota),S(\iota))d\iota
+\int_t^{\tau_n\wedge T}V^{F_i}_{\hat{x}_i}(\iota,\hat{X}_i(\iota),Y_i(\iota),Y_j(\iota),S(\iota))
q_i(\iota)\sigma_idW_i(\iota)\\
&-\int_t^{\tau_n\wedge T}V^{F_i}_{\hat{x}_i}(\iota,\hat{X}_i(\iota),Y_i(\iota),Y_j(\iota),S(\iota))
k_iq_j(\iota)\sigma_jdW_j(\iota)+\int_t^{\tau_n\wedge T}V^{F_i}_{\hat{x}_i}(\iota,\hat{X}_i(\iota),Y_i(\iota),Y_j(\iota),S(\iota))
(b_i(\iota)\\
&-k_ib_j(\iota))\sigma(S(\iota))^\beta dW(\iota)+\int_t^{\tau_n\wedge T}V^{F_i}_{s}(\iota,\hat{X}_i(\iota),Y_i(\iota),Y_j(\iota),S(\iota))
\sigma(S(\iota))^{\beta+1} dW(\iota).
\end{align*}
Because the last four terms are square-integrable martingales with zero expectations, taking expectation on both sides of the above equation conditional on $\hat{X}_i(t)=\hat{x}_i$, $Y_i(t)=y_i$, $Y_j(t)=y_j$ and $S(t)=s$, we have
\begin{align*}
&E_{t,\hat{x}_i,y_i,y_j,s}[V^{F_i}(\tau_n\wedge T,\hat{X}_i(\tau_n\wedge T),Y_i(\tau_n\wedge T),Y_j(\tau_n\wedge T),S(\tau_n\wedge T))]\\
&=V^{F_i}(t,\hat{x}_i,y_i,y_j,s)
+E_{t,\hat{x}_i,y_i,y_j,s}\big[\int_t^{\tau_n\wedge T}\mathcal{A}^{F_i}V^{F_i}(\iota,
\hat{X}_i(\iota),Y_i(\iota),Y_j(\iota),S(\iota))d\iota\big]\leq V^{F_i}(t,\hat{x}_i,y_i,y_j,s).
\end{align*}
In terms of Lemma \ref{lemma1}, the uniform integrability of $V^{F_i}(\tau_n\wedge T,\hat{X}_i(\tau_n\wedge T),Y_i(\tau_n\wedge T),Y_j(\tau_n\wedge T),S(\tau_n\wedge T))$ yields
\begin{align*}
&\sup_{(q_i(\cdot),b_i(\cdot))\in\Pi_i}E_{t,\hat{x}_i,y_i,y_j,s}
\big[U_i\big(\hat{X}_i^{\pi_i}(T)+\eta_iY_i(T)-k_i\eta_jY_j(T)\big)\big]\\
&=\lim_{n\rightarrow+\infty}E_{t,\hat{x}_i,y_i,y_j,s}\big[V^{F_i}(\tau_n\wedge T,\hat{X}_i(\tau_n\wedge T),Y_i(\tau_n\wedge T),Y_j(\tau_n\wedge T),S(\tau_n\wedge T))\big]\leq V^{F_i}(t,\hat{x}_i,y_i,y_j,s).
\end{align*}
When $(q_i(\cdot),b_i(\cdot))=(q_i^{\ast}(\cdot),b_i^{\ast}(\cdot))$, the inequality in the above formula becomes an equality, and thus
\begin{align*}
\sup_{(q_i(\cdot),b_i(\cdot))\in\Pi_i}E_{t,\hat{x}_i,y_i,y_j,s}
\big[U_i\big(\hat{X}_i^{\pi_i}(T)+\eta_iY_i(T)-k_i\eta_jY_j(T)\big)\big]
=V^{F_i}(t,\hat{x}_i,y_i,y_j,s).
\end{align*}

Following similar derivations, we can obtain
\begin{align*}
\sup_{(p(\cdot),b_L(\cdot))\in\Pi_L}E_{t,x_L,y_L,s}
\left[U_L(X_L^{\pi_L}(T)+\eta_LY_L(T))\right]
\leq V^L(t,x_L,y_L,s).
\end{align*}
And when $(p(\cdot),b_L(\cdot))=(p^{\ast}(\cdot),b_L^{\ast}(\cdot))$, the above inequality becomes an equality. Then, the proof is complete.
\end{proof}

\end{document}